\documentclass{article}

\usepackage[T1]{fontenc}

\usepackage[a4paper, margin=1.2in]{geometry}
\usepackage{authblk}

\usepackage{amsfonts}
\usepackage{amsthm}
\usepackage{mathtools}

\usepackage{tikz}

\newcommand{\lBrace}{\lbrace\mkern-4mu |}
\newcommand{\rBrace}{|\mkern-4mu \rbrace}
\newcommand{\bs}{b}
\newcommand{\tworows}[2]{$\genfrac{}{}{0pt}{}{\mbox{#1}}{\mbox{#2}}$}

\newtheorem{example}{Example}
\newtheorem{theorem}{Theorem}
\newtheorem{lemma}[theorem]{Lemma}
\newtheorem{proposition}[theorem]{Proposition}

\usepackage{natbib}
\setcitestyle{authoryear,open={(},close={)}}

\begin{document}

\title{Axiomatic Characterization of PageRank}

\author{Tomasz W\k{a}s}

\author{Oskar Skibski%
\thanks{\texttt{\{t.was, o.skibski\}@mimuw.edu.pl} \\
This work is partially based on \citet{Was:Skibski:2018:pagerank} presented at the 27th International Joint Conference on Artificial Intelligence (IJCAI-18). Compared to the conference publication, the set of axioms has changed which led to a new proof of uniqueness. All other parts of the paper are also new, including the proof of independence, the comparison with other centrality measures and the extended related work. \\
This work was supported by the National Science Centre under Grant No. 2018/31/B/ST6/03201 and the Foundation for Polish Science under Grant Homing/2016-1/7.}
}
%

\affil{University of Warsaw}

\date{}
 
\maketitle

\begin{abstract}
This paper examines the fundamental problem of identifying the most important nodes in a network.
To date, more than a hundred centrality measures have been proposed, each evaluating the position of a node in a network from a different perspective. 
Our work focuses on PageRank which is one of the most important centrality measures in computer science used in a wide range of scientific applications.
To build a theoretical foundation for choosing (or rejecting) PageRank in a specific setting, we propose to use an axiomatic approach.
Specifically, we propose six simple properties and prove that PageRank is the only centrality measure that satisfies all of them.
In this way, we provide the first axiomatic characterization of PageRank in its general form.
\end{abstract}




\section{Introduction}\label{section:introduction}
The problem of identifying the most important elements in an interconnected system is the fundamental question in network analysis. 
While the first attempts to answer this question systematically date back to the 1940s, in the last few decades we have experienced a fast growth of the literature on the subject of centrality analysis.
As a result, the centrality analysis has played an important role in the rapid development of many areas, including computer science \cite{Page:etal:1999,Weng:etal:2010}, physics \cite{Hill:Braha:2010,Brandes:Fleischer:2005}, bioinformatics~\cite{Jeong:etal:2001,Ivan:Grolmusz:2010} and social sciences \cite{Borgatti:2006,Brandes:2001}.

To date, more than a hundred centrality measures have been proposed, each evaluating the position of a node in a network from a different perspective~\cite{Brandes:Erlebach:2005}.
PageRank, introduced 20 years ago by \citet{Page:etal:1999} to assess the importance of web pages, quickly became one of the most popular centrality measures. 
Not only it plays a fundamental role in the success of the Google search engine, but it is often used in a wide range of scientific applications.
Other AI applications include finding the most influential users in the Twitter social network \cite{Weng:etal:2010} and improving propagation scheme of the graph neural networks~\cite{Klicpera:etal:2018}.
In other fields, PageRank was used to identify cancer genes in proteomic data \cite{Ivan:Grolmusz:2010}, evaluate prestige of journals based on the citation network \cite{Bollen:etal:2006}, but also in more exotic applications such as determining the best player in the history of tennis~\cite{Radicchi:2011} and the key procedure in the Linux kernel~\cite{Chepelianskii:2010}.

Networks, however, differ, and what gives desirable results in one type of network, may turn out to be ill-fitting in another type. 
For instance, cellular terrorist networks require significantly different methods to identify key terrorists~\cite{Lindelauf:etal:2013} than air-traffic networks to evaluate hub airports~\cite{Guimera:etal:2005}.
Given this, and the large number of centrality measures proposed in the literature, there is no clear reason to prefer PageRank over other centrality measures.
To put it more generally, a choice of a centrality measure for a specific application based on its performance in some other settings, its intuitive interpretation or its popularity, cannot be considered a well-founded scientific approach.
What is needed is a rigorous analysis which would justify that a given centrality measure has desirable properties for a specific goal.

In this paper, we propose to use an axiomatic approach to build a theoretical foundation for choosing (or rejecting) PageRank in a specific setting.
More in detail, we propose six axioms based on simple graph operations.
Our main result is that PageRank is the only centrality measure that satisfies all of these axioms.
This result gives new theoretical foundations of PageRank and highlights the properties which uniquely characterize this measure.

Our axiomatization consists of six axioms: \emph{Node Deletion}, \emph{Edge Deletion}, \emph{Edge Multiplication}, \emph{Edge Swap}, \emph{Node Redirect} and \emph{Baseline}.
The first five axioms are so called \emph{invariance axioms}---each of them is named after a graph operation that, under some additional restrictions, does not affect centralities of some nodes in the graph.
For example, Edge Multiplication states that replacing every outgoing edge of a node by a fixed number of copies does not affect centrality of any other node in the network.
The sixth axiom, \emph{Baseline}, specifies centrality of a node in a simple borderline case.
We prove that if a centrality measure satisfies all five invariance axioms, then it is equal to PageRank up to a scalar multiplication. 
Furthermore, if the centrality measure additionally satisfies Baseline, then it is equal to PageRank.

To illustrate the advantages of our approach, consider a network of matches between tennis players.
Here, every node represents a player and every edge represents a single match. 
Specifically, an edge from player A to player B represents a match played by both players and won by player B.
The resulting network, similarly to the World Wide Web, may have multiple edges between two nodes.
Several authors have applied PageRank to this network with an aim to determine the best player in the history of tennis (see, e.g., \cite{Radicchi:2011}).
While the results may or may not be correct, quick check of axioms shows that PageRank is not a good choice of a centrality measure.
Specifically, by looking at Edge Multiplication again, we know that multiplying the number of lost matches of one player by a large number, say 1000, will not affect the value of PageRank of any player. 
If we do such a multiplication for all the lost matches of the top player, we end up with a paradoxical result that the player who now lost most often is still ranked first, as a champion of all times.

As the domain of our axiomatization we choose directed multigraphs, i.e., directed graphs in which multiple edges between the same nodes are allowed. 
Directed multigraphs are a natural model of the World Wide Web which is the original application of PageRank. 
Also, multigraphs are simpler and more illustrative than graphs with edge weights. We note, however, our result can be easily adapted to the setting with simple graphs with real positive edge weights.
Also, following the seminal paper on PageRank, we assume nodes have weights that corresponds to their baseline importance or ''source of rank''.

The problem of finding an axiomatic characterization of PageRank has already been raised in the literature.
So far, however, no axiomatization has been found.
In the most related work, \citet{Altman:Tennenholtz:2005} proposed an axiomatization of the ranking provided by Seeley index, which is a simplified version of PageRank well-defined only for strongly connected graphs. 
Another axiomatization of Seeley index was provided by \citet{Slutzki:Volij:2006}.
\citet{Palacios-Huerta:Volij:2004} focused on a specific setting of a journal citation network and proposed an axiomatization of the invariant method which is equal to Seeley index of a journal divided by the number of its articles.
As we discuss in Section~\ref{section:simplified_pr_axioms}, these axiomatizations cannot be extended to PageRank.


The paper is organized as follows.
In Section~\ref{section:related_work}, the related work is described.
In Section~\ref{section:preliminaries}, we introduce the necessary concepts from graph theory used throughout the paper.
In Section~\ref{section:axioms}, we present our axioms and formulate the main theorem of the paper that PageRank is the only centrality measure that satisfies all of them.
Section~\ref{section:uniqueness} is devoted to the proof of this main theorem.
In Section~\ref{section:comparison}, we consider other centrality measures from the literature and using our axioms show how they differ from PageRank.
In Section~\ref{section:discussion}, we discuss different variants of the definition of PageRank that appear in the literature (Section~\ref{section:definitions}) and we relate our axiomatization to the axiomatizations of a simplified version of PageRank (Section~\ref{section:simplified_pr_axioms}).
Conclusions follow.
The paper contains three appendices:
In \ref{section:appendix:proof-2} we prove that PageRank indeed satisfies our axioms. \ref{section:appendix:independence} contains the proof that our axioms are independent. \ref{section:appendix:other_centralities} presents the analysis of other centrality measures.

\section{Related work}\label{section:related_work}

In this section, we summarize the related literature and compare it to our work.
We focus on two areas of research: the analysis of the properties of PageRank and the axiomatic approach to centrality measures. 

\subsection{Properties of PageRank}

Since PageRank was proposed by \citet{Page:etal:1999} in 1999, many papers appeared that examine its properties.
In the beginning, a majority of them focused on the computational aspects: PageRank was originally created to serve as a centrality measure for networks with millions of nodes, so its quick computation was considered the key issue.
In their original paper, \citet{Page:etal:1999} proposed a simple Power Method which enables the approximate calculation of PageRank.
Since then, many techniques were proposed for faster and more accurate calculations, e.g.: Extrapolation Method~\cite{Kamvar:etal:2003:extrapolation}, Adaptive Method~\cite{Kamvar:etal:2004}, Block Structure Method~\cite{Kamvar:etal:2003:exploiting}, DAG Structure Method~\cite{Arasu:etal:2002}, methods derived from spectral graph theory~\cite{DelCorso:etal:2005}, and others~\cite{Lee:etal:2003,Brezinski:Redivo-Zaglia:2006}.
See a survey by \citet{Berkhin:2005} for a comprehensive comparison of these methods.

Another line of research closely related to our work is the sensitivity analysis of PageRank.
The sensitivity analysis studies what is the impact on PageRank of a small change in the decay factor or the network. 
In turn, our work identifies such graph modifications that do not change PageRank at all.
A number of papers examines the impact of changing the decay factor on the computation of PageRank and the centrality itself.
In the early work, \citet{Pretto:2002} showed that the decay factor affects the ranking of nodes resulting from PageRank.
\citet{Boldi:etal:2005} analyzed how does the measure behave when the decay factor is close to 1 and showed that this is not always desirable.
Another group of research in that area focuses on changes in the network structure: \citet{Kamvar:Haveliwala:2003} estimated the bounds for the impact of such changes and \citet{Chien:etal:2004} proved that adding an edge $(u,v)$ always has a positive impact on both PageRank of node $v$ and its position in the ranking of nodes based on PageRank.
See the work of \citet{Langville:Meyer:2004} for a survey on early results on this topic.

Finally, also similarly to our work, there are also plenty of results that focus on theoretical similarities and differences between PageRank and other centrality measures.
Notably, \citet{Franceschet:2011} observed the resemblance of PageRank with several other measures, including 
HITS~\cite{Kleinberg:1999} and SALSA~\cite{Lempel:Moran:2001} algorithms, the invariant method~\cite{Pinski:Narin:1976}, Katz centrality~\cite{Katz:1953} and Leontief input-output model~\cite{Leontief:1951}.
Similarly, \citet{Vigna:2016} summarized the history of measures that---like PageRank---applies the theory of linear maps to adjacency matrices and discussed how does PageRank fits in this wider context.
\citet{Baeza-Yates:etal:2006} defined a class of centrality measures based on a decay factor to which PageRank belongs and showed how these measures approximate each other.
In turn, \citet{Bianchini:etal:2005} analyzed how groups of pages can reorganize their connections in order to increase their PageRank and \citet{Avrachenkov:Litvak:2006} proved a bound for an increase in PageRank that a node can obtain by changes in its outgoing edges.

A comprehensive study and a survey of computational and theoretical properties of PageRank can be found in a book by \citet{Langville:Meyer:2011}.

\subsection{Axiomatic analysis of centrality measures}

The axiomatic analysis of centrality measures began with the seminal work of \citet{Sabidussi:1966}.
Sabidussi proposed a number of simple axioms and argued that they should be satisfied by all reasonable centrality measures.
Most of the centrality measures used nowadays violate at least one of these axioms.
\citet{Nieminen:1973} proposed a similar axiomatization for directed graphs that precludes centrality measures which, unlike closeness and decay centralities, are not based on distances.
More recently, \citet{Boldi:Vigna:2014}, in a similar approach, introduced three new axioms and analyzed the satisfiability of these axioms.
The authors showed that out of popular centrality measures only the harmonic centrality, an alternative to closeness centrality, satisfies all of them.
More recently, \citet{Boldi:etal:2017} analyzed centrality measures from the perspective of one specific axiom called \emph{rank monotonicity} and proved that PageRank satisfies it.

Another approach, presented also in this paper, is to develop a set of axioms and show that they uniquely characterize some centrality measure, i.e., to show that they can be satisfied only by one centrality measure.
To this day, such axiomatic characterizations have been proposed for several centrality measures.

Arguably due to their complex nature, a lot of effort was put to axiomatize feedback centralities.
The first such axiomatization was proposed by \citet{Brink:Gilles:2000} for beta measure in directed graphs and later on extended by \citet{Brink:etal:2008} to undirected graphs.
Seeley index, which is a simplified version of PageRank, was considered by \citet{Altman:Tennenholtz:2005} and \citet{Palacios-Huerta:Volij:2004}.
In Section~\ref{section:simplified_pr_axioms}, we discuss these two axiomatizations in more detail, as they are most closely related to our work.
\citet{Dequiedt:Zenou:2017} created axiomatizations of eigenvector and Katz centralities in undirected graphs by studying graphs in which some nodes have a fixed centrality.
\citet{Was:Skibski:2018:eigenvector} also provided a joint axiomatization of eigenvector and Katz centralities, but in directed graphs.
\citet{Kitti:2016} created yet another axiomatic characterization of eigenvector centrality where the axioms are defined as the algebraic properties of the adjacency matrix.

The class of centrality measures based on distances was studied by \citet{Garg:2009} and later on by \citet{Skibski:Sosnowska:2018:distance}.
In the later paper, the first axiomatization of closeness centrality was proposed.
Decay centrality has been extended to a random-walk version using the axiomatic approach by \citet{Was:etal:2019:rwd}.

More recently, game-theoretic centrality measures based on the coalitional game theory have gain popularity in the literature~\cite{Lindelauf:etal:2013,Amer:Gimenez:2004,Michalak:etal:2015:defeating}. 
Several axiomatic results have been obtained.
\citet{Skibski:etal:2018:gtc} created an axiomatic characterization of the whole class of game-theoretic centralities.
In turn, \citet{Skibski:etal:2019:attachment} proposed the first axiomatized game-theoretic centrality measure under the name \emph{attachment centrality}.

In a more general approach, \citet{Bloch:etal:2019} showed that many popular centrality measures can be axiomatized with the same set of meta-axioms parametrized with different node statistics.

In a number of the listed papers, including \cite{Brink:Gilles:2000,Skibski:etal:2019:attachment,Garg:2009}, degree centrality was also axiomatized as the borderline case of more complex characterization scheme. 
Notably, to date, no axiomatization of betweenness centrality has been proposed.

\section{Preliminaries}\label{section:preliminaries}
In this section, we introduce the necessary concepts and notations pertaining to graph theory that will be used throughout the paper. 
Since the popularity of PageRank was initiated by its application to the World Wide Web, we will employ this domain as the background of our narrative.

The basic element of the World Wide Web is a \emph{webpage}, or simply a \emph{page}.
Pages are connected with each other through \emph{links}---each page contains links (sometimes called \emph{hyperlinks}) to other pages.
A link for the page it refers to is called a \emph{backlink}.
The same link may appear on a page multiple times.
Also, a page may contain a link to itself. 

In graph theory, a suitable model of the World Wide Web is a directed multigraph.
A \emph{multigraph} is a pair, $G=(V,E)$, where $V$ is a set of nodes that represent pages and $E$ is a multiset of edges that represent links.
Specifically, an edge is an ordered pair of nodes, i.e., $(u,v)$ for some nodes $u,v$, where $u$ is the start of an edge and $v$ is its end.
Consequently, $(u,v)$ is an \emph{outgoing edge} (i.e., a link) for node $u$ and an \emph{incoming edge} (i.e., a backlink) for node $v$.
In multigraphs, one edge may appear several times; hence, $E$ is formally a multiset. 
To emphasize this fact, we will list elements of $E$ using the double brackets $\lBrace \dots \rBrace$ and for two multisets $E$, $E'$, we will denote their sum and difference by $E \sqcup E'$ and $E - E'$, respectively.
From now on, we will refer to directed multigraphs simply as \emph{graphs}.
See Fig.~\ref{figure:main} for illustration.

\begin{figure}[t]
\centering
\begin{tikzpicture}
  \def\x{0.8cm} 
  \def\y{0cm} 
  \def\arrdist{0.3cm}

  \tikzset{
    node_blank/.style={circle,draw,minimum size=0.6cm,inner sep=0, color=white}, 
    node/.style={circle,draw,minimum size=0.6cm,inner sep=0, fill = black!05}, 
    edge/.style={sloped,-latex,above,font=\footnotesize}, 
    el/.style={below,font=\footnotesize}, 
    operation/.style={sloped,>=stealth,above,font=\footnotesize},
    arrow/.style={draw, single arrow, minimum width = 0.9cm, minimum height=\y-6*\x+\s, fill=black!10},
    blank/.style={}
  } 
  \node[node] (A_1) at (\y+2*\x, 4*\x) {$v_1$};
  \node[node] (A_2) at (\y+4*\x, 3*\x) {$v_2$}; 
  \node[node] (A_3) at (\y+4*\x, 1*\x) {$v_3$}; 
  \node[node] (A_4) at (\y+2*\x, 0*\x) {$v_4$};
  \node[node] (A_5) at (\y+0*\x, 1*\x) {$v_5$}; 
  \node[node] (A_6) at (\y+0*\x, 3*\x) {$v_6$}; 
  \node[node] (A_7) at (\y+1.25*\x, 2*\x) {$v_7$}; 
  \node[node] (A_8) at (\y+2.75*\x, 2*\x) {$v_8$};
  \node[node_blank] (A_start) at (\y - \arrdist, 2*\x) {};
  \node[node_blank] (A_end) at (\y+4*\x + \arrdist, 2*\x) {};

  \path[->,draw,thick]
  (A_1) edge[edge]  (A_8)
  (A_2) edge[edge]  (A_3)
  (A_4) edge[edge]  (A_8)
  (A_5) edge[edge, bend right=15]  (A_4)
  (A_5) edge[edge, bend right=15]  (A_6)
  (A_5) edge[edge]  (A_7)
  (A_6) edge[edge, bend right=15]  (A_5)
  (A_6) edge[edge, bend left=15]  (A_1)
  (A_6) edge[edge]  (A_7)
  (A_7) edge[edge]  (A_1)
  (A_8) edge[edge, bend left = 15]  (A_2)
  (A_8) edge[edge, bend right = 15]  (A_2)
  (A_8) edge[edge]  (A_7)
  ;
\end{tikzpicture}
\caption{A graph---a theoretical model of the structure of the World Wide Web. Circles represent nodes, i.e., pages on the World Wide Web, and arrows represent directed edges, i.e., links.
Specifically, an arrow from one node to another represents a link on the former page that refers to the later one. 
Note that one node can have several edges to another node, as one page can have many links to another page.}
\label{figure:main}
\end{figure}
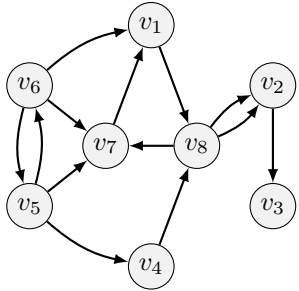

Let us introduce the notation regarding edges \emph{incident} to node $v$, i.e., its outgoing and incoming edges.
The multiset of outgoing edges of node $v$ is denoted by $\Gamma^+_v(G)$ and incoming edges by $\Gamma^-_v(G)$. 
Two nodes $u,v$ are called \emph{out-twins} when the multisets of ends of their outgoing edges are equal: $\Gamma^+_v(G) = \{(v,w) : (u,w) \in \Gamma^+_u(G)\}$.
The \emph{out-degree} of node $v$, denoted by $\deg^+_v(G)$, is the number of its outgoing edges.
If a node has no outgoing edges, then it is called a \emph{sink}.
If it has no incoming edges, then it is called a \emph{source}.
Finally, a node is \emph{isolated} if it is both a sink and a source, i.e., it has no incident edges.

Many graph-theory concepts that we will use rely on the notion of a \emph{path}.
A \emph{path} is a sequence of nodes $(v_1,\dots,v_k)$ with $k > 1$ such that every two consecutive nodes are connected by an edge.
If for two nodes, $u,v$, there exist a path from $u$ to $v$, i.e., path $(v_1,\dots,v_k)$ in which $v_1=u$ and $v_k=v$, then $v$ is called a \emph{successor} of node $u$ and $u$ is a \emph{predecessor} of node $v$.
The length of a shortest such path is called a \emph{distance} from $u$ to $v$, denoted by $dist_{u,v}(G)$.
Node $v$ is a \emph{direct successor} of $u$ (and $u$ is a \emph{direct predecessor} of $v$) if this distance equals 1, i.e., if edge $(u,v)$ is in the graph.
We denote the set of all successors of node $v$ by $S_v(G)$ and all direct successors of $v$ by $S^1_v(G)$.
Analogously, we denote the set of all predecessors of node $v$ by $P_v(G)$ and all direct predecessors of $v$ by $P^1_v(G)$.
Note that a node may be its own successor or predecessor.
Finally, a \emph{cycle} is a path, $(v_1,\dots,v_k)$, on which all nodes are distinct and such that there exists an edge from $v_k$ to $v_1$ in the graph.

We have discussed so far how to model the structure of connections of the World Wide Web using graphs.
PageRank, however, is based not only on the topology of the World Wide Web, but also takes into account a basic exogenously given importance of each page. 
As argued by \citet{Page:etal:1999}, this importance can be used to personalize the measure for a specific user, for example by assuming the importance is non-zero only for a browser homepage or user's bookmarks.
Also, it can model how well a page fits into a given topic~\cite{Haveliwala:2002} or the fact that a page is trusted~\cite{Gyongyi:etal:2004}.
If no such information about the importance is available, it can be assumed that the importance of all pages are equal.
To include the basic importance of each page in our model, we will assume node have weights. Formally, a \emph{graph with node weights} is a pair $(G,\bs)$, where $G$ is a graph and $\bs$ is a \emph{node weight function} that maps a node to a non-negative real number. 
Hence, the weight of a node $v$ is denoted by $\bs(v)$.
We will also use $\bs(G)$ to denote the sum of weights of all nodes in graph $G$.

We are now ready to formally define the concept of \emph{centrality measures} and PageRank, in particular.
A \emph{centrality measure} $F$ is a function that for every node $v$ in a graph with node weights $(G,\bs)$ assigns a non-negative real value, denoted by $F_v(G,\bs)$.
Typically, the higher the number, the more \emph{important} or \emph{central} the node is.
\emph{PageRank}~\cite{Page:etal:1999} is a centrality measure defined as a unique function that satisfies the following recursive equation:
\begin{equation}
\label{eq:pr:main}
PR^a_v(G, \bs) = a \cdot \left( \sum_{(u,v) \in \Gamma^-_v(G)} \frac{PR^a_u(G,\bs)}{\deg_u^+(G)} \right) + \bs(v),
\end{equation}
where $a \in [0,1)$ is a decay factor which is an additional external parameter usually set to a value close to $1$, such as $0.85$ or $0.9$ (see~\ref{section:appendix:proof-2} for the proof of uniqueness).
We note that many variants on the definition of PageRank appear in the literature.
PageRank, as defined above, is unnormalized and may not sum up to one.
We discuss other variants in details in Section~\ref{section:definitions}.

The high-level intuition behind equation~\eqref{eq:pr:main} can be captured by saying that \emph{``important pages link to other important pages''}.
In more detail, according to equation~\eqref{eq:pr:main}, the importance of a node $v$ is defined as a sum of two factors: node's baseline importance, $\bs(v)$, and some share of the importance of nodes that link to $v$ multiplied by the decay factor $a$.
To interpret this share, consider node $u$ which is a direct predecessor of $v$ and assume that the importance of this node is equally distributed among its outgoing edges. 
Thus, from edge $(u,v)$, node $v$ obtains exactly the importance of node $u$ divided by the number of outgoing edges of $u$.
Note that we sum over all edges in the multiset of incoming edges of $v$, so edge $(u,v)$ can be taken into account multiple times.

To illustrate the way PageRank works, \citet{Page:etal:1999} proposed the \emph{random surfer model}.
The surfer starts browsing the World Wide Web from a random page (with distribution specified by the weights). 
Then, in each step of her walk with the probability $a$ the surfer chooses one of the links on a page and clicks it which moves her to the next page, or with the probability $1-a$ the surfer gets bored and jumps to a random page on the World Wide Web, as she would start browsing the World Wide Web all over again.
Now, PageRank is equal (up to a scalar multiplication) to the probability that we find the surfer on a page at an arbitrary time in the distant future, i.e., the limiting probability.
Note that we assume here that every page has at least one link; we discuss in Section 7.1 how to deal with pages without links.

The random surfer model provides an intuition on how PageRank works. 
However, jumps make the random walk hard to analyze.
As an alternative, it can also be assumed that the surfer instead of a jump simply stops surfing altogether.
Now, for this model, PageRank is equal (up to a scalar multiplication) to the expected number of times the surfer visits a page.
Details of this model that we named \emph{busy random surfer model} can be found in~\ref{section:appendix:proof-2}.

\begin{example}\label{example:preliminaries}
Consider PageRank of nodes in the graph from Fig.~\ref{figure:main} for $a = 0.9$, assuming each node has weight one:
\begin{center}
\begin{tabular}{c|cccccccc}
$v$ & $v_1$ & $v_2$ & $v_3$ & $v_4$ & $v_5$ & $v_6$ & $v_7$ & $v_8$ \\
\hline
$PR_v^a$ & $4.91$ & $5.02$ & $5.52$ & $1.43$ & $1.43$ & $1.43$ & $3.87$ & $6.71$
\end{tabular}
\end{center}
The most important node is $v_8$ which is the end of all outgoing edges of nodes $v_1$ and $v_4$. 
Hence, it aggregates the importance of both these nodes.
On the second and third place we have $v_2$ to which the most important node $v_8$ has two out of three outgoing edges and consequently $v_3$ which gets the only outgoing edge of $v_2$.
In turn, at the bottom of the ranking, we have three nodes: $v_5$ and $v_6$ which have incoming edges only from each other and $v_4$ who have exactly the same set of incoming edges as $v_6$; hence, it has the same PageRank.

Finally, consider node $v_7$.
Note that $v_7$ has more incoming edges than $v_8$ which is selected as the most important by PageRank.
However, these edges are from less important nodes which also have several other outgoing edges. 
As a result, their importance is distributed among other nodes and, according to PageRank, $v_7$ turn out to be the fourth least important node.
\end{example}


\section{Axioms for PageRank}\label{section:axioms}
Let us present our main result: the axiomatic characterization of PageRank.
Specifically, we introduce six simple properties, or \emph{axioms}, that PageRank satisfies.
Some of these axioms are satisfied also by other known centrality measures (see Section~\ref{section:comparison}).
However, what is important, PageRank is a unique centrality measure that satisfies all of them which is stated in Theorem~\ref{theorem:main}. The axioms are:

\begin{itemize}
\item \textit{\textbf{Node Deletion} (removing an isolated node from the graph does not affect centralities of other nodes):
For every graph $G = (V,E)$, node weights $\bs$ and isolated node $u \in V$ it holds that
$$ F_v(G,\bs) = F_v \big( \big( V \setminus \{u\}, E \big), \bs \big) \quad \mbox{for every } v \in V \setminus \{u\}.$$
}
\item \textit{\textbf{Edge Deletion} (removing an edge from the graph does not affect centralities of nodes which are not successors of the start of this edge):
For every graph $G = (V,E)$, node weights $\bs$ and edge $(u,w) \in E$ it holds that
$$ F_v(G,\bs) = F_v \big( \big( V , E - \lBrace (u,w) \rBrace \big), \bs \big) \quad \mbox{for every } v \in V \setminus S_u(G).$$
}
\item \textit{\textbf{Edge Multiplication} (creating additional copies of the outgoing edges of a node does not affect the centrality of any node):
For every graph $G=(V,E)$, node weights $\bs$, node $u \in V$ and $k \in \mathbb{N}$ it holds that
$$F_v(G,\bs) = F_v\big((V,E \sqcup \underbrace{\Gamma^+_u(G) \sqcup \dots \sqcup \Gamma^+_u(G)}_{k}), \bs\big) \quad \mbox{for every } v \in V.$$
}
\item \textit{\textbf{Edge Swap} (swapping ends of two outgoing edges of nodes with equal centralities and out-degrees does not affect the centrality of any node):
For every graph $G=(V,E)$, node weights $\bs$ and edges $ (u, u'), (w, w') \in E$
such that $F_u(G,\bs)=F_w(G,\bs)$ and
$\deg^+_u(G)=\deg^+_w(G)$
it holds that
\begin{equation*}
F_v(G,\bs) = F_v\big( (V, E - \lBrace(u,u'), (w,w')\rBrace \sqcup \lBrace(u,w'),(w,u')\rBrace), \bs) \quad \mbox{for every } v \in V.
\end{equation*}
}
\item \textit{\textbf{Node Redirect} (redirecting a node into its out-twin sums up their centralities and does not affect the centrality of other nodes):
For every graph $G=(V,E)$, node weights $\bs$ and out-twins $u,w \in V$ it holds that
\begin{equation*}
F_v(G,\bs) = F_v \big(R_{u \rightarrow w}(G, \bs) \big) \quad \mbox{for every } v \in V \setminus \{u,w\}
\end{equation*}
and $F_u(G,\bs) + F_w(G,\bs) = F_w\big(R_{u \rightarrow w}(G, \bs) \big)$, where
$R_{u \rightarrow w}(G,\bs) = \big( \big( V \setminus \{u\} , E - \Gamma^+_u(G)-\Gamma^-_u(G)  \sqcup \lBrace (v,w) : (v,u) \in \Gamma^-_u(G) \land v \neq u \rBrace \big), b' \big)$, where $\bs'(w) = \bs(u) + \bs(w)$ and $\bs'(v) = \bs(v)$ for every $v \in V \setminus \{u,w\}$.
}
\item \textit{\textbf{Baseline} (the centrality of an isolated node is equal to its weight):
For every graph $G = (V,E)$, node weights $\bs$ and isolated node $v \in V$ it holds that $F_v(G, \bs) = \bs(v)$.
}
\end{itemize}

The first five axioms are \emph{invariance axioms}.
Each invariance axiom is characterized by a graph operation, additional conditions on a graph and a set of nodes; given this, the axiom states that if the conditions are satisfied, then the graph operation does not affect the centrality of nodes in question.
Each axiom is named after the graph operation it considers.
Node Deletion and Edge Deletion concerns removing an isolated node or an outgoing edge of a node which is not a predecessor.
Edge Multiplication and Edge Swap focus on edge modifications: the former axiom considers replacing each outgoing edge of one node by multiple copies; the later one concerns swapping the ends of outgoing edges of two nodes with the same centrality and out-degree.
Finally, Node Redirect, considers removing a node from the graph and rewiring its incoming edges to its out-twin, i.e., a node with the same outgoing edges.
This operation is called a \emph{redirecting}.
Here, the axiom states not only that this operation does not affect centralities of other nodes, but also that the total centrality of both out-twins remain unchanged.
We illustrate all invariance axioms on Fig.~\ref{figure:axioms}.

The invariance axioms characterize PageRank up to a scalar multiplication, i.e., they are satisfied not only by PageRank, but also by PageRank multiplied by some constant.
In order to uniquely characterize PageRank, the last, sixth axiom called Baseline specifies the centrality of an isolated node.
The following theorem presents our main result.

\begin{theorem}\label{theorem:main}
A centrality measure satisfies Node Deletion, Edge Deletion, Edge Multiplication, Edge Swap, Node Redirect, and Baseline if and only if it is PageRank.
\end{theorem}

The next section is devoted to the proof of Theorem~\ref{theorem:main}.

Let us discuss the interpretation of the axioms with respect to the World Wide Web.
The first two axioms---Node Deletion and Edge Deletion---identify elements (pages and links) which are irrelevant for the importance of a page in question.
Node Deletion considers a page with no links or backlinks (e.g., a resource hidden on the server).
The axiom states that such a page does not have any impact on the rest of the network and its removal does not affect the importance of all the remaining pages.

For the Edge Deletion, imagine that there is a page A from which it is not possible to reach page B through a sequence of links (the studies show that such pairs of pages are very common~\cite{Broder:etal:2000}).
For example, imagine that page A has only links to its subpages that do not have external links themselves.
The axiom states that the links on A do not have an impact on the importance of B.
Hence, if we remove one of them, it will not affect the importance of B.
In particular, if A cannot be reached from A, i.e., if it is not possible to enter page A again after leaving through one of its links, then links of A does not affect also its own importance.
We note that Edge Deletion combined with Node Deletion implies that the importance of a page depends solely on the part of the World Wide Web from which this page can be reached.

Our next axiom, Edge Multiplication, considers multiplying the whole content of the page several times. 
This operation naturally increases the number of backlinks for many pages.
However, the axiom states that the importance of these pages, as well as all other pages in the network, do not change.
This means that the absolute number of links on a page does not matter as long as the proportion of links to other pages remains the same.
Looking from a different perspective, Edge Multiplication can be interpreted as robustness to manipulations by creating a large number of backlinks. 
Regardless of the number of links, the impact of a page is fixed to some extent.
As there is no cost of creating a link on the World Wide Web, avoiding such a manipulation lay at the foundation of PageRank.

For the next axiom, Edge Swap, consider a case where there are two equally important pages with an equal number of links.
The axiom states that the links from these pages have equal impact. It does not matter for the importance of any page from which of these two pages it has a backlink.
Hence, ends of edges can be swapped without affecting the importance of pages they link to and any other pages in the network.

For Node Redirect, imagine that there are two copies of the same page, i.e., two pages with identical content and links. 
Their backlinks, however, can differ.
Node Redirect states that URL redirecting, i.e., removing one of the copies and redirecting its incoming traffic to the other one, does not change the importance of other pages.
Moreover, the total importance of both pages will also remain intact.
At a high level, this axiom concerns a simple manipulation technique through creating several copies of the same page: the axiom states that merging them into one page does not change importance of any other page in the network.

Edge Multiplication and Node Redirect identify two manipulation techniques that do not affect PageRank of a page.
However, we note that PageRank is not resilient to other types of manipulations.
In particular, by modifying links a page may increase its PageRank~\cite{Was:etal:2019:rwd}.

Our last axiom, Baseline, considers a page without any links nor backlinks.
Such a page does not profit from the network structure, as it is not connected to any other page.
Hence, the axiom states that its centrality is equal to its basic importance.

\begin{figure}[t]
\centering
\begin{tikzpicture}
  \def\s{0.5cm} 
  \def\x{0cm} 
  \def\y{0cm} 
  \def\arrdist{0.6cm}

  \tikzset{
    node_blank/.style={circle,draw,minimum size=0.5cm,inner sep=0, color=white}, 
    node/.style={circle,draw,minimum size=0.5cm,inner sep=0, fill = black!05}, 
    edge/.style={sloped,-latex,above,font=\footnotesize}, 
    el/.style={below,font=\footnotesize}, 
    operation/.style={sloped,>=stealth,above,font=\footnotesize},
    arrow/.style={draw, single arrow, minimum width = 0.9cm, minimum height=\x-6*\s+\s, fill=black!10},
    blank/.style={}
  } 
  
  \def\x{0cm} 
  \def\y{0cm} 
  
  \node[node, minimum size=0.63cm, black!15] (A_0) at (\x+2*\s, 4*\s + \y) {$v_1$}; 
  \node[node] (A_1) at (\x+2*\s, 4*\s + \y) {$v_1$}; 
  \node[node] (A_2) at (\x+4*\s, 3*\s + \y) {$v_2$}; 
  \node[node] (A_3) at (\x+4*\s, 1*\s + \y) {$v_3$}; 
  \node[node] (A_4) at (\x+2*\s, 0*\s + \y) {$v_4$};
  \node[node] (A_5) at (\x+0*\s, 1*\s + \y) {$v_5$}; 
  \node[node] (A_6) at (\x+0*\s, 3*\s + \y) {$v_6$}; 
  \node[node, minimum size=0.63cm, black!15] (A_0) at (\x+2*\s, 2*\s + \y) {$v_7$}; 
  \node[node] (A_7) at (\x+2*\s, 2*\s + \y) {$v_7$}; 
  \node[blank] (A_) at (\x+4*\s-0.02cm, -0.2cm + \y) {$G_f$};
  \node[node_blank] (NR_start) at (\x+4*\s + \arrdist, 2*\s + \y) {};

  \path[->,draw,thick]
  (A_1) edge[edge, bend left=20, looseness = 0.9]  (A_2)
  (A_4) edge[edge, bend right=20, looseness = 0.9]  (A_3)
  (A_5) edge[edge, bend right=25]  (A_2)
  (A_5) edge[edge, bend right=20, looseness = 0.9]  (A_4)
  (A_5) edge[edge, bend right=15]  (A_6)
  (A_5) edge[edge] (A_7)
  (A_6) edge[edge, bend left=20, looseness = 0.9]  (A_1)
  (A_6) edge[edge]  (A_2)
  (A_6) edge[edge, bend right=15]  (A_5)
  (A_6) edge[edge]  (A_7)
  (A_7) edge[edge]  (A_2)
  ;
  
  \def\x{6cm} 
  
  \node[node] (A_1) at (\x+2*\s, 4*\s + \y) {$v_1$}; 
  \node[node] (A_2) at (\x+4*\s, 3*\s + \y) {$v_2$}; 
  \node[node] (A_3) at (\x+4*\s, 1*\s + \y) {$v_3$}; 
  \node[node] (A_4) at (\x+2*\s, 0*\s + \y) {$v_4$};
  \node[node] (A_5) at (\x+0*\s, 1*\s + \y) {$v_5$};
  \node[node] (A_6) at (\x+0*\s, 3*\s + \y) {$v_6$};
  \node[blank] (A_) at (\x+4*\s-0.02cm, -0.2cm + \y) {$G_e$};
  \node[node_blank] (NR_end) at (\x - \arrdist, 2*\s + \y) {};
  \node[node_blank] (ES_start) at (\x+4*\s + \arrdist, 2*\s + \y) {};

  \path[->,draw,thick]
  (A_1) edge[edge, bend left=20, looseness = 0.9]  (A_2)
  (A_4) edge[edge, bend right=20, looseness = 0.9]  (A_3)
  (A_5) edge[edge, bend right=10]  (A_1)
  (A_5) edge[edge, bend right=25, preaction={draw,black!15,-,double=black!15,double distance=2\pgflinewidth}]  (A_2)
  (A_5) edge[edge, bend right=20, looseness = 0.9]  (A_4)
  (A_5) edge[edge, bend right=15]  (A_6)
  (A_6) edge[edge] (A_1)
  (A_6) edge[edge, bend left=20, looseness = 0.9]  (A_1)
  (A_6) edge[edge]  (A_2)
  (A_6) edge[edge, bend right=15, preaction={draw,black!15,-,double=black!15,double distance=2\pgflinewidth}]  (A_5)
  ;
  
  \def\x{12cm} 
 
  \node[node_blank] (ES_end) at (\x - \arrdist, 2*\s + \y) {};
  \node[node_blank] (EM_start) at (\x+2*\s, 0.1cm + 4*\s + \y) {}; 
  \node[node] (A_1) at (\x+2*\s, 4*\s + \y) {$v_1$};
  \node[node] (A_2) at (\x+4*\s, 3*\s + \y) {$v_2$}; 
  \node[node] (A_3) at (\x+4*\s, 1*\s + \y) {$v_3$}; 
  \node[node] (A_4) at (\x+2*\s, 0*\s + \y) {$v_4$};
  \node[node] (A_5) at (\x+0*\s, 1*\s + \y) {$v_5$}; 
  \node[node, minimum size=0.63cm, black!15] (A_0) at (\x+0*\s, 3*\s + \y) {$v_6$}; 
  \node[node] (A_6) at (\x+0*\s, 3*\s + \y) {$v_6$};
  \node[blank] (A_) at (\x+4*\s-0.02cm, -0.2cm + \y) {$G_d$};

  \path[->,draw,thick]
  (A_1) edge[edge, bend left=20, looseness = 0.9]  (A_2)
  (A_4) edge[edge, bend right=20, looseness = 0.9]  (A_3)
  (A_5) edge[edge, bend right=10]  (A_1)
  (A_5) edge[edge, bend right=20, looseness = 0.9]  (A_4)
  (A_5) edge[edge, looseness=5, out = 200, in = 120] (A_5)
  (A_5) edge[edge]  (A_6)
  (A_6) edge[edge] (A_1)
  (A_6) edge[edge, bend left=20, looseness = 0.9]  (A_1)
  (A_6) edge[edge]  (A_2)
  (A_6) edge[edge, bend right=25]  (A_2)
  ;

  \def\y{3.7cm} 
  \def\x{12cm} 

  \node[node_blank] (ED_start) at (\x - \arrdist, 2*\s + \y) {};
  \node[node_blank] (EM_end) at (\x+2*\s, -0.1cm + \y) {};  
  \node[node] (A_1) at (\x+2*\s, 4*\s + \y) {$v_1$};
  \node[node] (A_2) at (\x+4*\s, 3*\s + \y) {$v_2$}; 
  \node[node] (A_3) at (\x+4*\s, 1*\s + \y) {$v_3$}; 
  \node[node] (A_4) at (\x+2*\s, 0*\s + \y) {$v_4$};
  \node[node] (A_5) at (\x+0*\s, 1*\s + \y) {$v_5$}; 
  \node[node] (A_6) at (\x+0*\s, 3*\s + \y) {$v_6$};
  \node[blank] (A_) at (\x+4*\s-0.02cm, -0.2cm + \y) {$G_c$};

  \path[->,draw,thick]
  (A_1) edge[edge, bend left=20, looseness = 0.9]  (A_2)
  (A_4) edge[edge, bend right=20, looseness = 0.9, preaction={draw,black!15,-,double=black!15,double distance=2\pgflinewidth}]  (A_3)
  (A_5) edge[edge, bend right=10]  (A_1)
  (A_5) edge[edge, bend right=20, looseness = 0.9]  (A_4)
  (A_5) edge[edge, looseness=5, out = 200, in = 120] (A_5)
  (A_5) edge[edge]  (A_6)
  (A_6) edge[edge, bend left=20, looseness = 0.9]  (A_1)
  (A_6) edge[edge]  (A_2)
  ;

  \def\x{6cm} 

  \node[node] (A_1) at (\x+2*\s, 4*\s + \y) {$v_1$};
  \node[node] (A_2) at (\x+4*\s, 3*\s + \y) {$v_2$}; 
  \node[node, minimum size=0.63cm, black!15] (A_0) at (\x+4*\s, 1*\s + \y) {$v_3$}; 
  \node[node] (A_3) at (\x+4*\s, 1*\s + \y) {$v_3$}; 
  \node[node] (A_4) at (\x+2*\s, 0*\s + \y) {$v_4$};
  \node[node] (A_5) at (\x+0*\s, 1*\s + \y) {$v_5$}; 
  \node[node] (A_6) at (\x+0*\s, 3*\s + \y) {$v_6$};
  \node[blank] (A_) at (\x+4*\s-0.02cm, -0.2cm + \y) {$G_b$};
  \node[node_blank] (ND_start) at (\x - \arrdist, 2*\s + \y) {};
  \node[node_blank] (ED_end) at (\x+4*\s + \arrdist, 2*\s + \y) {};

  \path[->,draw,thick]
  (A_1) edge[edge, bend left=20, looseness = 0.9]  (A_2)
  (A_5) edge[edge, bend right=10]  (A_1)
  (A_5) edge[edge, bend right=20, looseness = 0.9]  (A_4)
  (A_5) edge[edge, looseness=5, out = 200, in = 120] (A_5)
  (A_5) edge[edge]  (A_6)
  (A_6) edge[edge, bend left=20, looseness = 0.9]  (A_1)
  (A_6) edge[edge]  (A_2)
  ;
  
  \def\x{0cm} 

  \node[node] (A_1) at (\x+2*\s, 4*\s + \y) {$v_1$};
  \node[node] (A_2) at (\x+4*\s, 3*\s + \y) {$v_2$}; 
  \node[node] (A_4) at (\x+2*\s, 0*\s + \y) {$v_4$};
  \node[node] (A_5) at (\x+0*\s, 1*\s + \y) {$v_5$}; 
  \node[node] (A_6) at (\x+0*\s, 3*\s + \y) {$v_6$};
  \node[blank] (A_) at (\x+4*\s-0.02cm, -0.2cm + \y) {$G_a$};
  \node[node_blank] (ND_end) at (\x+4*\s + \arrdist, 2*\s + \y) {};

  \path[->,draw,thick]
  (A_1) edge[edge, bend left=20, looseness = 0.9]  (A_2)
  (A_5) edge[edge, bend right=10]  (A_1)
  (A_5) edge[edge, bend right=20, looseness = 0.9]  (A_4)
  (A_5) edge[edge, looseness=5, out = 200, in = 120] (A_5)
  (A_5) edge[edge]  (A_6)
  (A_6) edge[edge, bend left=20, looseness = 0.9]  (A_1)
  (A_6) edge[edge]  (A_2)
  ;

  \path[->,draw,very thick]
  (ND_start) edge[operation] node[el] {Node Deletion} (ND_end)
  (ED_start) edge[operation] node[el] {Edge Deletion} (ED_end)
  (EM_start) edge[operation] node[el, label={[align=center]90:Edge\\Multiplication}] {} (EM_end)
  (ES_start) edge[operation] node[el] {Edge Swap} (ES_end)
  (NR_start) edge[operation] node[el] {Node Redirect} (NR_end)
  ;
  
\end{tikzpicture}
\caption{An illustration of invariance axioms. In all graphs $\bs(v_i)=1$ for $i \in \{1,\dots,6\}$ and $\bs(v_7) = 0$.}
\label{figure:axioms}
\end{figure}
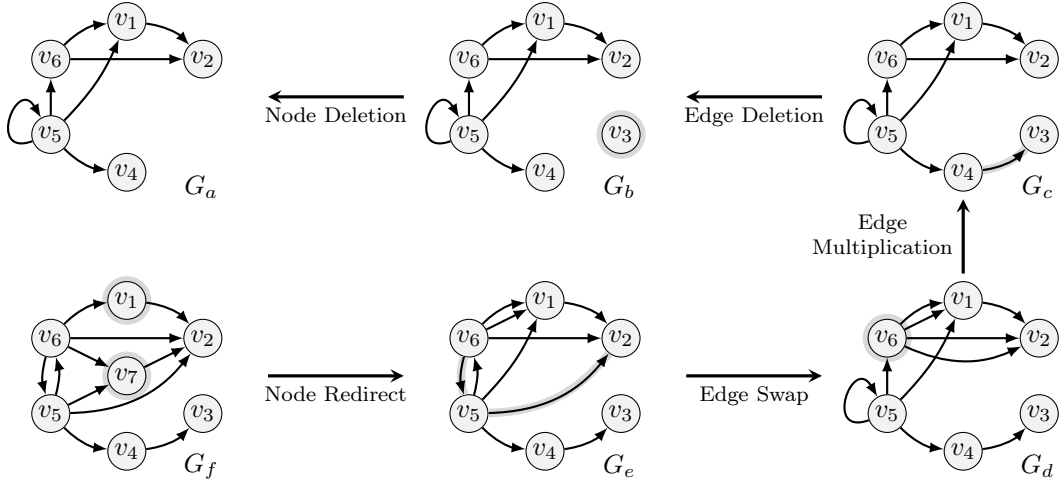

\begin{example}
As an illustration of our axioms consider graphs in Fig.~\ref{figure:axioms}.
Assume that centrality measure $F$ satisfies our five invariance axioms. 
Then, the centrality of node $v_1$ is not affected by the operations that transform graph $G_e$ into $G_a$ and equals the sum of centralities of $v_1$ and $v_7$ in graph $G_f$, i.e., $F_{v_1}(G_a,\bs) = F_{v_1}(G_f,\bs) + F_{v_7}(G_f,\bs)$.
More in detail:
\begin{itemize}
\item $F_{v_1}(G_a,\bs) = F_{v_1}(G_b,\bs)$ from \emph{Node Deletion}.
Node $v_3$ is isolated in graph $G_b$, hence its deletion does not affect the centralities of the remaining nodes.
\item $F_{v_1}(G_b,\bs) = F_{v_1}(G_c,\bs)$ from \emph{Edge Deletion}.
The only successor of node $v_4$ is node $v_3$,
thus deleting edge $(v_4,v_3)$ does not affect centralities of nodes other than $v_3$.
\item $F_{v_1}(G_c,\bs) = F_{v_1}(G_d,\bs)$ from \emph{Edge Multiplication},
since $G_d$ is obtained from $G_c$ by doubling the edges of node $v_6$.
\item $F_{v_1}(G_d,\bs) = F_{v_1}(G_e,\bs)$ from \emph{Edge Swap}.
Nodes $v_5$ and $v_6$ have both 4 outgoing edges.
If they have the same centrality (note that they have incoming edges only from themselves), then exchanging edges $(v_5,v_2), (v_6,v_5)$ for edges $(v_6,v_2), (v_5,v_5)$ does not affect the centrality of any node.
\item $F_{v_1}(G_e,\bs) = F_{v_1}(G_f,\bs) + F_{v_7}(G_f,\bs)$ from \emph{Node Redirect}.
Nodes $v_1$ and $v_7$ both have only one edge to node $v_2$, hence they are out-twins.
Therefore, redirecting $v_7$ into $v_1$ does not affect the centrality of other nodes and the centrality of node $v_1$ becomes the sum of the centralities of $v_1$ and $v_7$.
\end{itemize}
\end{example}


\section{Proof of the main theorem}\label{section:uniqueness}
In this section, we present the proof of Theorem~\ref{theorem:main}.
Here, we focus on showing that if a centrality measure satisfies all six of our axioms, then it must also satisfy PageRank recursive equation. 
Hence, it is PageRank for some decay factor $a \in [0,1)$.
The proof that PageRank for every decay factor $a \in [0,1)$ satisfies all six axioms can be found in~\ref{section:appendix:proof-2}.
Both parts combined imply Theorem~\ref{theorem:main}.
Additionally, in Section~\ref{section:proof-3}, we show that all six axioms are independent and necessary in our characterization.

Before we proceed, let us introduce some additional notation and definitions. 
Two graphs $G=(V,E)$, $G'=(V',E')$ are called \emph{disjoint} if their set of nodes are disjoint: $V \cap V' = \emptyset$.
In such a case, their sum $G+G'$ is a graph $(V \cup V',E \sqcup E')$.
For arbitrary node weights $\bs$, $\bs'$, the sum $(G,\bs) + (G',\bs')$ is a graph $(G+G',\bs+\bs')$, where $(\bs+\bs')(v) = \bs(v)$ for $v \in V$ and $(\bs+\bs')(v) = \bs'(v)$ for $v \in V'$.
For a natural number $k \in \mathbb{N}$ and an arbitrary multiset $E$, $k \cdot E$ denotes a union of $k$ copies of $E$, i.e., $0 \cdot E = \emptyset$ and $k \cdot E = ((k-1) \cdot E) \sqcup E$ for $k \ge 1$.

To denote small graphs, we use the following notation: 
\[ (G,\bs) = \big( ( \{v_1,\dots,v_n\},\lBrace e_1,\dots,e_m \rBrace ), [b_1,\dots,b_n] \big) \]
which means $G=( \{v_1,\dots,v_n\},\lBrace e_1,\dots,e_m \rBrace )$ and $\bs(v_i)=b_i$ for $i \in \{1,\dots,n\}$.
We will use a class of simple graphs called \emph{$k$-out stars}.
A $k$-out star is a graph that consists of $k+1$ nodes---one source and $k$ sinks---and $k$ edges connecting the source to all sinks in which all sinks have zero weights (we note that the assumption on the node weights is usually not a part of the definition).
Formally, graph $G = (V,E)$ with weights $\bs$ is a $k$-out star if $V = \{u, v_1, \dots, v_k\}$, $E = \lBrace (u,v_1), \dots, (u,v_k) \rBrace$ and $\bs(v_i) = 0$ for every $i \in \{1,\dots,k\}$.
In particular, every 1-out star is of the form $((\{u,v\}, \lBrace u,v \rBrace), [x,0])$, for some nodes $u,v$ and $x \in \mathbb{R}_{\ge 0}$.

For two nodes $u,v$, we denote the number of times edge $(u,v)$ appears in $G$ by $\#_{(u,v)}(G)$.


Now, let us move to the proof that if a centrality measure $F$ satisfies first five axioms, i.e., Node Deletion, Edge Deletion, Edge Multiplication, Edge Swap, and Node Redirect, then it is equal to PageRank for some decay factor up to a scalar multiplication.
Formally, we show that there exist two constants $c_F \in \mathbb{R}_{\ge 0}$ and $a_F \in [0,1)$, such that for every graph $(G,\bs)$ it holds that $F(G,\bs) = c_F \cdot PR^{a_F}(G,\bs)$.
Our proof has the following structure:
\begin{itemize}
\item First, we show that the centrality measure $F$ satisfies two basic properties: \emph{Locality}---the centrality of a node depends only on the part of the graph connected to it (Lemma~\ref{lemma:locality}) and \emph{Source Node}---the centrality of a source is equal to its weight multiplied by some non-negative constant (Lemma~\ref{lemma:source-node}). 
Moreover, this constant is the same for every source in every graph: it will be our constant $c_F$.
\item Then, we show that the centrality of a sink in 1-out star graph is equal to the centrality of a source multiplied by some non-negative constant (Lemma~\ref{lemma:1-arrow-df}). 
Moreover, this constant is the same for every 1-out star and lies in the interval $[0,1)$ (Lemma~\ref{lemma:1-arrow-af}): it will be our constant $a_F$.
\item Having defined $a_F$ and $c_F$, we turn our attention to proving that in every graph the centrality of any node is equal to PageRank with decay factor $a_F$ multiplied by $c_F$. 
We do it by considering increasingly complex graphs. 
Specifically, we start with 1-out stars (Lemma~\ref{lemma:1-arrow}) and then $k$-out stars (Lemma~\ref{lemma:k-arrow}). 
Furthermore, we consider arbitrary graphs with no cycles (Lemma~\ref{lemma:no-cycles}) and ultimately arbitrary graphs with possible cycles (Lemma~\ref{lemma:cycles}).
\end{itemize}
Finally, if $F$ additionally satisfies Baseline, then $c_F=1$ and $F(G,\bs) = PR^{a_F}(G,\bs)$ for every graph $(G,\bs)$ (Lemma~\ref{lemma:baseline}).


First, let us focus on a basic property of \emph{Locality}, which is implied by our first two axioms: Node Deletion and Edge Deletion.
Locality states that if a graph consists of several disjoint parts (also called \emph{connected components}), then the centrality of a node can be calculated by looking only at the part it is in~\cite{Skibski:etal:2019:attachment}.

\begin{lemma}\label{lemma:locality} (Locality)
If a centrality measure $F$ satisfies Node Deletion and Edge Deletion, then for every two disjoint graphs $G=(V,E)$ and $G'=(V',E')$, node weights $\bs$ and $\bs'$ and node $v \in V$ it holds that
$F_v\big((G,\bs)+(G',\bs')\big) = F_v(G,\bs)$.
\end{lemma}
\begin{proof}
Fix $v \in V$ and consider an arbitrary edge $(u,w) \in E'$.
Since graphs $G$ and $G'$ are disjoint, there is no path from $u$ to $v$.
In particular, $u$ is not a successor of $v$. 
Hence, from Edge Deletion, if we remove edge $(u,w)$, the centrality of node $v$ will remain unchanged.
Using this argument for all the edges from $E'$ we get that removing them does not affect the centrality of node $v$:
\begin{equation}
\label{eq:locality:1}
	F_v((G,\bs) + (G',\bs')) = F_v((V \cup V', E \sqcup E'),\bs+\bs') = F_v((V \cup V', E), \bs+\bs').
\end{equation}
Now, in graph $(V \cup V', E)$ all nodes from $V'$ are isolated.
Hence, from Node Deletion, removing these nodes does not affect the centrality of node $v$ as well:
\begin{equation}
\label{eq:locality:2}
	F_v((V \cup V', E), \bs+\bs') = F_v((V,E), \bs) = F_v(G,\bs).
\end{equation}
Combining equations~\eqref{eq:locality:1} and~\eqref{eq:locality:2} yields the thesis.
\end{proof}


In the second lemma, we prove that if a centrality measure satisfies Node Deletion, Edge Deletion and Node Redirect, then it also satisfies the property of \emph{Source Node}: the centrality of a source, i.e., a node without incoming edges, is proportional to its weight.
This property is similar to Baseline, but there are two differences:
First, Baseline applies only to isolated nodes and Source Node apply to all sources.
Second, Baseline implies that the centrality of a node is equal, not proportional, to its weight.

\begin{lemma}\label{lemma:source-node} (Source Node)
If a centrality measure $F$ satisfies Node Deletion, Edge Deletion and Node Redirect, then there exists a constant $c_F \in \mathbb{R}_{\ge 0}$ such that for every graph $G=(V,E)$, weights $\bs$ and every source $v \in V$ it holds:
$F_v(G,\bs) = c_F \cdot \bs(v)$.
Specifically, $c_F = F_{w}((\{w\}, \emptyset), [1])$ for an arbitrary node $w$.
\end{lemma}
\begin{proof}
We begin by considering graphs with one node and zero edges, i.e., graphs of the following form: $((\{v\}, \emptyset), [x])$ for some node $v$ and $x \in \mathbb{R}_{\ge 0}$. 
We will later show the relation between such graphs and sources in arbitrary graphs.

Consider two graphs $((\{u\},\emptyset), [x])$ and $((\{v\},\emptyset), [y])$ for arbitrary $u \neq v$ and $x,y \in \mathbb{R}_{\ge 0}$.
Let $(G,\bs)$ be their sum: $(G,\bs) = ((\{u,v\}, \emptyset), [x,y])$.
Since both $u$ and $v$ are isolated in $(G,\bs)$, from Node Deletion we know that their centralities are the same as in the original graphs. In particular:
\begin{equation} \label{eq:source-node:1}
F_u(G,\bs) + F_v(G,\bs) = F_u((\{u\},\emptyset), [x]) + F_v((\{v\},\emptyset), [y]).
\end{equation}
Nodes $u$ and $v$ are out-twins in $(G,\bs)$ (both have the same empty set of outgoing edges), so from Node Redirect, redirecting node $v$ into $u$ increases the centrality of $u$ by the centrality of $v$.
Such a redirecting results in graph $((\{u\}, \emptyset), [x+y])$, so we get:
\begin{equation} \label{eq:source-node:2}
F_u((\{u\}, \emptyset), [x+y]) = F_u(R_{v \rightarrow u}(G,\bs)) = F_u(G,\bs) + F_v(G,\bs).
\end{equation}
Combining equations \eqref{eq:source-node:1} and \eqref{eq:source-node:2} we have:
\begin{equation}\label{eq:source-node:3}
F_u((\{u\}, \emptyset), [x+y]) = F_u((\{u\}, \emptyset), [x]) + F_v((\{v\}, \emptyset), [y]).
\end{equation}
We make the following observations:
\begin{enumerate}
\item[(a)] $F_v((\{v\}, \emptyset), [0]) = 0$ for every $v$ (from equation~\eqref{eq:source-node:3} with $y=0$);
\item[(b)] $F_v((\{v\}, \emptyset), [y]) = F_u((\{u\}, \emptyset), [y])$ for every $u \neq v$ and $y \in \mathbb{R}_{\ge 0}$ (from equation~\eqref{eq:source-node:3} with $x=0$ and (a));
\item[(c)] $F_v((\{v\}, \emptyset), [x+y]) = F_v((\{v\}, \emptyset), [x]) + F_v((\{v\}, \emptyset), [y])$ for every $v$ and $x,y \in \mathbb{R}_{\ge 0}$ (from equation~\eqref{eq:source-node:3} and (b)).
\end{enumerate}
Note that (b) implies that the centrality of $v$ in the graph with node weights $((\{v\}, \emptyset), [x])$ depends solely on weight $x$. 
In other words, there exists a function $f: \mathbb{R}_{\ge 0} \rightarrow \mathbb{R}$ such that $F_v((\{v\}, \emptyset), [x]) = f(x)$.
Since centralities are non-negative, we know that $f$ is also non-negative: $f(x) \ge 0$ for every $x \in \mathbb{R}_{\ge 0}$.
On the other hand, from (c) we know that $f$ is additive: $f(x+y) = f(x)+f(y)$ for every $x,y \in \mathbb{R}_{\ge 0}$.
Non-negativity and additivity combined imply that $f$ is linear~\cite{Cauchy:1821}: $f(x) = c_F \cdot x$ for some $c_F \in \mathbb{R}_{\ge 0}$ for every $x \in \mathbb{R}_{\ge 0}$.
As a result, we know that that there exists $c_F \in \mathbb{R}_{\ge 0}$ such that for every node $v$:
\begin{equation}\label{eq:source-node:4}
F_v((\{v\}, \emptyset), [x]) = c_F \cdot x.
\end{equation}

Now, let $G = (V,E)$ be an arbitrary graph with node weights $\bs$ and $v$ be a source in $G$.
Since $v$ has no incoming edges, we know that it is not its own successor.
Hence, from Edge Deletion, removing its outgoing edges does not affect its centrality: $F_v(G,\bs) = F_v((V, E - \Gamma^+_v(G)), \bs)$.
In the resulting graph, $v$ is isolated, so from Locality (Lemma~\ref{lemma:locality}) we have that $F_v((V, E - \Gamma^+_v(G)), \bs) = F_v((\{v\}, \emptyset), [\bs(v)])$.
This combined with equation~\eqref{eq:source-node:4} yields the thesis.

Finally, from equation~\eqref{eq:source-node:4} for $v = w$ and $x=1$ we get that $c_F = F_{w}((\{w\}, \emptyset), [1])$ which concludes the proof.
\end{proof}

Consider an arbitrary graph $(G,\bs)$ in which node $v$ is a source.
Since the set of incoming edges of $v$ is empty, i.e., $\Gamma_v^-(G) = \emptyset$, from PageRank recursive equation~\eqref{eq:pr:main} we know that $PR_v^{a_F}(G,\bs) = \bs(v)$ for every decay factor $a_F \in [0,1)$.
This implies that---regardless of the decay factor $a_F$---centrality of $v$ is equal to PageRank of $v$  multiplied by $c_F$: $F_v(G, \bs) = c_F \cdot PR_v^{a_F}(G, \bs)$.


Now, let us focus on 1-out star graphs.
Recall that in a 1-out star there is one sink and one source connected by an edge and the sink has zero weight.
In the next lemma, we prove that the centrality of the sink is proportional to the weight of the source.

\begin{lemma}\label{lemma:1-arrow-df}
If a centrality measure $F$ satisfies Node Deletion, Edge Deletion, Edge Multiplication, Edge Swap, and Node Redirect, then there exists a constant $d_F \in \mathbb{R}_{\ge 0}$ such that for every 1-out star graph $(G,\bs) = ((\{u,v\}, \lBrace (u,v)\rBrace), [x,0])$ it holds $F_v(G,\bs) = d_F \cdot x$.
\end{lemma}
\begin{proof}
Our proof has a similar structure as the proof of Lemma~\ref{lemma:source-node}, but instead of graphs with a single node we consider 1-out stars.

\begin{figure}[t]
\centering
\begin{tikzpicture}
  \def\x{0.6cm} 
  \def\y{0cm} 

  \tikzset{
    node_blank/.style={circle,draw,minimum size=0.5cm,inner sep=0, color=white}, 
    node/.style={circle,draw,minimum size=0.5cm,inner sep=0, fill = black!05}, 
    edge/.style={sloped,-latex,above,font=\footnotesize},
    arrow/.style={draw, single arrow, minimum width = 0.9cm, minimum height=\y-6*\x+\s, fill=black!10},
    blank/.style={}
  }
  
  \node[node] (u) at (\y + 0*\x, 3*\x) {$u$};
  \node[node] (v) at (\y + 0*\x, 0*\x) {$v$};
  
  \node[blank] (G) at (\y + 1.5*\x-0.07cm, -0.7cm) {$(G,\bs)$};

  \path[->,draw,thick]
  (u) edge[edge]  (v)
  ;
  
  \def\y{1.5cm} 
  
  \node[node] (u) at (\y + 0*\x, 3*\x) {$u'$};
  \node[node] (v) at (\y + 0*\x, 0*\x) {$v'$};
  
  \node[blank] (G) at (\y + 0*\x-0.02cm, -0.7cm) {};

  \path[->,draw,thick]
  (u) edge[edge]  (v)
  ;
  
  \def\y{4cm} 
  
  \node[node] (u) at (\y + 0*\x, 3*\x) {$u$};
  \node[node] (v) at (\y + 1*\x, 0*\x) {$v$};
  \node[node] (u') at (\y + 2*\x, 3*\x) {$u'$};
  
  \node[blank] (G) at (\y + 1*\x-0.02cm, -0.7cm) {$(G',\bs')$};

  \path[->,draw,thick]
  (u) edge[edge]  (v)
  (u') edge[edge]  (v)
  ;
  
  \def\y{8cm} 
  
  \node[node] (u) at (\y + 0*\x, 3*\x) {$u$};
  \node[node] (v) at (\y + 0*\x, 0*\x) {$v$};
  
  \node[blank] (G) at (\y + 0*\x-0.02cm, -0.7cm) {$R_{u' \rightarrow u}(G',\bs')$};

  \path[->,draw,thick]
  (u) edge[edge]  (v)
  ;
  
\end{tikzpicture}
\caption{Graphs illustrating the proof of Lemma~\ref{lemma:1-arrow-df}.}
\label{figure:1-arrow-df}
\end{figure}
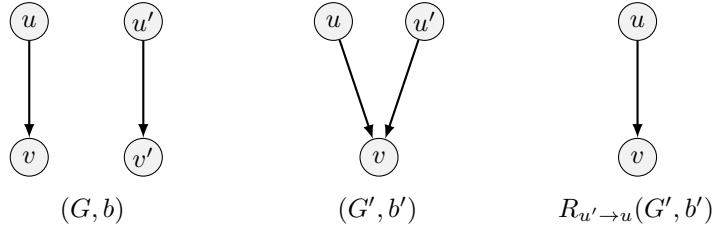

Consider two arbitrary 1-out stars $((\{u,v\}, \lBrace (u,v) \rBrace), [x,0])$ and $((\{u',v'\}, \lBrace (u',v') \rBrace), [y,0])$ such that node $u,v,u',v'$ are distinct.
Let $G$ be their sum: $(G,\bs) = ((\{u,v,u',v'\}, \lBrace (u,v), (u',v') \rBrace),$ $[x,0,y,0])$ (see Fig.~\ref{figure:1-arrow-df} for illustration). 
From Locality (Lemma~\ref{lemma:locality}) we know that centralities of all nodes in $(G, \bs)$ are the same as in the original graphs.
In particular,
\begin{equation}\label{eq:1-arrow-df:1}
F_v(G, \bs) + F_{v'}(G, \bs) = F_v((\{u,v\}, \lBrace (u,v) \rBrace), [x,0]) + F_{v'}((\{u',v'\}, \lBrace (u',v') \rBrace), [y,0]).
\end{equation}
Moreover, nodes $v$ and $v'$ are out-twins in $(G, \bs)$ (both have the same empty set of outgoing edges). 
Hence, from Node Redirect, redirecting $v'$ into $v$ increases the centrality of $v$ by the centrality of $v'$. 
Let us denote the resulting graph by $(G', \bs') = ((\{u,v,u'\}, \lBrace (u,v), (u',v) \rBrace), [x,0,y])$ (see Fig.~\ref{figure:1-arrow-df}).
Thus:
\begin{equation}\label{eq:1-arrow-df:2}
F_v(G',\bs') = F_v(R_{v' \rightarrow v}(G, \bs)) = F_v(G,\bs) + F_{v'}(G,\bs).
\end{equation}
Furthermore, nodes $u$ and $u'$ are out-twins in $(G',\bs')$ (both have only one edge to node $v$). 
Hence, again from Node Redirect, redirecting $u'$ into $u$ does not affect the centrality of node $v$.
Such a redirecting results in graph $((\{u,v\}, \lBrace (u,v) \rBrace), [x+y,0])$ which is a 1-out star (see Fig.~\ref{figure:1-arrow-df}), and we get:
\begin{equation}\label{eq:1-arrow-df:3}
F_v((\{u,v\}, \lBrace (u,v) \rBrace), [x+y,0]) = F_v(R_{u' \rightarrow u}(G',\bs')) = F_v(G',\bs').
\end{equation}
Combining equations~\eqref{eq:1-arrow-df:1}--\eqref{eq:1-arrow-df:3} we have:
\begin{equation}\label{eq:1-arrow-df:4}
F_v((\{u,v\}, \lBrace (u,v) \rBrace), [x+y,0]) =  F_v((\{u,v\}, \lBrace (u,v) \rBrace), [x,0]) + F_{v'}((\{u',v'\}, \lBrace (u',v') \rBrace), [y,0]).
\end{equation}
From equation~\eqref{eq:1-arrow-df:4} we get the following conditions:
\begin{enumerate}
\item[(a)] $F_{v'}((\{u',v'\}, \lBrace (u',v') \rBrace), [0,0]) = 0$ (from equation~\eqref{eq:1-arrow-df:4} with $y=0$);
\item[(b)] $F_v((\{u,v\}, \lBrace (u,v) \rBrace), [0,0]) = 0$ (from (a) and the fact that $u',v'$ were chosen arbitrary);
\item[(c)] $F_{v}((\{u,v\}, \lBrace (u,v) \rBrace), [y,0]) = F_{v'}((\{u',v'\}, \lBrace (u',v') \rBrace), [y,0])$ for every $y \in \mathbb{R}_{\ge 0}$ (from equation~\eqref{eq:1-arrow-df:4} with $x=0$ and (b));
\item[(d)] $F_{v}((\{u,v\}, \lBrace (u,v) \rBrace), [x+y,0]) = F_{v}((\{u,v\}, \lBrace (u,v) \rBrace), [x,0]) + F_{v}((\{u,v\}, \lBrace (u,v) \rBrace), [y,0])$ for every $x,y \in \mathbb{R}_{\ge 0}$ (from equation~\eqref{eq:1-arrow-df:4} and (c)).
\end{enumerate}
From the fact that nodes were chosen arbitrarily, we know that (c) holds for every four pairwise distinct nodes $u,v,u',v'$.
However, to prove that $F_{v}((\{u,v\}, \lBrace (u,v) \rBrace), [y,0])$ does not depend on nodes' names we need to prove the equality holds for any two pairs of nodes.

To this end, assume they are not distinct, i.e., $u \neq v$ and $u' \neq v'$, but $\{u,v\} \cap \{u',v'\} \neq \emptyset$.
Then, let us take two nodes $u'',v''$ which are different than $u,v,u',v'$ and consider graph $((\{u'',v''\}, \lBrace (u'',v'') \rBrace), [y,0])$ constructed from them.
Using (c) two times we get that:
\[ F_{v}((\{u,v\}, \lBrace (u,v) \rBrace), [y,0]) = F_{v''}((\{u'',v''\}, \lBrace (u'',v'') \rBrace), [y,0]) = F_{v'}((\{u',v'\}, \lBrace (u',v') \rBrace), [y,0]) \]
for every $y \in \mathbb{R}_{\ge 0}$ which implies (c) also for not pairwise distinct nodes.
Furthermore, this means that the centrality of a sink in a 1-out star graph depends only on the weight of the source, i.e., there exists a function $f: \mathbb{R}_{\ge 0} \rightarrow \mathbb{R}$ such that $F_v(\{u,v\}, \lBrace u,v \rBrace, [x,0]) = f(x)$.
Since centralities are non-negative, we know that $f$ is also non-negative: $f(x) \ge 0$ for every $x \in \mathbb{R}_{\ge 0}$. 
On the other hand, from (d) we know that $f$ is additive: $f(x+y) = f(x) + f(y)$ for every $x,y \in \mathbb{R}_{\ge 0}$. 
Now, non-negativity combined with additivity implies that $f$ is linear~\cite{Cauchy:1821}: $f(x) = d_F \cdot x$.
This concludes the proof.
\end{proof}


In the next lemma, building upon Lemma~\ref{lemma:1-arrow-df}, we show that there exists a constant $a_F \in [0,1)$ such that the centrality of a sink in 1-out star graph equals $a_F \cdot c_F \cdot x$ where $x$ is the weight of the source (see Lemma~\ref{lemma:source-node} for the definition of $c_F$). 

\begin{lemma}\label{lemma:1-arrow-af}
If a centrality measure $F$ satisfies Node Deletion, Edge Deletion, Edge Multiplication, Edge Swap, and Node Redirect, then there exists a constant $a_F \in [0,1)$ such that for every 1-out star graph $(G,\bs) = ((\{u,v\}, \lBrace (u,v)\rBrace), [x,0])$ it holds $F_v(G,\bs) = a_F \cdot c_F \cdot x$.
Specifically, $a_F = F_{w}((\{w',w\}, \lBrace (w',w) \rBrace), [1,0])/c_F$ for arbitrary nodes $w,w'$ if $c_F > 0$ and $a_F = 0$, otherwise.
\end{lemma}
\begin{proof}
So far, we have proved that there exist constants $c_F,d_F \in \mathbb{R}_{\ge 0}$ such that for every 1-out star the centrality of the source equals $c_F \cdot x$ (Lemma~\ref{lemma:source-node}) and the centrality of the sink equals $d_F \cdot x$ (Lemma~\ref{lemma:1-arrow-df}), where $x$ is the weight of the source.
Hence, to show that the centrality of the sink equals $a_F \cdot c_F \cdot x$ for some $a_F \in [0,1)$ it is enough to prove that if $c_F = 0$, then $d_F = 0$, and if $c_F \neq 0$, then $d_F < c_F$.

Assume $c_F = 0$. Consider graph $(G,\bs) = ((\{u,v,u',v'\}, \lBrace (u,v), (u',v') \rBrace), [1,0,0,0])$ which is a sum of two 1-out stars.
Since $c_F = 0$, from Source Node (Lemma~\ref{lemma:source-node}) we know that $F_u(G,\bs) = 0 = F_{u'}(G,\bs)$ and from Lemma~\ref{lemma:1-arrow-df} and Locality (Lemma~\ref{lemma:locality}) $F_v(G,\bs) = d_F$.
Nodes $u$ and $u'$ each have one edge, so from Edge Swap we get that if we replace ends of both these edges centralities in the graph will not change.
Formally, for graph $(G',\bs') = ((\{u,v,u',v'\}, \lBrace (u,v') (v,u') \rBrace), [1,0,0,0])$ we have that $F_v(G',\bs') = F_v(G,\bs)$.
However, from Lemma~\ref{lemma:1-arrow-df} and Locality (Lemma~\ref{lemma:locality}) we get that $F_v(G',\bs') = 0$.
Hence, $d_F = 0$.

\begin{figure}[t]
\centering
\begin{tikzpicture}
  \def\x{0.6cm} 
  \def\y{0cm} 

  \tikzset{
    node_blank/.style={circle,draw,minimum size=0.5cm,inner sep=0, color=white}, 
    node/.style={circle,draw,minimum size=0.5cm,inner sep=0, fill = black!05}, 
    edge/.style={sloped,-latex,above,font=\footnotesize},
    arrow/.style={draw, single arrow, minimum width = 0.9cm, minimum height=\y-6*\x+\s, fill=black!10},
    blank/.style={}
  }
  
  \node[node] (u) at (\y + 0*\x, 0*\x) {$u$};
  \node[node] (v) at (\y + 1*\x, 2*\x) {$v$};
  
  \node[blank] (G) at (\y + 1*\x-0.02cm, -0.7cm) {$G$};

  \path[->,draw,thick]
  (u) edge[edge]  (v)
  ;
  
  \def\y{3cm} 
  
  \node[node] (u) at (\y + 0*\x, 0*\x) {$u$};
  \node[node] (v) at (\y + 1*\x, 2*\x) {$v$};
  \node[node] (w) at (\y + 2*\x, 0*\x) {$w$};
  
  \node[blank] (G) at (\y + 1*\x-0.02cm, -0.7cm) {$G'$};

  \path[->,draw,thick]
  (u) edge[edge]  (w)
  (v) edge[edge, out =320, in = 60, looseness=7]  (v)
  ;
  
  \def\y{6cm} 
  
  \node[node] (u) at (\y + 0*\x, 0*\x) {$u$};
  \node[node] (v) at (\y + 1*\x, 2*\x) {$v$};
  \node[node] (w) at (\y + 2*\x, 0*\x) {$w$};
  
  \node[blank] (G) at (\y + 1*\x-0.02cm, -0.7cm) {$G''$};

  \path[->,draw,thick]
  (u) edge[edge]  (v)
  (v) edge[edge]  (w)
  ;
  
  \def\y{9cm} 
  
  \node[node] (u) at (\y + 0*\x, 0*\x) {$u$};
  \node[node] (v) at (\y + 1*\x, 2*\x) {$v$};
  \node[node] (v') at (\y + 2*\x, 0*\x) {$v'$};
  
  \node[blank] (G) at (\y + 1*\x-0.02cm, -0.7cm) {$G^*$};

  \path[->,draw,thick]
  (u) edge[edge]  (v)
  ;
  
\end{tikzpicture}
\caption{Graphs illustrating the proof of Lemma~\ref{lemma:1-arrow-af}.}
\label{figure:1-arrow-af}
\end{figure}
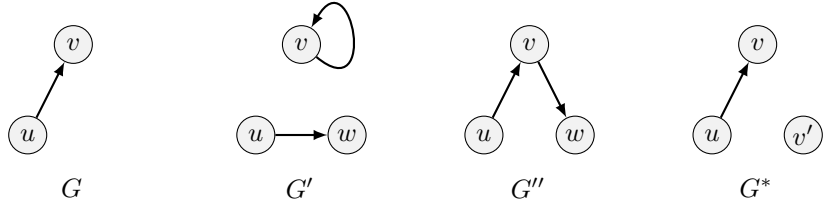

Assume now that $c_F \neq 0$. To show that $d_F < c_F$, we will show that in one particular 1-out star the centrality of the sink is strictly smaller than the centrality of the source which building upon the above general results will imply the thesis.
Let $u,v$ be two arbitrary distinct nodes and let $y$ be the centrality of $v$ in graph $((\{v\}, \lBrace (v,v) \rBrace), [1])$ divided by $c_F$: $y = F_v((\{v\}, \lBrace (v,v) \rBrace), [1]) / c_F$.
We will show that in 1-out star $((\{u,v\}, \lBrace (u,v) \rBrace), [y,0])$ the centrality of the sink is smaller than the centrality of the source:
\begin{equation}\label{eq:1-arrow-af:0} 
F_u((\{u,v\}, \lBrace (u,v) \rBrace), [y,0]) > F_v((\{u,v\}, \lBrace (u,v) \rBrace), [y,0]).
\end{equation}
To this end, we begin by proving that the centralities of the source and the sink in graph $(G,\bs) = ((\{u,v\}, \lBrace (u,v) \rBrace), [y, 1])$ are equal:
\begin{equation}\label{eq:1-arrow-af:1}
F_u((\{u,v\}, \lBrace (u,v) \rBrace), [y, 1]) = F_v((\{u,v\}, \lBrace (u,v) \rBrace), [y, 1]).
\end{equation}

To prove equation~\eqref{eq:1-arrow-af:1}, first let us consider graph $(G',\bs') = ((\{u,v,w\}, \lBrace (u,w), (v,v) \rBrace), [y,1,0])$ which is the sum of two graphs: 1-out star $((\{u,w\}, \lBrace (u,w) \rBrace), [y,0])$ and already mentioned graph with one node $((\{v\}, \lBrace (v,v) \rBrace), [1])$ (see Fig.~\ref{figure:1-arrow-af} for illustration).
From Locality (Lemma~\ref{lemma:locality}) and the definition of $y$, we know that the centrality of node $v$ equals $c_F \cdot y$.
On the other hand, $u$ is a source in $(G',\bs')$, so from Source Node (Lemma~\ref{lemma:source-node}), we know that its centrality also equals $c_F \cdot y$. 
Hence, centralities of both nodes are equal:
\[ F_u(G',\bs') = F_v(G',\bs'). \]
Since $u$ and $v$ have the same centralities in $(G',\bs')$ and both have one outgoing edge, from Edge Swap we know that exchanging the ends of these edges does not affect the centralities in the graph.
Let $(G'',\bs'')$ be the resulting graph: $(G'',\bs'') = ((\{u,v,w\}, \lBrace (u,v), (v,w) \rBrace), [y,1,0])$ (see Fig.~\ref{figure:1-arrow-af}). 
We get that:
\[ F_u(G'',\bs'') = F_v(G'',\bs''). \]
Furthermore, since $v$ is not its own successor and is not a successor of $u$, i.e., $v \not \in S_u(G'',\bs'')$, $v \not \in S_v(G'',\bs'')$, from Edge Deletion we can remove edge $(v,w)$ from $(G'',\bs'')$ without affecting centralities of nodes $u,v$.
Moreover, after deleting edge $(v,w)$ node $w$ is isolated, so from Node Deletion $w$ can also be deleted without affecting these centralities.
In this way, we obtain graph $(G,\bs) = ((\{u,v\}, \lBrace (u,v) \rBrace), [y,1])$ which proves equation~\eqref{eq:1-arrow-af:1}.

Let us go back to proving equation~\eqref{eq:1-arrow-af:0}.
Observe that graph $((\{u,v\}, \lBrace (u,v) \rBrace), [y, 0])$ is obtained from $(G,\bs)$ by changing the weight of $v$ to zero.
Since $u$ is a source also in this graph and its weight did not change, from Source Node (Lemma~\ref{lemma:source-node}) we know that its centrality is the same as in graph $(G,\bs)$:
\begin{equation}\label{eq:1-arrow-af:2}
F_u((\{u,v\}, \lBrace (u,v) \rBrace), [y,0]) = F_u((\{u,v\}, \lBrace (u,v) \rBrace), [y,1]).
\end{equation}
Now, let us turn our attention to the centrality of $v$.
Consider graph $(G^*,\bs^*) = ((\{u,v,v'\}, \lBrace (u,v) \rBrace),$ $[y, 0, 1])$ (see Fig.~\ref{figure:1-arrow-af}).
Since $v'$ is a source in this graph, we know that $F_{v'}(G^*,\bs^*) = c_F$.
Moreover, both $v$ and $v'$ do not have any outgoing edges in $(G^*,\bs^*)$, so from Node Redirect, redirecting $v'$ into $v$ increases the centrality of $v$ by the centrality of $v'$, i.e., by $c_F$. 
Such redirecting results in graph $((\{u,v\}, \lBrace (u,v) \rBrace), [y,1])$.
Hence, we get:
\begin{equation}\label{eq:1-arrow-af:3}
F_v((\{u,v\}, \lBrace (u,v) \rBrace), [y,0]) = F_v((\{u,v\}, \lBrace (u,v) \rBrace), [y,1]) - c_F.
\end{equation}
Since we assumed that $c_F > 0$, equation~\eqref{eq:1-arrow-af:3} combined with equations~\eqref{eq:1-arrow-af:1} and \eqref{eq:1-arrow-af:2} implies equation~\eqref{eq:1-arrow-af:0}.

Therefore, indeed, $F_w((\{w', w), \lBrace (w', w)\rBrace), [1,0]) = a_F \cdot c_F$ for arbitrary nodes $w,w'$. 
Hence, if $c_F > 0$ we get that $a_F = F_w((\{w', w), \lBrace (w', w)\rBrace), [1,0])/c_F$ and if $c_F = 0$, $a_F$ can be defined arbitrary: we will assume $a_F = 0$.
\end{proof}


A direct corollary from Source Node (Lemma~\ref{lemma:source-node}) and Lemma~\ref{lemma:1-arrow-af} is the fact that in 1-out stars centralities are equal to PageRank with decay factor $a_F$ multiplied by $c_F$ (see these lemmas for the definitions of $a_F$ and $c_F$).

\begin{lemma}\label{lemma:1-arrow}
If a centrality measure $F$ satisfies Node Deletion, Edge Deletion, Edge Multiplication, Edge Swap, and Node Redirect, then for every 1-out star graph $(G,\bs) = ((\{u,v\}, \lBrace (u,v)\rBrace), [x,0])$ it holds $F(G,\bs) = c_F \cdot PR^{a_F}(G,\bs)$.
Specifically, $F_u(G,\bs) = c_F \cdot x$ and $F_v(G,\bs) = a_F \cdot c_F \cdot x$.
\end{lemma}
\begin{proof}
From Source Node (Lemma~\ref{lemma:source-node}) we know that $F_u(G,\bs) = c_F \cdot x$ and from Lemma~\ref{lemma:1-arrow-af} we know that $F_v(G,\bs) = a_F \cdot c_F \cdot x$. 
On the other hand, from PageRank recursive equation~\eqref{eq:pr:main} we have that $PR^{a_F}_u(G,\bs) = x$ (node $u$ has weight $x$ and no incoming edges) and $PR^{a_F}_v(G,\bs) = a_F \cdot PR^{a_F}_u(G,\bs) = a_F \cdot c_F \cdot x$ (node $v$ has a zero weight and only one incoming edge from node $u$ with $\deg_u^+(G)=1$). This concludes the proof.
\end{proof}


In the next lemma, we extend the result from Lemma~\ref{lemma:1-arrow} concerning 1-out stars to $k$-out stars.
More in detail, we show that the centrality of every sink in $k$-out star equals $a_F \cdot c_F \cdot x/k$, where $x$ is the weight of the source.
Hence, the centrality $a_F \cdot c_F \cdot x$ of the sink in a 1-out star is split equally among all $k$ sinks.

\begin{lemma}\label{lemma:k-arrow}
If centrality measure $F$ satisfies Node Deletion, Edge Deletion, Edge Multiplication, Edge Swap, and Node Redirect, then for every $k$-out star graph
\[
    (G,\bs) = ((\{u,v_1,\dots,v_k\},\lBrace(u,v_1),\dots,(u,v_k)\rBrace), [x,0,\dots,0])
\]
it holds that $F(G,\bs) = c_F \cdot PR^{a_F}(G,\bs)$.
Specifically, $F_u(G,\bs) = c_F \cdot x$ and $F_{v_i}(G,\bs) = a_F \cdot c_F \cdot x / k$ for every $i \in \{1,\dots,k\}$.
\end{lemma}
\begin{proof}
Consider an arbitrary $k$-out star graph: $(G,\bs) = ((\{u_1,v_1,\dots,v_k\}, \lBrace (u_1,v_1), \dots, (u_1,v_k)\rBrace),$ $[x,0,\dots,0])$ (note that for notational convenience we will denote the source node by $u_1$, not $u$).
See Fig.~\ref{figure:k-arrow} for illustration.
We need to prove that $F_{v_i}(G,\bs) = a_F \cdot c_F \cdot x/k$ for every $i \in \{1,\dots,k\}$.
To this end, through a series of invariance operations, we will show that splitting the source of a $k$-out star into $k$ separate sources, each with $1$ edge and $1/k$ of the original weight does not affect the centralities of sinks.
In so doing, we obtain $k$ separate 1-out stars and Lemma~\ref{lemma:1-arrow} will imply the thesis.

First, take $k-1$ distinct nodes, $u_2,\dots,u_k$, that do not appear in $(G,\bs)$ and consider a graph obtained from $(G,\bs)$ by splitting $u_1$ into $k$ nodes, $u_1,\dots,u_k$, each with the same edges as $u_1$ in $(G,\bs)$ and $1/k$ of the original weight (see Fig.~\ref{figure:k-arrow}):
\[ (G',\bs') = ((\{u_1,\dots,u_k,v_1,\dots,v_k\}, \bigsqcup_{i=1}^k \lBrace (u_i, v_1), \dots, (u_i, v_k) \rBrace), [x/k, \dots, x/k, 0, \dots, 0]). \]
More formally, we have $\bs'(u_i) = x/k$ and $\bs'(v_i) = 0$ for every $i \in \{1,\dots,k\}$.
Graph $(G',\bs')$ contains $k$ identical sources ($u_1,\dots,u_k$) and $k$ identical sinks ($v_1,\dots,v_k$).
All sources are out-twins, so from Node Redirect redirecting one of the nodes $u_2,\dots,u_k$ into $u_1$ does not affect the centralities of sinks. 
This operation does not change edges of $u_1$ and the remaining sources, hence they are still out-twins. 
By performing all such redirectings one by one, i.e., redirecting $u_2$ into $u_1$, $u_3$ into $u_1$ and so on, we will eventually obtain the original graph $(G,\bs)$.
Hence, we have: $F_{v_i}(G', \bs') = F_{v_i}(G, \bs)$ for every $i \in \{1,\dots,k\}$.

Next, fix arbitrary $i,j \in \{1,\dots,k\}$ and consider replacing in graph $(G',\bs')$ edges $(u_i,v_j), (u_j,v_i)$ with edges $(u_i,v_i),(u_j,v_j)$, i.e., an edge swap.
Both sources $u_i$ and $u_j$ have exactly $k$ outgoing edges and from Source Node (Lemma~\ref{lemma:source-node}) they have the same centrality equal to $c_F \cdot x/k$.
Hence, from Edge Swap this operation does not affect centralities in the graph.
Moreover, it does not affect the number of outgoing edges of any node.
Hence, by sequentially replacing edges $(u_i,v_j), (u_j,v_i)$ with $(u_i,v_i),(u_j,v_j)$ for all (unordered) pairs $i,j \in \{1,\dots,k\}$ we obtain graph:
\[ (G'',\bs'') = ((\{u_1,\dots,u_k,v_1,\dots,v_k\}, k \cdot \lBrace (u_1, v_1), \dots, (u_k, v_k) \rBrace), [x/k, \dots, x/k, 0, \dots, 0]), \]
(see Fig.~\ref{figure:k-arrow}) and we know that centralities of sinks did not change: $F_{v_i}(G'', \bs'') = F_{v_i}(G', \bs')$ for every $i \in \{1,\dots,k\}$.

In graph $(G'',\bs'')$ each source $u_i$ has $k$ edges, all to the same node $v_i$. 
From Edge Multiplication we know that replacing these $k$ edges with only one edge does not affect centralities in the graph.
Hence, for a graph 
\[ (G^*,\bs^*) = ((\{u_1,\dots,u_k,v_1,\dots,v_k\}, \lBrace (u_1, v_1), \dots, (u_k, v_k) \rBrace), [x/k, \dots, x/k, 0, \dots, 0]), \]
we get $F_{v_i}(G^*, \bs^*) = F_{v_i}(G'', \bs'')$ for every $i \in \{1,\dots,k\}$.

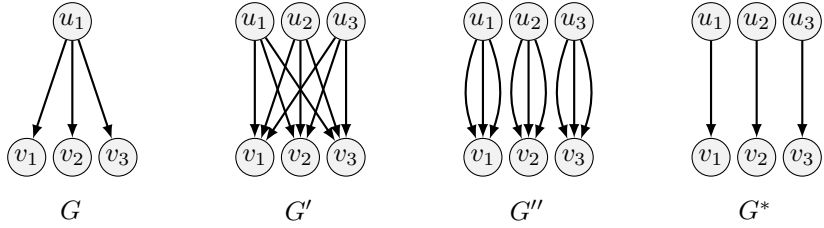
\begin{figure}[t]
\centering
\begin{tikzpicture}
  \def\x{0.6cm} 
  \def\y{0cm} 

  \tikzset{
    node_blank/.style={circle,draw,minimum size=0.5cm,inner sep=0, color=white}, 
    node/.style={circle,draw,minimum size=0.5cm,inner sep=0, fill = black!05}, 
    edge/.style={sloped,-latex,above,font=\footnotesize},
    arrow/.style={draw, single arrow, minimum width = 0.9cm, minimum height=\y-6*\x+\s, fill=black!10},
    blank/.style={}
  }
  
  \node[node] (u_1) at (\y + 1*\x, 3*\x) {$u_1$};
  \node[node] (v_1) at (\y + 0*\x, 0*\x) {$v_1$};
  \node[node] (v_2) at (\y + 1*\x, 0*\x) {$v_2$};
  \node[node] (v_3) at (\y + 2*\x, 0*\x) {$v_3$};
  
  \node[blank] (G) at (\y + 1*\x-0.02cm, -0.7cm) {$G$};

  \path[->,draw,thick]
  (u_1) edge[edge]  (v_1)
  (u_1) edge[edge]  (v_2)
  (u_1) edge[edge]  (v_3)
  ;
  
  \def\y{3cm} 
  
  \node[node] (u_1) at (\y + 0*\x, 3*\x) {$u_1$};
  \node[node] (u_2) at (\y + 1*\x, 3*\x) {$u_2$};
  \node[node] (u_3) at (\y + 2*\x, 3*\x) {$u_3$};
  \node[node] (v_1) at (\y + 0*\x, 0*\x) {$v_1$};
  \node[node] (v_2) at (\y + 1*\x, 0*\x) {$v_2$};
  \node[node] (v_3) at (\y + 2*\x, 0*\x) {$v_3$};
  
  \node[blank] (G) at (\y + 1*\x-0.02cm, -0.7cm) {$G'$};

  \path[->,draw,thick]
  (u_1) edge[edge]  (v_1)
  (u_1) edge[edge]  (v_2)
  (u_1) edge[edge]  (v_3)
  (u_2) edge[edge]  (v_1)
  (u_2) edge[edge]  (v_2)
  (u_2) edge[edge]  (v_3)
  (u_3) edge[edge]  (v_1)
  (u_3) edge[edge]  (v_2)
  (u_3) edge[edge]  (v_3)
  ;
  
  \def\y{6cm} 
  
  \node[node] (u_1) at (\y + 0*\x, 3*\x) {$u_1$};
  \node[node] (u_2) at (\y + 1*\x, 3*\x) {$u_2$};
  \node[node] (u_3) at (\y + 2*\x, 3*\x) {$u_3$};
  \node[node] (v_1) at (\y + 0*\x, 0*\x) {$v_1$};
  \node[node] (v_2) at (\y + 1*\x, 0*\x) {$v_2$};
  \node[node] (v_3) at (\y + 2*\x, 0*\x) {$v_3$};
  
  \node[blank] (G) at (\y + 1*\x-0.02cm, -0.7cm) {$G''$};

  \path[->,draw,thick]
  (u_1) edge[edge, bend left = 20]  (v_1)
  (u_1) edge[edge]  (v_1)
  (u_1) edge[edge, bend right = 20]  (v_1)
  (u_2) edge[edge, bend left = 20]  (v_2)
  (u_2) edge[edge]  (v_2)
  (u_2) edge[edge, bend right = 20]  (v_2)
  (u_3) edge[edge, bend left = 20]  (v_3)
  (u_3) edge[edge]  (v_3)
  (u_3) edge[edge, bend right = 20]  (v_3)
  ;
  
  \def\y{9cm} 
  
  \node[node] (u_1) at (\y + 0*\x, 3*\x) {$u_1$};
  \node[node] (u_2) at (\y + 1*\x, 3*\x) {$u_2$};
  \node[node] (u_3) at (\y + 2*\x, 3*\x) {$u_3$};
  \node[node] (v_1) at (\y + 0*\x, 0*\x) {$v_1$};
  \node[node] (v_2) at (\y + 1*\x, 0*\x) {$v_2$};
  \node[node] (v_3) at (\y + 2*\x, 0*\x) {$v_3$};
  
  \node[blank] (G) at (\y + 1*\x-0.02cm, -0.7cm) {$G^*$};

  \path[->,draw,thick]
  (u_1) edge[edge]  (v_1)
  (u_2) edge[edge]  (v_2)
  (u_3) edge[edge]  (v_3)
  ;
  
\end{tikzpicture}
\caption{Example graphs illustrating the proof of Lemma~\ref{lemma:k-arrow} for $k=3$.}
\label{figure:k-arrow}
\end{figure}

Finally, observe that graph $G^*$ is a sum of $k$ separate 1-out stars. 
Hence, from Locality (Lemma~\ref{lemma:locality}) and Lemma~\ref{lemma:1-arrow} we get that for every $i \in \{1,\dots,k\}$ we have:
\[ F_{v_i}(G^*, \bs^*) = F_{v_i}((\{u_i,v_i\}, \lBrace (u_i,v_i) \rBrace ), [x/k, 0]) = a_F \cdot c_F \cdot x/k.\]
As a result, we showed that the centrality of every sink $v_i$ in $(G^*,\bs^*)$ is the same as in $(G'',\bs'')$, $(G',\bs')$ and eventually in $(G,\bs)$, so we have $F_{v_i}(G,\bs) = a_F \cdot c_F \cdot x/k$.

Now, equality $F_{u_1}(G,\bs) = c_F \cdot x$ comes directly from Source Node (Lemma~\ref{lemma:source-node}).
Furthermore, from PageRank recursive equation~\eqref{eq:pr:main} we have that $PR_{u_1}(G,\bs) = x$ ($u_1$ is a node with no incoming edges and weight $x$) and $PR_{v_i}(G,\bs) = a_F \cdot PR_{u_1}(G,\bs)/k = a_F \cdot x/k$ for every $i \in \{1,\dots,k\}$ ($v_i$ has a zero weight and one incoming edge from node $u$ which has $\deg^+_{u_1}(G) = k$).
This shows that $F(G,\bs) = PR(G,\bs)$ and concludes the proof.
\end{proof}


Now, let us turn our attention to more complex graphs.
In the following lemma we consider an arbitrary graph with no cycles and prove that the centralities are equal to PageRank with decay factor $a_F$ multiplied by $c_F$ (recall that these constants where defined in Lemmas~\ref{lemma:source-node} and \ref{lemma:1-arrow-af}). 

\begin{lemma}\label{lemma:no-cycles}
If a centrality measure $F$ satisfies Node Deletion, Edge Deletion, Edge Multiplication, Edge Swap, and Node Redirect, then for every graph with no cycles $(G,\bs)$ it holds that $F(G,\bs) = c_F \cdot PR^{a_F}(G,\bs)$.
\end{lemma}
\begin{proof}
We will use the induction on the number of predecessors of a node in a graph.
If $P_v(G) = \emptyset$, then node $v$ is a source in $G$, hence from Lemma~\ref{lemma:source-node} (Source Node) we have $F_v(G,\bs) = c_F \cdot \bs(v) = c_F \cdot PR^{a_F}_v(G,\bs)$.

Take a graph, $(G,\bs)$, and a node, $v$, with non-empty set of predecessors: $P_v(G) \neq \emptyset$.
Let $u$ be an arbitrary predecessor of $v$.
Clearly, $P_u(G) \subseteq P_v(G)$, but since graph has no cycles also $u \not \in P_u(G)$.
This implies that $u$ has less predecessor than $v$, so from the inductive assumption, we get that $F_u(G,\bs) = c_F \cdot PR_u^{a_F}(G,\bs)$ for every $u \in P_v(G)$.
Thus, to show that $F_v(G,\bs) = c_F \cdot PR_v^{a_F}(G,\bs)$, based on PageRank recursive equation \eqref{eq:pr:main} it is enough to prove that:
\begin{equation}\label{eq:no-cycles:1}
F_v(G,\bs) = a_F \cdot \left( \sum_{(u,v) \in \Gamma_v^-(G)} \frac{F_u(G,\bs)}{\deg^+_u(G)}\right) + c_F \cdot \bs(v).
\end{equation}
Note that from Edge Deletion and the fact that $v$ is not a predecessor of itself nor its direct predecessors, we know that outgoing edges of node $v$ does not affect $F_v(G)$, $\Gamma_v^-(G)$, nor $F_u(G)$ and $\deg_u^+(G)$ for every $(u,v) \in \Gamma_v^-(G)$; hence, they do not affect equation~\eqref{eq:no-cycles:1}.
Consequently, in what follows, we will assume that $v$ has no outgoing edges, i.e., is a sink.

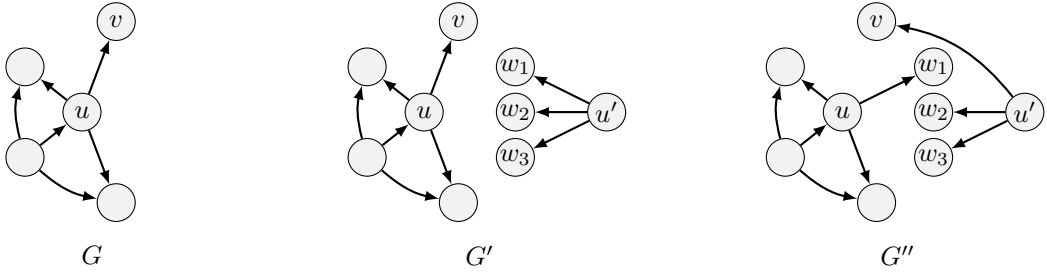
\begin{figure}[t]
\centering
\begin{tikzpicture}
  \def\x{0.6cm} 
  \def\y{0cm} 
  \def\arrdist{0.3cm}

  \tikzset{
    node_blank/.style={circle,draw,minimum size=0.5cm,inner sep=0, color=white}, 
    node/.style={circle,draw,minimum size=0.5cm,inner sep=0, fill = black!05}, 
    edge/.style={sloped,-latex,above,font=\footnotesize},
    el/.style={below,font=\footnotesize},
    operation/.style={sloped,>=stealth,above,font=\footnotesize},
    arrow/.style={draw, single arrow, minimum width = 0.9cm, minimum height=\y-6*\x+\s, fill=black!10},
    blank/.style={}
  } 
  
  \node[node] (A_1) at (\y+2*\x, 4*\x) {$v$}; 
  \node[node] (A_4) at (\y+2*\x, 0*\x) {};
  \node[node] (A_5) at (\y+0*\x, 1*\x) {}; 
  \node[node] (A_6) at (\y+0*\x, 3*\x) {}; 
  \node[node] (A_7) at (\y+1.25*\x, 2*\x) {$u$};
  \node[blank] (A_) at (\y+1.5*\x-0.02cm, -0.7cm) {$G$};
  \node[node_blank] (A_start) at (\y - \arrdist, 2*\x) {};
  \node[node_blank] (A_end) at (\y+4*\x + \arrdist, 2*\x) {};

  \path[->,draw,thick]
  (A_5) edge[edge, bend right=15]  (A_4)
  (A_5) edge[edge, bend left=15]  (A_6)
  (A_5) edge[edge]  (A_7)
  (A_7) edge[edge]  (A_1)
  (A_7) edge[edge]  (A_6)
  (A_7) edge[edge]  (A_4)
  ;
  
  \def\y{4.5cm} 
  
  \node[node] (A_1) at (\y+2*\x, 4*\x) {$v$}; 
  \node[node] (A_4) at (\y+2*\x, 0*\x) {};
  \node[node] (A_5) at (\y+0*\x, 1*\x) {}; 
  \node[node] (A_6) at (\y+0*\x, 3*\x) {}; 
  \node[node] (A_7) at (\y+1.25*\x, 2*\x) {$u$};
  \node[blank] (A_) at (\y+2.5*\x-0.02cm, -0.7cm) {$G'$};
  \node[node_blank] (A_start) at (\y - \arrdist, 2*\x) {};
  \node[node_blank] (A_end) at (\y+4*\x + \arrdist, 2*\x) {}; 
  \node[node] (A_u') at (\y+5.25*\x, 2*\x) {$u'$};
  \node[node] (A_w1) at (\y+3.25*\x, 3*\x) {$w_1$};
  \node[node] (A_w2) at (\y+3.25*\x, 2*\x) {$w_2$};
  \node[node] (A_w3) at (\y+3.25*\x, 1*\x) {$w_3$};

  \path[->,draw,thick]
  (A_5) edge[edge, bend right=15]  (A_4)
  (A_5) edge[edge, bend left=15]  (A_6)
  (A_5) edge[edge]  (A_7)
  (A_7) edge[edge]  (A_1)
  (A_7) edge[edge]  (A_6)
  (A_7) edge[edge]  (A_4)
  (A_u') edge[edge]  (A_w1)
  (A_u') edge[edge]  (A_w2)
  (A_u') edge[edge]  (A_w3)
  ;
  
  \def\y{10cm} 
  
  \node[node] (A_1) at (\y+2*\x, 4*\x) {$v$}; 
  \node[node] (A_4) at (\y+2*\x, 0*\x) {};
  \node[node] (A_5) at (\y+0*\x, 1*\x) {}; 
  \node[node] (A_6) at (\y+0*\x, 3*\x) {}; 
  \node[node] (A_7) at (\y+1.25*\x, 2*\x) {$u$};
  \node[blank] (A_) at (\y+2.5*\x-0.02cm, -0.7cm) {$G''$};
  \node[node_blank] (A_start) at (\y - \arrdist, 2*\x) {};
  \node[node_blank] (A_end) at (\y+4*\x + \arrdist, 2*\x) {}; 
  \node[node] (A_u') at (\y+5.25*\x, 2*\x) {$u'$};
  \node[node] (A_w1) at (\y+3.25*\x, 3*\x) {$w_1$};
  \node[node] (A_w2) at (\y+3.25*\x, 2*\x) {$w_2$};
  \node[node] (A_w3) at (\y+3.25*\x, 1*\x) {$w_3$};

  \path[->,draw,thick]
  (A_5) edge[edge, bend right=15]  (A_4)
  (A_5) edge[edge, bend left=15]  (A_6)
  (A_5) edge[edge]  (A_7)
  (A_u') edge[edge, bend right=20]  (A_1)
  (A_7) edge[edge]  (A_6)
  (A_7) edge[edge]  (A_4)
  (A_7) edge[edge]  (A_w1)
  (A_u') edge[edge]  (A_w2)
  (A_u') edge[edge]  (A_w3)
  ;
  
\end{tikzpicture}
\caption{Example graphs illustrating the first part of the proof of Lemma~\ref{lemma:no-cycles} for $k=3$.}
\label{figure:lemma:no-cycles:1}
\end{figure}

First, let us assume that $v$ has a zero weight $\bs(v) = 0$ and only one incoming edge: $(u,v)$ (see Fig.~\ref{figure:lemma:no-cycles:1} for illustration).
Let us denote the number of outgoing edges of $u$ by $k$ and its PageRank by $x$, i.e., $k = \deg_u^+(G)$ and $x = PR_{u}^{a_F}(G,\bs)$.
From the inductive assumption we know that $F_u(G,\bs) = c_F \cdot x$.
Hence, to prove that equation~\eqref{eq:no-cycles:1} holds, we need to show that $F_v(G,\bs) = a_F \cdot c_F \cdot x/k$.
This, combined with Lemma~\ref{lemma:k-arrow} is equivalent to proving that the centrality of $v$ is equal to the centrality of a sink in a $k$-out star in which the source has weight $x$.
To prove this, consider adding such a graph to $(G,\bs)$:
\[ (G',\bs') = ((V',E'), \bs') = (G,\bs) + ((\{u',w_1,\dots,w_k\},\lBrace(u',w_1),\dots,(u',w_k)\rBrace), [x,0,\dots,0])). \]
From Locality (Lemma~\ref{lemma:locality}), we know that the centrality of $v$ did not change: $F_v(G,\bs) = F_v(G',\bs')$.
Let us turn our attention to node $u$ and the source of a $k$-out star: $u'$.
For $u$, from Locality (Lemma~\ref{lemma:locality}) we know that $F_u(G',\bs') = F_u(G,\bs) = c_F \cdot x$.
For $u'$, from Source Node (Lemma~\ref{lemma:source-node}) we have that $F_{u'}(G',\bs') = c_F \cdot x$.
Hence, $u$ and $u'$ have equal centralities and equal numbers of outgoing edges.
As a result, from Edge Swap we know that we can replace edges $(u,v)$ and $(u',w_1)$ with edges $(u,w_1)$ and $(u',v)$ and centralities in the graph will not change.
Such a swap results in a graph $(G'',\bs'') = ((V', E' - \lBrace (u,v), (u',w_1) \rBrace \sqcup \lBrace (u,w_1), (u',v) \rBrace), \bs')$ in which $v$ is a sink in a component of the graph which is a $k$-out star (see Fig.~\ref{figure:lemma:no-cycles:1}).
In this $k$-out star the source has weight $x$, hence node $v$ has the centrality $a_F \cdot c_F \cdot x/k$.
Formally, we proved:
\begin{equation}\label{eq:no-cycles:2}
F_v(G,\bs) = F_v(G',\bs') = F_v(G'',\bs'') = a_F \cdot c_F \cdot x/k = a_F \cdot \frac{F_u(G,\bs)}{\deg^+_u(G)},
\end{equation}
where the consecutive equalities comes from Locality (Lemma~\ref{lemma:locality}), Edge Swap and Lemma~\ref{lemma:k-arrow} combined with Locality and the definition of constants $x$ and $k$.

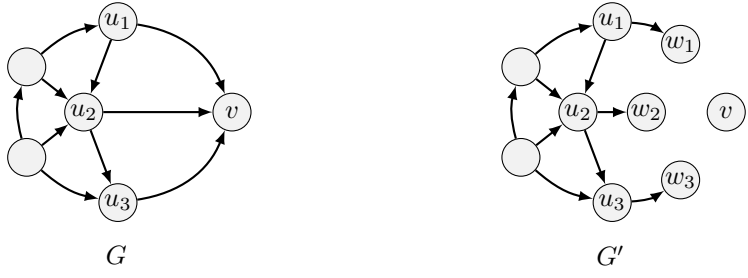
\begin{figure}[t]
\centering
\begin{tikzpicture}
  \def\x{0.6cm} 
  \def\y{0cm} 
  \def\arrdist{0.3cm}

  \tikzset{
    node_blank/.style={circle,draw,minimum size=0.5cm,inner sep=0, color=white}, 
    node/.style={circle,draw,minimum size=0.5cm,inner sep=0, fill = black!05}, 
    edge/.style={sloped,-latex,above,font=\footnotesize},
    el/.style={below,font=\footnotesize},
    operation/.style={sloped,>=stealth,above,font=\footnotesize},
    arrow/.style={draw, single arrow, minimum width = 0.9cm, minimum height=\y-6*\x+\s, fill=black!10},
    blank/.style={}
  } 
  
  \node[node] (u0) at (\y+0*\x, 1*\x) {};
  \node[node] (u_) at (\y+0*\x, 3*\x) {};
  \node[node] (u1) at (\y+2*\x, 4*\x) {$u_1$};
  \node[node] (u2) at (\y+1.25*\x, 2*\x) {$u_2$}; 
  \node[node] (u3) at (\y+2*\x, 0*\x) {$u_3$};
  \node[node] (v1) at (\y+4.5*\x, 2*\x) {$v$}; 
  \node[blank] (A_) at (\y+2*\x-0.02cm, -0.7cm) {$G$};

  \path[->,draw,thick]
  (u_) edge[edge, bend left=15]  (u1)
  (u_) edge[edge] (u2)
  (u0) edge[edge, bend left=15]  (u_)
  (u0) edge[edge]  (u2)
  (u0) edge[edge, bend right=15]  (u3)
  (u1) edge[edge]  (u2)
  (u2) edge[edge]  (u3)
  (u1) edge[edge, bend left=30]  (v1)
  (u2) edge[edge]  (v1)
  (u3) edge[edge, bend right=30]  (v1)
  ;
  
  \def\y{6.5cm} 
  
  \node[node] (u0) at (\y+0*\x, 1*\x) {};
  \node[node] (u_) at (\y+0*\x, 3*\x) {};
  \node[node] (u1) at (\y+2*\x, 4*\x) {$u_1$};
  \node[node] (u2) at (\y+1.25*\x, 2*\x) {$u_2$}; 
  \node[node] (u3) at (\y+2*\x, 0*\x) {$u_3$};
  \node[node] (v1) at (\y+4.5*\x, 2*\x) {$v$}; 
  \node[blank] (A_) at (\y+2*\x-0.02cm, -0.7cm) {$G'$};
  \node[node] (w1) at (\y+3.5*\x, 3.5*\x) {$w_1$};
  \node[node] (w2) at (\y+2.75*\x, 2*\x) {$w_2$}; 
  \node[node] (w3) at (\y+3.5*\x, 0.5*\x) {$w_3$};

  \path[->,draw,thick]
  (u_) edge[edge, bend left=15]  (u1)
  (u_) edge[edge] (u2)
  (u0) edge[edge, bend left=15]  (u_)
  (u0) edge[edge]  (u2)
  (u0) edge[edge, bend right=15]  (u3)
  (u1) edge[edge]  (u2)
  (u2) edge[edge]  (u3)
  (u1) edge[edge, bend left=15]  (w1)
  (u2) edge[edge]  (w2)
  (u3) edge[edge, bend right=15]  (w3)
  ;
  
\end{tikzpicture}
\caption{Example graphs illustrating the second part of the proof of Lemma~\ref{lemma:no-cycles} for $m=3$.}
\label{figure:lemma:no-cycles:2}
\end{figure}

Now, assume that $v$ has $m$ ($m \ge 1$) incoming edges and possibly non-zero weight.
In such a case let us split node $v$ into $m+1$ separate nodes, one with the original weight of $v$ and no incoming edges and $m$ nodes, each with zero weight and one incoming edge (see Fig.~\ref{figure:lemma:no-cycles:2} for illustration).
Formally, assume $\Gamma_v^-(G) = \lBrace (u_1,v), \dots, (u_m, v)\}$ (note that $u_i$ may not be pairwise different) and consider adding nodes $w_1,\dots,w_m$ to the graph and replacing edges $\Gamma_v^-(G)$ with $(u_1,w_1), \dots, (u_m,w_m)$. 
Let $(G',\bs')$ be the resulting graph:
\[ (G',\bs') = ((V \cup \{w_1,\dots,w_m\}, E - \Gamma_v^-(G) \sqcup \lBrace (u_1,w_1), \dots, (u_m, w_m) \rBrace), \bs'), \]
where $\bs'(u)=\bs(u)$ for every $u \in V$ and $\bs'(u) = 0$, otherwise.
Then, we clearly have that $(G,\bs) = R_{w_m \rightarrow v}\left(\dots\left(R_{w_1 \rightarrow v}(G',\bs')\right)\right)$. 
Moreover, node $v$ and nodes $w_1,\dots,w_m$ are out-twins in $(G',\bs')$ (they are all sinks).
Hence, from Node Redirect we get that the centrality of $v$ is the sum of centralities of $v$ and $w_1,\dots,w_m$ in $(G',\bs')$:
\begin{equation}\label{eq:no-cycles:3}
F_v(G,\bs) = F_v(G',\bs') + F_{w_1}(G',\bs') + \dots + F_{w_m}(G',\bs').
\end{equation}
Furthermore, from Node Redirect we know that centralities of nodes other than $v$ did not change: $F_u(G,\bs) = F_u(G',\bs')$.
Since the out-degrees of these nodes did not change either, from our analysis of nodes with a single edge and equation~\eqref{eq:no-cycles:2}, in particular, we get that:
\begin{equation}\label{eq:no-cycles:4}
\quad F_{w_i}(G',\bs') = a_F \cdot \frac{F_{u_i}(G',\bs')}{\deg_{u_i}^+(G')} = a_F \cdot \frac{F_{u_i}(G,\bs)}{\deg_{u_i}^+(G)}.
\end{equation}
Finally, Source Node (Lemma~\ref{lemma:source-node}) implies that $F_v(G',\bs') = c_F \cdot \bs(v)$ (recall that $v$ has no incoming edges in $(G',\bs')$).
This combined with Equation~\eqref{eq:no-cycles:3} and \eqref{eq:no-cycles:4} proves Equation~\eqref{eq:no-cycles:1}.
\end{proof}


We are now ready to prove that in every graph centralities are equal to PageRank with the decay factor $a_F$ multiplied by $c_F$.

\begin{lemma}\label{lemma:cycles}
If a centrality measure $F$ satisfies Node Deletion, Edge Deletion, Edge Multiplication, Edge Swap, and Node Redirect, then for every graph $(G,\bs)$ it holds that $F(G,\bs) = c_F \cdot PR^{a_F}(G,\bs)$.
\end{lemma}
\begin{proof}
Take an arbitrary graph $G = (V,E)$ and weights $\bs$.
We will prove the thesis by induction on the number of cycles in $G$.
If there are no cycles in graph $G$, then the thesis follows from Lemma~\ref{lemma:no-cycles}.

Assume otherwise.
Fix node $w$ that belongs to at least one cycle and let $x_w$ be its PageRank in $(G,\bs)$: $x_w = PR^{a_F}_{w}(G, \bs)$.
Consider graph $(G',\bs')$ obtained from $(G,\bs)$ by adding two nodes, $s$ with weight $x_w$ and $t$ with weight $0$, and splitting each outgoing edge of $w$, $(w,w')$, into two edges: $(w,t)$ and $(s,w')$ (see Fig.~\ref{figure:lemma:cycles} for illustration). 
Formally:
\[ (G',\bs') = ((V \cup \{s,t\}, E - \Gamma^+_{w}(G) \sqcup \lBrace (w,t), (s,w') : (w,w') \in \Gamma^+_{w}(G) \rBrace), \bs'), \]
with $\bs'(v) = \bs(v)$ for every $v \in V$, $\bs'(s) = x_w$ and $\bs'(t) = 0$.
Observe that graph $(G',\bs')$ has less cycles than graph $(G,\bs)$: every cycle in the former graph is also a cycle in the later one, but the former graph does not contain cycles with $w$.
Hence, from the inductive assumption we know that
\begin{equation}\label{eq:cycles:1}
F_v(G',\bs') = c_F \cdot PR^{a_F}_v(G',\bs') \quad \mbox{ for every } v \in V.
\end{equation}
Thus, to show that $F(G,\bs) = c_F \cdot PR^{a_F}(G,\bs)$ we will prove that for every node $v \in V$ we have: (1) $PR^{a_F}_v(G,\bs) = PR^{a_F}_v(G',\bs')$ and (2) $F_v(G,\bs) = F_v(G',\bs')$.
We will prove the two equalities separately.

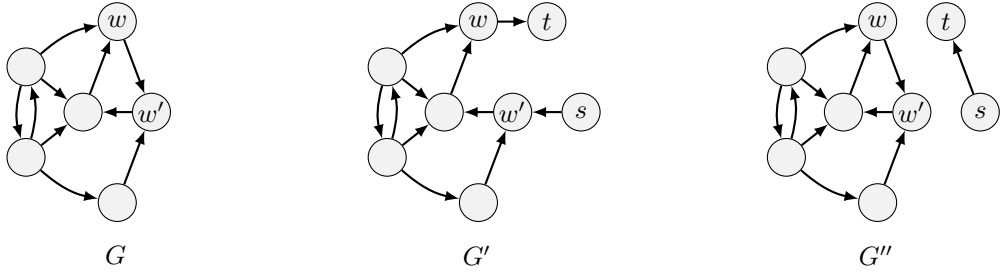
\begin{figure}[t]
\centering
\begin{tikzpicture}
  \def\x{0.6cm} 
  \def\y{0cm} 
  \def\arrdist{0.3cm}

  \tikzset{
    node_blank/.style={circle,draw,minimum size=0.5cm,inner sep=0, color=white}, 
    node/.style={circle,draw,minimum size=0.5cm,inner sep=0, fill = black!05}, 
    edge/.style={sloped,-latex,above,font=\footnotesize}, 
    el/.style={below,font=\footnotesize}, 
    operation/.style={sloped,>=stealth,above,font=\footnotesize},
    arrow/.style={draw, single arrow, minimum width = 0.9cm, minimum height=\y-6*\x+\s, fill=black!10},
    blank/.style={}
  } 
  \node[node] (A_1) at (\y+2*\x, 4*\x) {$w$}; 
  \node[node] (A_4) at (\y+2*\x, 0*\x) {};
  \node[node] (A_5) at (\y+0*\x, 1*\x) {}; 
  \node[node] (A_6) at (\y+0*\x, 3*\x) {}; 
  \node[node] (A_7) at (\y+1.25*\x, 2*\x) {}; 
  \node[node] (A_8) at (\y+2.75*\x, 2*\x) {$w'$};
  \node[blank] (A_) at (\y+2*\x-0.02cm, -0.7cm) {$G$};
  \node[node_blank] (A_start) at (\y - \arrdist, 2*\x) {};
  \node[node_blank] (A_end) at (\y+4*\x + \arrdist, 2*\x) {};

  \path[->,draw,thick]
  (A_1) edge[edge]  (A_8)
  (A_4) edge[edge]  (A_8)
  (A_5) edge[edge, bend right=15]  (A_4)
  (A_5) edge[edge, bend right=15]  (A_6)
  (A_5) edge[edge]  (A_7)
  (A_6) edge[edge, bend right=15]  (A_5)
  (A_6) edge[edge, bend left=15]  (A_1)
  (A_6) edge[edge]  (A_7)
  (A_7) edge[edge]  (A_1)
  (A_8) edge[edge]  (A_7)
  ;
  
  \def\y{4.75cm} 
  
  \node[node] (A_1) at (\y+2*\x, 4*\x) {$w$}; 
  \node[node] (A_4) at (\y+2*\x, 0*\x) {};
  \node[node] (A_5) at (\y+0*\x, 1*\x) {}; 
  \node[node] (A_6) at (\y+0*\x, 3*\x) {}; 
  \node[node] (A_7) at (\y+1.25*\x, 2*\x) {}; 
  \node[node] (A_8) at (\y+2.75*\x, 2*\x) {$w'$};
  \node[blank] (A_) at (\y+2*\x-0.02cm, -0.7cm) {$G'$};
  \node[node_blank] (B_start) at (\y - \arrdist, 2*\x) {};
  \node[node_blank] (B_end) at (\y+4*\x + \arrdist, 2*\x) {};
  \node[node] (A_t) at (\y+3.5*\x, 4*\x) {$t$}; 
  \node[node] (A_s) at (\y+4.25*\x, 2*\x) {$s$};

  \path[->,draw,thick]
  (A_1) edge[edge]  (A_t)
  (A_s) edge[edge]  (A_8)
  (A_4) edge[edge]  (A_8)
  (A_5) edge[edge, bend right=15]  (A_4)
  (A_5) edge[edge, bend right=15]  (A_6)
  (A_5) edge[edge]  (A_7)
  (A_6) edge[edge, bend right=15]  (A_5)
  (A_6) edge[edge, bend left=15]  (A_1)
  (A_6) edge[edge]  (A_7)
  (A_7) edge[edge]  (A_1)
  (A_8) edge[edge]  (A_7)
  ;
  
  \def\y{10cm} 
  
  \node[node] (A_1) at (\y+2*\x, 4*\x) {$w$}; 
  \node[node] (A_4) at (\y+2*\x, 0*\x) {};
  \node[node] (A_5) at (\y+0*\x, 1*\x) {}; 
  \node[node] (A_6) at (\y+0*\x, 3*\x) {}; 
  \node[node] (A_7) at (\y+1.25*\x, 2*\x) {}; 
  \node[node] (A_8) at (\y+2.75*\x, 2*\x) {$w'$};
  \node[blank] (A_) at (\y+2*\x-0.02cm, -0.7cm) {$G''$};
  \node[node_blank] (B_start) at (\y - \arrdist, 2*\x) {};
  \node[node_blank] (B_end) at (\y+4*\x + \arrdist, 2*\x) {};
  \node[node] (A_t) at (\y+3.5*\x, 4*\x) {$t$}; 
  \node[node] (A_s) at (\y+4.25*\x, 2*\x) {$s$};

  \path[->,draw,thick]
  (A_1) edge[edge]  (A_8)
  (A_4) edge[edge]  (A_8)
  (A_5) edge[edge, bend right=15]  (A_4)
  (A_5) edge[edge, bend right=15]  (A_6)
  (A_5) edge[edge]  (A_7)
  (A_6) edge[edge, bend right=15]  (A_5)
  (A_6) edge[edge, bend left=15]  (A_1)
  (A_6) edge[edge]  (A_7)
  (A_7) edge[edge]  (A_1)
  (A_8) edge[edge]  (A_7)
  (A_s) edge[edge]  (A_t)
  ;

\end{tikzpicture}
\caption{Example graphs illustrating the proof of Lemma~\ref{lemma:cycles}.}
\label{figure:lemma:cycles}
\end{figure}

(1): 
We will show that $PR^{a_F}_v(G,\bs) = PR^{a_F}_v(G',\bs')$ holds for every $v \in V$.
To this end, let us define $x_v = PR^{a_F}_v(G,\bs)$ for every $v \in V$ (note that this is consistent with $x_w = PR^{a_F}_w(G,\bs)$), $x_s = x_w$ and $x_t = a_F \cdot x_w$. 
We will prove that $(x_v)_{v \in V \cup \{s,t\}}$ satisfies the system of PageRank recursive equations~\eqref{eq:pr:main} for graph $(G',\bs')$, i.e., for every $v \in V \cup \{s,t\}$ we have:
\begin{equation}\label{eq:cycles:2}
x_v = a_F \cdot \left( \sum_{(u,v) \in \Gamma^-_v(G')} \frac{x_u}{\deg^+_u(G')} \right) + \bs'(v).
\end{equation}
Since this system of equations has exactly one solution which is $PR^{a_F}(G',\bs')$, in this way we will show that $PR^{a_F}_v(G',\bs') = x_v = PR^{a_F}_v(G,\bs)$ for every $v \in V$.

First, consider $v \in \{s,t\}$.
Since node $s$ has no incoming edges, equation~\eqref{eq:cycles:2} for $v=s$ simplifies to $x_s = \bs'(s) = x_w$ which is equivalent to the definition of $x_s$.
Node $t$ has a zero weight and incoming edges only from $w$; on the other hand, all outgoing edges of $w$ goes to $t$ which implies $|\Gamma_t^-(G')| = \deg_w^+(G')$. Hence, equation~\eqref{eq:cycles:2} for $v=t$ simplifies to $x_t = a_F \cdot x_w$ which is equivalent to   the definition of $x_t$.

Consider $v \in V$.
From PageRank recursive equation for graph $(G,\bs)$ we have:
\begin{equation}\label{eq:cycles:3}
x_v = a_F \cdot \left( \sum_{(u,v) \in \Gamma^-_v(G)} \frac{x_u}{\deg^+_u(G)} \right) + \bs(v).
\end{equation}
By the construction of $(G',\bs')$ we know that out-degree and weight of every node $u \in V$ is the same in $(G,\bs)$ as in $(G',\bs')$.
Moreover, if $v \in V$ is not a direct successor of $w$ in $(G,\bs)$, then it has the same set of incoming edges in both graphs.
Hence, replacing in equation~\eqref{eq:cycles:3} $\deg^+_u(G), \bs(v), \Gamma^-_v(G)$ with $\deg^+_u(G'), \bs'(v), \Gamma^-_v(G')$ proves equation~\eqref{eq:cycles:2}. 

Now, assume $v \in S_w^1(G)$.
From PageRank recursive equation~\eqref{eq:pr:main} for $(G,\bs)$ by separating edges from $w$ we get:
\begin{equation}\label{eq:cycles:4} 
x_v = a_F \cdot \left( \sum_{(u,v) \in \Gamma^-_v(G), u \neq w} \frac{x_u}{\deg^+_u(G)} + \sum_{(w,v) \in \Gamma^-_v(G)} \frac{x_w}{\deg^+_w(G)}\right) + \bs(v).
\end{equation}
Note that sets of incoming edges of node $v$ in graphs $G$ and $G'$ are equal except for edges from $w$ and $s$: all edges from $w$ in $G$ are replaced by the edges from $s$ in $G'$.
However, we know that $w$ in $G$ has the same out-degree as $s$ in $G'$: $\deg_w^+(G) = \deg_s^+(G')$.
Also, by the definition of $x_s$, we have $x_s = x_w$. 
Hence, by transforming equation~\eqref{eq:cycles:4} we get:
\[
x_v = a_F \cdot \left( \sum_{(u,v) \in \Gamma^-_v(G'), u \neq s} \frac{x_u}{\deg^+_u(G')} + \sum_{(s,v) \in \Gamma^-_v(G')} \frac{x_s} {\deg^+_s(G')}\right) + \bs'(v),
\]
which proves equation~\eqref{eq:cycles:2}. 
This concludes the proof that $PR^{a_F}_v(G,\bs) = PR^{a_F}_v(G',\bs')$ for every $v \in V$.
Note that we obtained that nodes $s$ and $w$ have equal PageRank in $(G',\bs')$:
\begin{equation}\label{eq:cycles:5}
PR^{a_F}_s(G',\bs') = x_s = x_w = PR^{a_F}_w(G,\bs) = PR^{a_F}_w(G',\bs').
\end{equation}
We will use this in the second part of the proof.

(2): 
It remains to prove that $F_v(G,\bs) = F_v(G',\bs')$ for every $v \in V$.
To this end, observe that nodes $s$ and $w$ have equal centralities in $(G',\bs')$: from equation~\eqref{eq:cycles:5} and the inductive assumption we have $F_s(G',\bs') = F_w(G',\bs')$.
Moreover, both nodes have equal out-degrees. 
Hence, from Edge Swap we know that swapping ends of outgoing edges of $w$ and $s$ does not affect any centralities in the graph. 
By sequentially swapping the ends of all outgoing edges of $w$ with the ends of the outgoing edges of $s$ we get the following graph (see Fig.~\ref{figure:lemma:cycles}):
\[ (G'',\bs'') = (G,\bs) + ((\{s,t\}, \deg^+_w(G) \cdot \lBrace (s,t) \rBrace ), [x_w, 0]), \]
and from Edge Swap we have $F_v(G'',\bs'') = F_v(G',\bs')$ for every $v \in V$.
On the other hand, from Locality we have $F_v(G'',\bs'') = F_v(G,\bs)$.
Therefore, $F_v(G,\bs) = F_v(G',\bs')$ for every $v \in V$ which concludes the proof.
\end{proof}


So far, we have proved that if $F$ satisfies Node Deletion, Edge Deletion, Edge Multiplication, Edge Swap and Node Redirect, then $F(G,\bs) = c_F \cdot PR^{a_F}(G,\bs)$ for every graph $(G,\bs)$.
In the last lemma of Part 1 of the proof, we show that if $F$ also satisfies Baseline, then $c_F = 1$; hence, $F(G,\bs) = PR^{a_F}(G,\bs)$.

\begin{lemma}\label{lemma:baseline}
If centrality measure $F$ satisfies Node Deletion, Edge Deletion, Edge Multiplication, Edge Swap, Node Redirect and Baseline, then for every graph $(G, \bs)$ it holds that $F(G, \bs) = PR^{a_F}(G, \bs)$.
\end{lemma}
\begin{proof}
If centrality measure $F$ satisfies Node Deletion, Edge Deletion, Edge Multiplication, Edge Swap, and Node Redirect, then from Lemma~\ref{lemma:cycles} we know that $F_v(G,\bs) = c_F \cdot PR^{a_F}_v(G,\bs)$ for $c_F$ and $a_F$ defined as follows: $c_F = F_w((\{w\}, \emptyset), [1])$ and $a_F = F_{w}((\{w',w\}, \lBrace (w',w) \rBrace), [1,0])/c_F$ if $c_F > 0$ and $a_F = 0$, otherwise, for arbitrary nodes $w,w'$.
Now, from Baseline we have that $F_w((\{w\},\emptyset),[1]) = 1$, which implies $c_F = 1$ and concludes the proof.
\end{proof}


\subsection{Independence of axioms}\label{section:proof-3}

In this section, we show that all six axioms used in our characterization of PageRank are necessary, i.e., if we remove one of them, then the remaining axioms will be satisfied also by some centrality other than PageRank.
In other words, we will show that no axiom is implied by the others.
To this end, in the following theorem for each axiom we show that there exists a centrality, other than PageRank, that satisfies all other axioms.

\begin{theorem}\label{theorem:independence} (Independence of Axioms)
From six axioms: Node Deletion, Edge Deletion, Edge Multiplication, Edge Swap, Node Redirect, and Baseline, none is implied by a combination of five others.
Specifically:
\begin{itemize}
\item A centrality measure $F$ defined for every graph $(G,\bs)$ and node $v$ as follows: 
\[
F_v(G,\bs) = PR^{a(G,\bs)}_v(G,\bs), \quad \quad \mbox{ where } a(G,\bs) = 1/(2 + \bs(G))
\]
satisfies Edge Deletion, Edge Multiplication, Edge Swap, Node Redirect and Baseline, but does not satisfy Node~Deletion.
\item A centrality measure $F$ defined for every graph $(G,\bs)$ and node $v$ and an arbitrary $a \in (0,1)$ as follows: 
\[
F^a_v(G,\bs) = 
\begin{cases}
    2 \cdot PR^a_v(G,\bs) - \bs(v) & \mbox{if } v \mbox{ is a sink,}\\
    PR^a_v(G,\bs) & \mbox{otherwise}
\end{cases}
\]
satisfies Node Deletion, Edge Multiplication, Edge Swap, Node Redirect and Baseline, but does not satisfy Edge~Deletion.
\item A centrality measure $F$ defined for every graph $(G,\bs)$ and node $v$ and an arbitrary $a \in (0,1)$ as follows: 
\[
F^a_v(G,\bs) = a \cdot \left( \sum_{(u,v) \in \Gamma^-_v(G)} \frac{F^a_u(G,\bs)}{\deg_u^+(G) + 1} \right) + \bs(v)
\]
satisfies Node Deletion, Edge Deletion, Edge Swap, Node Redirect and Baseline, but does not satisfy Edge~Multiplication.
\item A centrality measure $F$ defined for every graph $(G,\bs)$ and node $v$ as follows: 
\[
F_v(G,\bs) = \sum_{(u,v) \in \Gamma^-_v(G)} \frac{\bs(u)}{\deg_u^+(G)} + \bs(v)
\]
satisfies Node Deletion, Edge Deletion, Edge Multiplication, Node Redirect and Baseline, but does not satisfy Edge~Swap.
\item A centrality measure $F$ defined for every graph $(G,\bs)$ and node $v$ as follows: 
\[
F_v(G,\bs) = \sum_{(u,v) \in \Gamma^-_v(G)} \frac{1}{\deg_u^+(G)} + \bs(v)
\]
satisfies Node Deletion, Edge Deletion, Edge Multiplication, Edge Swap and Baseline, but does not satisfy Node~Redirect.
\item A centrality measure $F$ defined for every graph $(G,\bs)$ and node $v$ and an arbitrary $a \in (0,1)$ as follows: 
\[
F^a_v(G,\bs) = 2 \cdot PR^a_v(G,\bs)
\]
satisfies Node Deletion, Edge Deletion, Edge Multiplication, Edge Swap and Node Redirect, but does not satisfy Baseline.
\end{itemize}
\end{theorem}

The proof of Theorem~\ref{theorem:independence} can be found in
\ref{section:appendix:independence}.


\section{Comparison with other centrality measures}\label{section:comparison}
In this section, using the axioms from Section~\ref{section:axioms}, we discuss how PageRank differs from other standard centralities.
We begin by discussing \emph{feedback centralities}: a class of centrality measures which---like PageRank---assess the importance of a node based on the number and importance of its direct predecessors.
Then we focus on other classic centralities from the literature.

Most centrality measures do not take into account node weights. However, for consistency with the domain of our paper, we will define them as functions that take a graph with node weights as an input.

For directed graphs, most centralities can be defined in two ways: by focusing on predecessors or, more generally, paths ending in a given node or by focusing on successors and paths starting in a given node.
For consistency with our interpretation of the direction of edges and the definition of PageRank, we will assume the former version.

Several centrality measures that we discuss are defined only on a subset of all graphs, e.g., only on strongly connected graphs.
In effect, we cannot apply our axioms directly, as they concern graphs on which centralities are not defined.
To cope with this problem, we consider a weaker version of axioms restricted to a specific class of graphs. 
To give an example, Edge Swap restricted to strongly connected graphs states that if a graph is strongly connected and the graph resulting from an edge swap is also strongly connected, then centralities in both graphs are equal (the formal definition can be found in~\ref{section:appendix:other_centralities}).
Such a restriction is always possible, although it sometimes results in a tautology satisfied by every centrality measure.
This is the case, for example, with Node Deletion restricted to strongly connected graphs: since there is no graph with an isolated node in a class of strongly connected graphs, the axiom is trivially satisfied by every possible centrality measure.

\subsection{Feedback centralities}
Feedback centralities~\cite{Brandes:Erlebach:2005} assess importance of a node based on the number and importance of its direct predecessors. 
It constitutes a large and widely used class of centrality measures. 
In this section, we introduce classic centralities that belong to this class.
We will split feedback centralities into two groups: \emph{parallel} and \emph{distributed} (similar division can be found, e.g., in \cite{Newman:2010}).

\subsubsection{Parallel feedback centralities}
The simplest, yet one of the most popular centrality measures, is \emph{degree centrality}~\cite{Freeman:1977}.
It assesses a node simply by looking at the number of its incoming edges.
Using our notation, the degree centrality is defined as follows:
\begin{equation*}
D_v(G,\bs) = |\Gamma^-_v(G)|.
\end{equation*}

A natural extension of degree centrality is \emph{eigenvector centrality}~\cite{Bonacich:1972}.
In degree centrality, an edge from each predecessor equally contributes to the importance of a node.
In contrary, in eigenvector centrality, an edge from a more important node contributes more. 
Specifically, an edge from a node is as important as the node itself.
Consequently, the importance of a node is proportional not to the number of direct predecessors, but to their total importance.
Formally, eigenvector centrality is defined for strongly connected graphs as a solution of the following recursive equation:
\begin{equation}\label{eq:c:eigenvector}
EV_v(G,\bs)= \frac{1}{\lambda(G)} \sum_{(u,v) \in \Gamma^-_v(G)} EV_u(G,\bs).
\end{equation}
Here, $\lambda(G)$ is the largest eigenvalue of the adjacency matrix of a graph, i.e., for the adjacency matrix $A = (a_{u,v})_{u,v \in V}$ of graph $G$ defined as $a_{u,v} = \#_{(v,u)}(G)$, $\lambda(G)$ is the largest real value $c \in \mathbb{R}$ such that $A \cdot x = c \cdot x$ for some $x \in \mathbb{R}^{V}$.
Such vector $x$, which is the corresponding eigenvector, is the unique solution of this system of recursive equations up to a scalar multiplication.
Usually an additional normalization condition is assumed to make a solution unique.
In this section, as often in the literature, we will assume that the sum of centralities of all nodes equals $1$.

Another centrality based on a similar principle is \emph{Katz centrality} \cite{Katz:1953}.
Here also the importance of a node is mostly determined by the total importance of its direct predecessors. 
However, an additional small basic importance is added to every node.
In our paper, the weight of a node represents such a basic importance, but if weights are not provided, a fixed positive constant is used.
Formally, Katz centrality with the decay factor $a \in (0,1)$ is defined for graphs with $\lambda(G) < 1/a$ as a unique function that satisfies the following recursive equation:\footnote{Using matrix form, equation~\eqref{eq:c:katz} can be written as follows: $K = a \cdot A^T \cdot K + b$, where $A$ is the adjacency matrix and $b = (b_v)_{v \in V}$ is the vector of node weights. Let $\mathbb{I}$ be the identity matrix. Then, we get: $(\mathbb{I} - a \cdot A^T) K = b$. Now, if $\lambda(G) < 1/a$, then matrix $(\mathbb{I} - a \cdot A^T)$ is invertible which gives: $K = (\mathbb{I} - a \cdot A^T)^{-1} b$. See \ref{section:appendix:proof-2} for an analogous derivation of PageRank.}
\begin{equation}\label{eq:c:katz}
K^a_v(G,\bs) = a \cdot \left( \sum_{(u,v) \in \Gamma^-_v(G)} K^a_u(G,\bs) \right) + \bs(v).
\end{equation}
Adding a basic importance shifts the emphasis from the total importance of the direct predecessors back to their number.
This is because each edge $(u,v)$ contributes to the centrality of $v$ an additional value $a \cdot \bs(u)$ which is independent of the position or the importance of a predecessor.
That is why Katz centrality is sometimes seen as a middle-ground between degree and eigenvector centralities.

A small modification of Katz centrality is \emph{Bonacich centrality} \cite{Bonacich:1987} (sometimes called \emph{Bonacich-Katz centrality} \cite{Bloch:etal:2019}).
Here, however, a small basic importance is added not to every node, but to every edge.
Formally, Bonacich centrality for the decay factor $a \in (0,1)$ is defined for graphs with $\lambda(G) < 1/a$ as a unique function that satisfies the following recursive equation:
\begin{equation}\label{eq:c:bonacich}
BK^a_v(G,\bs) = \sum_{(u,v) \in \Gamma^-_v(G)} \left(a \cdot BK^a_u(G,\bs) + \bs(u) \right).
\end{equation}
It is easy to verify, that $BK_v^a(G,\bs) = (K_v^a(G,\bs)-\bs(v))/a$.


The common property of degree, eigenvector, Katz and Bonacich centralities is that the benefit for $v$ from an edge $(u,v)$ depends solely on the importance and the weight of node $u$.
In particular, it does not depend on the number of outgoing edges of node $u$.
As a result, an important node with many outgoing edges has a significantly higher impact on centralities in the graph than a node with few edges.
Such an approach makes sense if a centrality measure is used to capture the process that spread along all possible connections at once, e.g., through parallel duplication like in the case of information spread among social media users or virus infection in the population.
Based on the work of~\citet{Borgatti:2005} we call such centralities \emph{parallel feedback centralities}.

All listed parallel feedback centralities satisfy Node Deletion and Edge Deletion as they assess the importance of a node based only on its predecessors.
Also, as clear from the above discussion, all of them violates Edge Multiplication---if we double the outgoing edges of a node its impact on the direct successors will be doubled.
Edge Swap is satisfied by degree, eigenvector and Katz centralities.
However, it is violated by Bonacich centrality as the importance of a node according to this measure depends not only on the importance of the predecessors, but also on their weights.
Node Redirect is violated only by degree centrality, since merging two out-twins decreases the (in-)degree of their direct successors.
Finally, degree and Bonacich centrality violates Baseline.
In turn, Baseline is satisfied by Katz centrality and trivially satisfied by eigenvector which is defined only for strongly connected graphs.

\subsubsection{Distributed feedback centralities}

In \emph{distributed feedback centralities}, the importance of a node is equally split between its direct successors. 
As a result, the impact of a node on its direct successors is bounded by its own importance.
Such an approach is more reasonable if a centrality measure is used to describe a process that follows only one connection at once, e.g., like a user on the World Wide Web, the money in the financial network or a reader following references in a scientific article.
As we will show, parallel feedback centralities correspond to distributed feedback centralities.

\emph{Beta measure}~\cite{Brink:Gilles:2000} was proposed as a measure of dominance in directed networks.
In accordance with our interpretation of the direction of edges, edge $(u,v)$ means that $v$ dominates $u$.
Beta measure looks at the incoming edges to a node, similarly to degree centrality.
However, an edge from a node with many other edges is considered less important: a dominance over a node dominated by many other nodes is less significant.
Formally, beta measure is defined as follows:
\begin{equation*}
\beta_v(G,\bs)= \sum_{(u,v) \in \Gamma^-_v(G)} \frac{1}{\deg^+_u(G)}.
\end{equation*}

\emph{Seeley index}~\cite{Seeley:1949} (also known as the \emph{simplified PageRank centrality}~\cite{Page:etal:1999},
\emph{Katz prestige}~\cite{Jackson:2005}, or \emph{invariant scoring function}~\cite{Slutzki:Volij:2006}) is based on a similar principle.
Beta measure can be interpreted in the following way: each node has one unit of importance which he splits equally among its direct successors.
Now, in Seeley index, each node splits his whole importance equally among its direct successors.
In this way, we obtain a formula similar to eigenvector centrality.
Formally, Seeley index is defined for strongly connected graphs as a solution of the following recursive equation:
\begin{equation}\label{eq:c:katz_prestige}
SI_v(G,\bs)= \sum_{(u,v) \in \Gamma^-_v(G)} \frac{SI_u(G,\bs)}{\deg^+_u(G)}.
\end{equation}
This system of recursive equations has a unique solution up to a scalar multiplication.
To make a solution unique, we will assume that the sum of centralities of all nodes equals $1$, as we did in the case of eigenvector centrality.

Finally, a distributed feedback centrality measure that is an equivalent to Katz and Bonacich centralities is PageRank.

Consider distributed feedback centralities from the perspective of our axioms.
All three centralities satisfy Node Deletion, Edge Deletion, Edge Multiplication and Edge Swap. 
The satisfiability of the first two axioms is straightforward.
For Edge Multiplication, it follows from the definition of distributed feedback centralities---what matters is the proportion of edges from $u$ to $v$ and not the absolute value.
Edge Swap is satisfied since all centralities assess the importance of a node based on the importance and out-degree of their predecessors.
Now, Node Redirect is violated only by beta measure, since merging two out-twins decreases the beta measure of their direct successors.
Finally, Baseline is violated by beta measure and trivially satisfied by Seeley index which is defined only for strongly connected graphs.

\subsection{Other classic centralities}
In this section, we introduce two classic centralities: \emph{closeness} and \emph{betweenness centralities} along with one popular alternative: \emph{decay centrality}.

\emph{Closeness centrality}~\cite{Bavelas:1950} aims to find nodes which are at the center of the graph.
To this end, for each node it sums the distances from all other nodes in the graph.
Now, nodes with a small total sum, i.e., nodes which are close to all other nodes, are considered most central.
Formally, closeness centrality is defined for strongly connected graphs as follows:
\begin{equation*}
C_v(G, \bs) = \frac{1}{\sum_{u \in V \setminus \{v\}} dist_{u,v}(G)}.
\end{equation*}


\emph{Decay centrality} \cite{Jackson:2005} is a modification of closeness centrality that works for arbitrary graphs, not necessarily strongly connected.
Here, instead of looking at the sum of distances, each node at distance $k$ contributes $a^k$ for some $a \in (0,1)$.
Formally, for a decay factor $a \in (0,1)$, decay centrality is defined as follows:
\begin{equation}\label{eq:c:decay}
Y_v(G, \bs) = \sum_{u \in V \setminus \{v\}} a^{dist_{u,v}(G)}.
\end{equation}
Here, we assume that if there is no path from $u$ to $v$ in $G$, then $dist_{u,v}(G) = \infty$ which implies $a^{dist_{u,v}(G)} = 0$.
Decay centrality can also be considered an extension of degree centrality that counts not only direct predecessors, but also further predecessors with decreasing weights.

\emph{Betweenness centrality}~\cite{Freeman:1977} is also based on the notion of shortest paths. 
However, its goal is to measure how often a specific node is an intermediary between other nodes.
If we assume that the information travels through the shortest paths, then betweenness centrality assesses that by computing for every pair of nodes, $s,t$, a fraction of shortest paths between them that goes through a node in question.
Formally, betweenness centrality is defined as follows:
\begin{equation}\label{eq:c:betweenness}
B_v(G, \bs) = \sum_{s,t \in V \setminus \{v\} : \sigma_{st} \neq 0} \frac{\sigma_{st}(v)}{\sigma_{st}},
\end{equation}
where $\sigma_{st}$ is the number of shortest paths from $s$ to $t$ and $\sigma_{st}(v)$ is the number of such shortest paths that goes through $v$.
The condition that $\sigma_{st} \neq 0$ is necessary in order to make the definition correct for graphs which are not strongly connected.

All three centralities assess a node based only on a connected component it is in, so they satisfy Node Deletion.
Edge Deletion is violated only by betweenness centrality, as edges between successors affect the shortest paths a node belongs to.
Edge multiplication does not affect distances in a graph, so Edge Multiplication is satisfied by closeness and decay centralities.
It is, however, violated by betweenness centrality, as edge multiplication changes the proportion of shortest paths that go through a specific node.
Edge Swap and Node Redirect are not satisfied by any of the three listed centralities. 
Both axioms are specific for feedback centralities, as the corresponding operations strongly affect the structure of the shortest paths in a graph on which all three centralities depend on.
Finally, Baseline is violated by decay and betweenness centralities and trivially satisfied by closeness centrality which is defined only for strongly connected graphs.

\subsection{Summary}

\begin{table}[t!]
\small
\setlength{\tabcolsep}{5pt}
\centering
\begin{tabular}{l|lcccccc}
\multicolumn{2}{}{} & \tworows{Node}{Deletion} & \tworows{Edge}{Deletion} & \tworows{Edge}{Multiplication} & \tworows{Edge}{Swap} & \tworows{Node}{Redirect} & Baseline\\
\hline
\emph{Parallel} 		& Degree cent.	    & $+\ \ $       & $+\ \ $       & $-\ \ $ & $+\ \ $ & $-\ \ $ & $-\ \ $ \\
\emph{feedback} 		& Eigenvector cent.     & $+^{*\dagger}$ & $+^{*\dagger}$ & $-\ \ $ & $+^*$ & $+^*$ & $+^{*\dagger}$ \\
\emph{centralities} 	& Katz cent.		        & $+^*$       & $+^*$       & $-\ \ $ & $+^*$ & $+^*$ & $+^*$ \\
 	& Bonacich cent.		  				      & $+^*$       & $+^*$       & $-\ \ $ & $-\ \ $ & $+^*$ & $-\ \ $ \\
\hline
\emph{Distributed} 	& Beta measure               & $+\ \ $       & $+\ \ $       & $+\ \ $ & $+\ \ $ & $-\ \ $ & $-\ \ $ \\
\emph{feedback} 		& Seeley index              & $+^{*\dagger}$ & $+^{*\dagger}$ & $+^*$ & $+^*$ & $+^*$ & $+^{*\dagger}$ \\
\emph{centralities} 	& PageRank & $+\ \ $       & $+\ \ $       & $+\ \ $ & $+\ \ $ & $+\ \ $ & $+\ \ $ \\
\hline
\emph{Other}			& Closeness cent. & $+^{*\dagger}$ & $+^{*\dagger}$ & $+^*$ & $-\ \ $ & $-\ \ $ & $+^{*\dagger}$ \\
\emph{centralities}	& Decay cent. & $+\ \ $ & $+\ \ $ & $+\ \ $ & $-\ \ $ & $-\ \ $ & $-\ \ $ \\
					& Betweenness cent. & $+\ \ $ & $-\ \ $ & $-\ \ $ & $-\ \ $ & $-\ \ $ & $-\ \ $
\end{tabular}

\caption{The axiomatic comparison of centrality measures. For each centrality measure and each axiom 
a plus ($+$) means that the axiom is satisfied by the centrality measure and a minus ($-$) that it is not satisfied.
A star ($^*$) highlights the fact that a centrality measure is defined only for a (not complete) class of graphs and a weaker version of the axiom is considered.
In such a case, a dagger ($\dagger$) denotes the fact that for this class of graphs every possible centrality measure satisfies the given axiom (e.g., Node Deletion is trivially satisfied by all centrality measures defined on strongly connected graphs because there are no isolated nodes in such graphs).
}
\label{table:axioms}
\end{table}

Table~\ref{table:axioms} summarizes the discussion.
Details and proofs can be found in~\ref{section:appendix:other_centralities}.

As we can see, Node Deletion and Edge Deletion are simple axioms satisfied by most centralities, at least in their restricted version.
Edge Multiplication distinguishes parallel and distributed feedback centralities. 
For other centralities, satisfiability depends on the way multiple edges between two nodes are interpreted. 
Edge Swap is characteristic for feedback centralities---it is satisfied by almost all feedback centralities and violated by all other centralities. 
Similarly, Node Redirect is satisfied only by feedback centralities in which the importance of a node depends on the importance of predecessors.
Finally, as most centralities ignore node weights, Baseline is hardly satisfied.

As we proved in Theorem~\ref{theorem:main}, PageRank is the only centrality measure that satisfy all of the axioms in their unrestricted form.
Seeley index also satisfies all of the axioms, but in the restricted versions.
Eigenvector centrality and Katz centrality, the parallel versions of the former two centralities, violate only Edge Multiplication.
From centralities defined on all graphs, beta measure is the closest to PageRank, as it satisfies four out of six axioms. 
In particular, a centrality measure defined as a sum of beta measure and the weight of a node is an example of a measure that satisfies all axioms except for Node Redirect.
We used this centrality in the proof of the independence of axioms (Theorem~\ref{theorem:independence}).

\section{Discussion}\label{section:discussion}
In this section we focus on two connections of our work to the existing literature.
First, we discuss different variants of the definition of PageRank that were considered in previous works.
Next, we relate our axiomatization of PageRank to the axiomatic characterizations of Seeley index by~\citet{Altman:Tennenholtz:2005} and~\citet{Palacios-Huerta:Volij:2004}.

\subsection{PageRank definitions}\label{section:definitions}
Since PageRank was originally proposed by~\citet{Page:etal:1999}, many slightly different variants of its definition appeared in the literature. 
The main aspect in which these variants differ is the question whether the measure should be normalized, i.e., should the values of PageRank for all nodes sum up to 1.
Both normalized \cite{Kamvar:etal:2003:extrapolation,DelCorso:etal:2005,Berkhin:2005,Langville:Meyer:2004,Boldi:etal:2009} 
and unnormalized \cite{Brandes:Erlebach:2005,Bianchini:etal:2005,Fogaras:etal:2005,Gleich:2015,Lofgren:etal:2016} 
versions of PageRank appear in the literature. 
In this paper we do not assume normalization.
This approach is especially popular in the analysis of large networks as it prevents the values from being extremely small.

The idea of normalization originates from the fact that PageRank was initially interpreted as a stationary probability distribution of a random process on a network~\cite{Page:etal:1999}.
To obtain normalization, it is usually required that $\bs(G) = 1-a$, but if it is not the case, weights can be proportionally scaled.
Also, sinks are usually reconnected to other nodes (which we discuss at the end of this section).
For each graph $G=(V,E)$, node weights $b$, and node $v \in V$, the value of normalized PageRank is equal to the value of PageRank as it is defined in this work divided by the sum of centralities of all nodes, i.e, $PR_v(G,\bs)/\sum_{u \in V} PR_u(G,\bs)$~\cite{Boldi:etal:2006}.
Normalized PageRank satisfies three of our axioms: Edge Multiplication, Edge Swap, and Node Redirect.
The other three axioms, Node Deletion, Edge Deletion and Baseline, are not satisfied in the current form, but they can be adapted to be satisfied by normalized PageRank.

The normalization, or the lack of it, affects the properties of PageRank in the number of ways.
First, as indicated by Lemma~\ref{lemma:locality}, unnormalized PageRank satisfies Locality, i.e., the centrality of a node depends only on the connected component to which it belongs.
Normalized PageRank clearly does not satisfy such property.
Moreover, unnormalized PageRank is linear as a function of node weights.
Specifically, for every graph, $G=(V,E)$, two node weights, $\bs$ and $\bs'$, and constant $t \in [0,1]$, it holds that 
\[
    PR_v(G, t \cdot \bs + (1-t) \cdot \bs') =  t \cdot PR_v(G,\bs) + (1-t) \cdot PR_v(G,\bs'),
\]
for every $v \in V$.
However, the same property does not hold for normalized PageRank~\cite{Boldi:etal:2006}.
In turn, normalized PageRank may be invariant to graph changes that affect unnormalized PageRank. 
See Fig.~\ref{figure:pagerank:definitions} for an illustration.

\begin{figure}[t]
\centering
\begin{tikzpicture}
  \def\x{1cm} 
  \def\y{0cm} 
  \def\arrdist{0.3cm}

  \tikzset{
    node_blank/.style={circle,draw,minimum size=0.5cm,inner sep=0, color=white}, 
    node/.style={circle,draw,minimum size=0.5cm,inner sep=0, fill = black!05}, 
    edge/.style={sloped,-latex,above,font=\footnotesize}, 
    el/.style={below,font=\footnotesize}, 
    operation/.style={sloped,>=stealth,above,font=\footnotesize},
    arrow/.style={draw, single arrow, minimum width = 0.9cm, minimum height=\y-6*\x+\s, fill=black!10},
    blank/.style={}
  } 
  \node[node] (A) at (\y+0*\x, 0*\x) {$u$};
  \node[node, fill = white] (B) at (\y+2*\x, 0*\x) {$v$};
  \node[blank] (A_) at (\y+0.5*\x-0.02cm, -0.7cm) {$(G,\bs)$};

  \path[->,draw,thick]
  (A) edge[edge, bend right=-20]  (B)
  ;
  
  \def\y{5cm} 
  
  \node[node] (A) at (\y+0*\x, 0*\x) {$u$};
  \node[node, fill = white] (B) at (\y+2*\x, 0*\x) {$v$};
  \node[blank] (A_) at (\y+0.5*\x-0.02cm, -0.7cm) {$(G',\bs)$};

  \path[->,draw,thick]
  (A) edge[edge, bend right=-20]  (B)
  (B) edge[edge, bend right=-20]  (A)
  ;
  
\end{tikzpicture}
\caption{Sample graphs $(G,\bs) = ((\{u,v\},\lBrace (u,v) \rBrace), [1,0])$ and $(G',\bs) = ((\{u,v\},\lBrace (u,v),(v,u) \rBrace), [1,0])$ illustrating the difference between normalized and unnormalized PageRank. Assume $a = 0.9$. For unnormalized PageRank in $(G,\bs)$ we have $PR^{0.9}_u(G,\bs) = 1$ and $PR^{0.9}_v(G,\bs) = 0.9$, while in graph $(G',\bs)$ the values are approximately 5.26 times higher, i.e., $PR^{0.9}_u(G',\bs) \simeq 5.26$ and $PR^{0.9}_v(G',\bs) \simeq 4.74$. However, normalized PageRank gives the same values in both graphs: in both graphs normalized PageRank of node $u$ is approximately equal $0.53$ and normalized PageRank of $v$ to $0.47$.}
\label{figure:pagerank:definitions}
\end{figure}

Another aspect where definitions of PageRank differ is the treatment of sinks.
Sinks (or \emph{dangling links} as they are sometimes called in the WWW network setting) are problematic if we interpret PageRank as a stationary distribution of a random walk as the walk terminates when it arrives at a sink.
Hence, especially for normalized PageRank, various methods of dealing with sinks were considered.

In their original paper, \citet{Page:etal:1999} removed sinks altogether from a network.
Then, after the computation of values of PageRank for all nodes in the remaining graph, they assigned each sink the centrality as given by the recursive equation.
However, such operation was treated mainly as an ad-hoc fix that works ''without affecting things significantly''~\cite{Page:etal:1999}.

Another approach popular in the literature is to add outgoing edges from each sink to all other nodes in such a way that the probability of moving from the sink to a node is proportional to the weight of this node~\cite{DelCorso:etal:2005}.
For example, if node weights are uniform, we can add a single outgoing edge from each sink to every node in the network including this sink.
Such operation does not affect the centrality of any node in the normalized version of PageRank.
However, it affect the values of unnormalized PageRank.

Moreover, it is also possible to reconnect sinks to all other nodes using different distribution, e.g., we can add a single outgoing edge to each node in the network even when node weights are not uniform.
In this way, we obtain so called~\emph{weakly preferential} PageRank~\cite{Boldi:etal:2006}.

Finally, for unnormalized PageRank, usually sinks are left as they are.
This is also the approach we use in our work.
For the sake of random walk interpretations in which it is required that no sinks are present in the network, we can add a virtual node with a self-loop to a graph and a single edge from each sink to this new node.
As showed by \citet{Bianchini:etal:2005} such operation does not affect unnormalized PageRank of each of the original nodes.

\subsection{Axiomatic characterizations of the simplified PageRank}\label{section:simplified_pr_axioms}
In this section, we discuss the existing axiomatizations of Seeley index, a simplified version of PageRank, and its slight modification for the journal citation network, called \emph{the invariant method}, and relate them to our axioms from Section~\ref{section:axioms}.

\subsubsection{Axiomatizing the invariant method}
\citet{Palacios-Huerta:Volij:2004} considered the problem of measuring the importance of scientific journals based on the journal citation network.
In a citation network, nodes represent journals and an edge $(A,B)$ represents a reference of journal $A$ to journal $B$.
The authors presented an axiomatization of the \emph{invariant method} \cite{Pinski:Narin:1976} which is equal to Seeley index of a journal in a citation network divided by the number of articles it has published.
To this end, they proposed the following four axioms:
\begin{description}
\item[Invariance with Respect to Reference Intensity:]
\emph{Multiplying the references in every journal by arbitrary constants, specific for each journal, does not affect the importance of any journal.}

This axiom, satisfied by both the invariant method and Seeley index, is also satisfied by PageRank.
It is equivalent to Edge Multiplication with the only difference that it is formulated as an operation on all nodes at once. 

\item[Weak Homogeneity:]
\emph{Imagine that there are only two journals, $A$ and $B$, and both have the same number of articles and references (some to themselves, some to the other journal).
If journal $A$ has $x$ times more references from $B$ than $B$ from $A$, then $A$ is $x$ times more important.}

This axiom is satisfied by both the invariant method and Seeley index.
However, it is not satisfied by PageRank: PageRank is not proportional to the impact of predecessors as it takes into consideration also the weight of a node.

\item[Weak Consistency:]
\emph{Assume that every journal has the same number of articles and references.
Consider deleting one of the journals, $A$, and for every journal $B$ citing $A$ redirecting all references of $B$ to $A$ to journals that were originally cited by $A$, preserving the proportions in the numbers of citations (e.g., if $B$ had six references to $A$ and $A$ had two references to $C$ and one to $D$, then four references are added from $B$ to $C$ and two from $B$ to $D$).
Now, the importance of any journal in the resulting network is the same as in the original network.}

This axiom, satisfied by both the invariant method and Seeley index, is based on the fact that in both of these measures intermediaries transfer further the whole importance they got from their predecessors.
In particular, if a node is added in the middle of an edge, then the importance of all nodes remain the same.
PageRank does not satisfy this axiom, since the decay factor decreases the importance transferred by the intermediaries.

\item[Invariance to Splitting of Journals:]
\emph{Consider an operation of splitting a journal into $k$ identical copies in a way that each copy has exactly $1/k$ of the original edges to each cited journal.
Now, if every journal is split into an arbitrary number of copies, then every copy will have the same importance as the original journal in the original network.}

This axiom is satisfied by the invariant method, but it is not satisfied by Seeley index: according to Seeley index, the total importance of all copies equals the original importance of a journal.
The same is true also for PageRank.
Node Redirect is based on the same idea, but instead of splitting all nodes into several copies, it considers merging two copies with possible different incoming edges into one node.
\end{description}

\subsubsection{Axiomatizing the ranking}

\citet{Altman:Tennenholtz:2005} proposed an axiomatization of the ranking of nodes that result from Seeley index.
As a result, axioms are of the different nature, as they concern the relation between centralities of different nodes, but not the specific values.
The authors proposed five such axioms:

\begin{description}
\item[Isomorphism:]
\emph{In two isomorphic graphs the ranking is the same with respect to the isomorphism.}

This axiom, proposed in the seminal work by \citet{Sabidussi:1966}, is satisfied by all reasonable centrality measures, including all measures introduced in this paper.

\item[Self Edge:]
\emph{Adding a self-loop to a node can only increase its position in the ranking and does not affect the ranking of other nodes.}

This axiom is not satisfied by PageRank, because adding a self-loop to a node can significantly decrease the centrality of its direct successors and change their ranking with respect to other nodes.

\item[Vote by Committee:]
\emph{Splitting a node into $k+1$ parts in a way that one node has the original incoming edges and $k$ outgoing edges to other parts and $k$ nodes have one incoming edge each and the original outgoing edges of the node does not affect the ranking of nodes in the graph.}

This axiom is based on a similar principle as Weak Consistency from the axiomatization of the invariant method.
Similarly, it is not satisfied by PageRank since the decay factor decreases the importance transferred by the intermediaries.
Hence, the direct predecessors of the split node may end up with a lower ranking.

\item[Collapsing:]
\emph{Redirecting a node into its out-twin does not affect the ranking of other nodes.}

This axiom, very similar to Node Redirect, is satisfied also by PageRank.

\item[Proxy:]
\emph{Assume that there is a node $v$ with $k$ incoming and $k$ outgoing edges such that all direct predecessors are different and have equal centralities.
Removing node $v$ and adding $k$ edges from direct predecessors to direct successors, one edge to each node, does not affect the ranking of any node in a graph.}

This axiom, similar to Vote by Committee, is not satisfied by PageRank because of the same reason: the decay factor decreases the importance transferred by the intermediaries.
\end{description}

\subsubsection{Axiomatizing the values}

\citet{Slutzki:Volij:2006} proposed an axiomatization of Seeley index values, not the corresponding ranking. The authors used three axioms: Invariance to Reference Intensity, which is equivalent to the axiom under almost the same name from the axiomatization of the invariant method, and two additional axioms:

\begin{description}
\item[Uniformity:]
\emph{If all nodes have the same number of incoming and outgoing edges, then all nodes have equal importance as well.}

This simple borderline axiom is satisfied by PageRank (assuming node weights of all nodes are equal). 

\item[Weak Additivity:]
\emph{If all nodes have the same number of incoming and outgoing edges, then if we add arbitrary number of pairs of symmetric edges (pairs $(u,v)$ and $(v,u)$ for arbitrary nodes $u,v$) the importance will remain proportional to in-degrees of nodes.}

This axiom is not satisfied by PageRank, as because of the decay factor and added node weight, PageRank is not proportional to the importance obtained from its direct predecessors. This axiom is very strong, as it directly specifies the importance of nodes in a general scenario.

\end{description}

As we can see, in the first two axiomatizations there is an axiom that concerns splitting or merging out-twin nodes, as in Node Redirect. 
Also, in the first and last axiomatizations, there is an axiom equivalent to Edge Multiplication.
However, most axioms are not satisfied by PageRank.
That is why the axiomatic characterizations described in this section cannot be easily extended for the characterization of PageRank.

\section{Conclusions}\label{section:conclusions}

In this paper, we proposed the first axiomatic characterization of PageRank. 
Specifically, we proposed six simple axioms, namely Node Deletion, Edge Deletion, Edge Multiplication, Edge Swap, Node Redirect and Baseline, and proved that PageRank is the only possible centrality measure that satisfy all of them.
Our characterization is the first axiomatization of PageRank in the literature. 

There are various interesting directions in which our work can be extended.
A common axiomatization of the most popular feedback centralities based on the properties proposed in this paper is a desirable goal.
Also, from the application perspective, it would be interesting to study specific types of networks in more detail and analyze whether PageRank is a good fit.
Finally, we plan to study also axiomatic properties of centralities which, as PageRank, are based on the notion of random walk, such as the \emph{random walk closeness centrality}~\cite{White:Smyth:2003}.

\section{Acknowledgment}
This work is partially based on \citet{Was:Skibski:2018:pagerank} presented at the 27th International Joint Conference on Artificial Intelligence (IJCAI-18). Compared to the conference publication, the set of axioms has changed which led to a new proof of uniqueness. All other parts of the paper are also new, including the proof of independence, the comparison with other centrality measures, the extended related work, and the discussion.

We would like to thank the anonymous referees for their helpful comments, which, among others, inspired the discussion in Section~\ref{section:definitions}.

This work was supported by the National Science Centre under Grant No. 2018/31/B/ST6/03201 and the Foundation for Polish Science under Grant Homing/2016-1/7.

\appendix



\section[PageRank satisfies axioms (Theorem 1)]{PageRank satisfies axioms (Theorem~\ref{theorem:main})}\label{section:appendix:proof-2}
In this appendix, we show that PageRank for every decay factor $a \in [0,1)$ satisfies Node Deletion, Edge Deletion, Edge Multiplication, Edge Swap, Node Redirect and Baseline. This is a part of the proof of Theorem~\ref{theorem:main}.

First, let us argue that there is a unique centrality measure that satisfies PageRank recursive equation~\eqref{eq:pr:main}.
Using matrix form, our formula can be written as follows: 
\[ PR = a \cdot  \hat{A}^T \cdot PR + b, \] 
where $PR = (PR^a_v(G,b))_{v \in V}$ is the vector of PageRank values, $\hat{A}$ is the row-normalized adjacency matrix ($\hat{A}[u,v] = \#_{(u,v)}(G)/\deg^+_u(G)$), and $b = (\bs_v)_{v \in V}$ is the vector of node weights. By $\mathbb{I}$ let us denote the identity matrix. Then, we get: $(\mathbb{I} - a \cdot \hat{A}^T) PR = b$.
Matrix $(\mathbb{I} - a \cdot \hat{A}^T)$ is invertible, as its transpose is strictly diagonally dominant which gives: 
\[ PR = (\mathbb{I} - a \cdot \hat{A}^T)^{-1} b. \]

Let us now formally introduce a \emph{busy random surfer model} that provides a random walk interpretation of PageRank.
Our model, partially inspired by the work of~\citet{Bianchini:etal:2005}, is a modification of the original random surfer model proposed by~\citet{Page:etal:1999}.

\subsubsection*{Busy Random Surfer Model}
Imagine a surfer that is traversing the World Wide Web in a random manner.
She starts from a random page (with distribution specified by the weights) and then, in each step of her walk:
\begin{itemize}
\item with probability $a$ the surfer chooses one of the links on a page and clicks it, which moves her to the next page, or
\item with probability $1 - a$ (or if there is no link on the page) the surfer gets bored and stops surfing.
\end{itemize}
Now, PageRank is equal to the expected number of times the surfer visits a page (up to a scalar multiplication).

More generally, we define a \emph{random walk with decay} on graph $(G,\bs)$ as a random finite sequence of nodes $w^a_{G,\bs} = \big(w^a_{G,\bs}(0),w^a_{G,\bs}(1),\ldots,w^a_{G,\bs}(T)\big)$ selected in the following way:
\begin{itemize}
\item the first node is drawn from the initial distribution determined by the weights of nodes: $\mathbb{P}\big( w^a_{G,\bs}(0) = v \big) = \bs(v)/\bs(G)$;
\item the next node is determined by a transition probability: if after $t$ steps the walk is at a node $v_t$, then it either moves to a neighbor $v_{t+1}$ with probability $a \cdot \#_{(v_t,v_{t+1})}(G)/\deg^+_{v_t}(G)$ or ends with probability $1-a$ (or with probability $1$ if $v_t$ is a sink).
\end{itemize}
As a result, the probability of a specific sequence is given by:
\[
	\mathbb{P}\big( w^a_{G,\bs} = (v_0,v_1,\dots,v_T) \big) =
	\begin{cases}
		\bs(v_0)/\bs(G) \cdot a^T \cdot 
			\left( \prod_{t =0}^{T-1} \#_{(v_t,v_{t+1})}(G)/\deg^+_{v_t}(G)  \right)
			& \mbox{if } v_T \mbox{ is a sink,}\\
		\bs(v_0)/\bs(G) \cdot a^T \cdot 
			\left( \prod_{t =0}^{T-1} \#_{(v_t,v_{t+1})}(G)/\deg^+_{v_t}(G)  \right)
			\cdot (1-a)
			& \mbox{otherwise.}
	\end{cases}
\]
Note that only sequences which are paths have non-zero probability.

\begin{example}\label{example:rwwd}
Consider the random walk with decay on the graph from Fig.~\ref{figure:main} assuming each node has weight one.
Let us calculate the probability of path $(v_5, v_7, v_1, v_8)$.
The random walk starts from node $v_5$ with probability $1/8$ (there are 8 nodes, all with equal weights).
Now, from $v_5$ it either chooses one of the outgoing edges with probability $a$ or ends with probability $1-a$. 
If the random walk with decay chooses an edge, it chooses an edge to $v_7$ with probability $1/3$, because $v_5$ has three outgoing edges.
Hence, the random walk moves from $v_5$ to $v_4$ with the probability $a/3$.
Now, node $v_7$ has one outgoing edge to $v_1$ and $v_1$ has one outgoing edge to $v_8$; hence, if only random walk with decay does not end, it has to move from $v_7$ to $v_1$ and then to $v_8$.
Hence, the probability of such two moves is $a^2$.
Finally, node $v_8$ is not a sink, hence the random walk with decay ends here with probability $(1-a)$.
Summing up, the probability that the random walk with decay is a sequence $(v_5, v_7, v_1, v_8)$ equals $1/8 \cdot a/3 \cdot a \cdot a \cdot (1-a) = a^3(1-a)/24$.
This concludes our example.
\end{example}

In the following theorem, we prove that PageRank is indeed equal to the expected number of visits of a random walk with decay.
In particular, if the weights of all nodes in a graph sum up to $1$, then it is equal to the expected number of visits of a random walk with decay.

\begin{theorem}
\label{theorem:rpm}
For every graph $G=(V,E)$, weights $\bs$ and node $v \in V$ it holds that
\begin{equation}
\label{eq:rpm}
PR^a_v(G,\bs) = \sum_{t=0}^\infty \mathbb{P} \big( w^a_{G,\bs}(t)=v \big) \cdot \bs(G).
\end{equation}
\end{theorem}
\begin{proof}
Consider a centrality measure $F^a$ defined for every graph $G = (V,E)$, weights $\bs$ and node $v \in V$ as the right-hand side of equation~\eqref{eq:rpm}:
\begin{equation}\label{eq:rpm_x}
F^a_v(G,\bs)=\sum_{t=0}^\infty \mathbb{P}\big( w^a_{G,\bs}(t)=v \big) \cdot \bs(G).
\end{equation}
We will prove that $F^a$ satisfies PageRank recursive equation~\eqref{eq:pr:main} which implies $F^a = PR^a$.

Consider the probability that a random walk with decay is in node $v$ at time $t$: $\mathbb{P}\big( w^a_{G,\bs}(t)=v \big)$. 
If $t=0$, then from the initial distribution we know that $\mathbb{P}\big( w^a_{G,\bs}(0)=v \big) = \bs(v)/\bs(G)$.
If $t \ge 1$, then we know that the walk used one of the incoming edges of node $v$ to reach it; hence:
\[
\mathbb{P}\big( w^a_{G,\bs}(t)=v \big) = \sum_{(u,v)\in \Gamma^-_v(G)} \left(\mathbb{P}\big( w^a_{G,\bs}(t-1)=u\big) \cdot \frac{a}{\deg^+_u(G)}\right).
\]
Here, $a/\deg^+_u(G)$ is the probability that the walk while being in node $u$ will not terminate and choose one particular outgoing edge.
Summing over all $t \ge 1$ and reformulating the right-hand side we get:
\begin{equation}
\label{eq:rpm_x2}
\sum_{t = 1}^{\infty} \mathbb{P}\big( w^a_{G,\bs}(t)=v \big) = a \cdot \sum_{(u,v)\in \Gamma^-_v(G)} \left(\frac{\sum_{t = 0}^{\infty} \mathbb{P}\big( w^a_{G,\bs}(t)=u \big)}{\deg^+_u(G)} \right).
\end{equation}
From equation~\eqref{eq:rpm_x} we have that $\sum_{t = 0}^{\infty} \mathbb{P}\big( w^a_{G,\bs}(t)=u \big) = F^a_u(G,\bs)/\bs(G)$. 
Hence, by multiplying both sides of Equation~\eqref{eq:rpm_x2} by $\bs(G)$ and adding $\bs(v)$ we get that $F^a$ satisfies PageRank recursive equation~\eqref{eq:pr:main}:
\[
F^a_v(G,\bs) = \bs(v) + \sum_{t = 1}^{\infty} \mathbb{P}\big( w^a_{G,\bs}(t)=v \big) \cdot \bs(G) = \bs(v) + a \cdot \sum_{(u,v)\in \Gamma^-_v(G)} \frac{F^a_u(G, \bs)}{\deg^+_u(G)}.
\]
This concludes the proof.
\end{proof}

Let us move to the main proof of this section, i.e., that PageRank satisfies all of our axioms.

All axioms except for Baseline are invariance axioms, so we need to show that a specific graph operation does not affect PageRank of a node in question.
To this end, we will either look at the random walk interpretation of PageRank or PageRank recursive equation~\eqref{eq:pr:main}.
In the former case, based on Theorem~\ref{theorem:rpm} it is enough to show that a specific operation does not affect the expected number of visits by a random walk with decay (or if the sum of node weights decreases, then the expected number of visits increases accordingly); we will do this for four axioms: Node Deletion, Edge Deletion, Edge Multiplication and Node Redirect.
In the later case, we will show that the system of recursive equations after the graph modification has almost the same solution; we will do this for Edge Swap.
Finally, Baseline easily follows from PageRank recursive equation.

Fix an arbitrary graph $G = (V,E)$, node weights $\bs$.
We will consider each axiom separately.

\subsubsection*{Node Deletion}
Let $u$ be an isolated node and $v$ be an arbitrary node other than $u$.
Consider a graph obtained from $(G,\bs)$ by removing node $u$: $(G',\bs') = ((V \setminus \{u\}, E), \bs')$ with $\bs'(w) = \bs(w)$ for $w \in V \setminus \{u\}$. 
We have to prove that $F_v(G,\bs) = F_v(G',\bs')$.
To this end, based on Theorem~\ref{theorem:rpm}, it is enough to prove that the probability that a random walk with decay visits node $v$ at time $t \ge 0$ multiplied by the sum of node weights is the same for both graphs:
\begin{equation}\label{eq:axioms:nd1}
\mathbb{P} \big( w^a_{G,\bs}(t)=v \big) \cdot \bs(G) = \mathbb{P} \big( w^a_{G',\bs'}(t)=v \big) \cdot \bs'(G').
\end{equation}
Note that this means that the probability $\mathbb{P} \big( w^a_{G,\bs}(t)=v \big)$ increases $\bs(G)/\bs'(G')$ times if node $u$ is removed.

Consider a random walk with decay on graph with node weights $(G,\bs)$.
Since node $u$ is isolated, we know that it is not possible for the walk to start in $u$ and visit $v$ later on. 
Hence, we get:
\begin{equation}\label{eq:axioms:nd2}
\mathbb{P}\big( w^a_{G,\bs}(t)=v \big) = \sum_{s \in V \setminus \{u\}} \mathbb{P}\big( w^a_{G,\bs}(t)=v \land w^a_{G,\bs}(0) = s \big).
\end{equation}
Take an arbitrary $s \in V \setminus \{u\}$.
The walk starts in $s$ with the probability $\bs(s)/\bs(G)$.
Hence, if we denote the conditional probability that the walk visits node $v$ at time $t$ \emph{after} it starts in $s$ by $\mathbb{P}\big( w^a_{G,\bs}(t)=v \mid w^a_{G,\bs}(0)=s \big)$ we will get:
\begin{equation}\label{eq:axioms:nd3}
\sum_{s \in V \setminus \{u\}} \mathbb{P}\big( w^a_{G,\bs}(t)=v \land w^a_{G,\bs}(0) = s \big) = \sum_{s \in V \setminus \{u\}} \frac{\bs(s)}{\bs(G)} \cdot \mathbb{P}\big( w^a_{G,\bs}(t)=v \mid w^a_{G,\bs}(0)=s \big).
\end{equation}
Analogously, for graph $(G',\bs')$ we have:
\begin{align}
\notag
\mathbb{P}\big( w^a_{G',\bs'}(t)=v \big) &=
\sum_{s \in V \setminus \{u\}} \mathbb{P}\big( w^a_{G',\bs'}(t)=v \land w^a_{G',\bs'}(0) = s \big)\\
\label{eq:axioms:nd4}
&=\sum_{s \in V \setminus \{u\}} \frac{\bs'(s)}{\bs'(G')} \cdot \mathbb{P}\big( w^a_{G',\bs'}(t)=v \mid w^a_{G',\bs'}(0)=s \big).
\end{align}
Note that $V \setminus \{u\}$ is the set of all nodes in $(G',\bs')$ and $\bs'(s) = \bs(s)$ for every $s \in V \setminus \{u\}$.
Now, consider an arbitrary walk on graph $(G,\bs)$ that starts in node $s$.
Since $u$ is isolated, we know that the walk will not visit $u$.
Hence, in a graph with node $u$ removed, i.e., $(G',\bs')$, the probability that such a walk visits node $v$ at time $t$ is the same as in $(G,\bs)$:
\[ \mathbb{P}\big( w^a_{G,\bs}(t)=v \mid w^a_{G,\bs}(0)=s \big) = \mathbb{P}\big( w^a_{G',\bs'}(t)=v \mid w^a_{G',\bs'}(0)=s \big) \]
This combined with equations~\eqref{eq:axioms:nd2}--\eqref{eq:axioms:nd4} proves equation~\eqref{eq:axioms:nd1}.

\subsubsection*{Edge Deletion}
Fix edge $(u,u') \in E$ and let $v$ be an arbitrary node which is not a successor of $u$, i.e., $v \not \in S_u(G)$.
Consider a graph obtained from $(G,\bs)$ by removing edge $(u,u')$: $(G',\bs') = ((V, E - \lBrace (u,u') \rBrace), \bs)$.
We have to prove that $F_v(G,\bs) = F_v(G',\bs')$.
Note that the sum of weights is the same in both graphs.
Hence, based on Theorem~\ref{theorem:rpm}, to prove that $F_v(G,\bs) = F_v(G',\bs')$ it is enough to show that the probability that a random walk with decay visits node $v$ at time $t \ge 0$ is also the same in both graphs:
\begin{equation}\label{eq:axioms:ed1}
\mathbb{P}\big( w^a_{G,\bs}(t)=v \big) = \mathbb{P}\big( w^a_{G',\bs'}(t)=v \big).
\end{equation}

Consider a random walk with decay on graph $(G,\bs)$.
Note that since $v \not \in S_u(G)$, then there is no path from $u$ to $v$ and we know that if the walk visits $v$ at time $t$, then it could not have visit $u$ before. 
Formally, $w^a_{G,\bs}(t) = v$ implies $w^a_{G,\bs}(t') \neq u$ for $t' \le t$.
Hence:
\begin{equation}\label{eq:axioms:ed2}
\mathbb{P}\big( w^a_{G,\bs}(t)=v \big) = \mathbb{P}\big( w^a_{G,\bs}(t)=v \land (\forall_{t' \le t} w^a_{G,\bs}(t') \neq u) \big).
\end{equation}
Clearly, $v$ will not become a successor of $u$ if we remove an edge from the graph, so for $(G',\bs')$ we also get:
\begin{equation}\label{eq:axioms:ed3}
\mathbb{P}\big( w^a_{G',\bs'}(t)=v \big) = \mathbb{P}\big( w^a_{G',\bs'}(t)=v \land (\forall_{t' \le t} w^a_{G',\bs'}(t') \neq u) \big).
\end{equation}
Now, observe that outgoing edges of node $u$ does not affect the initial distribution of a random walk with decay, nor the transition probabilities for other nodes.
In particular, they do not affect the probability that the walk that starts in a node other than $u$ will reach node $v$ by going through nodes from $V \setminus \{u\}$.
Hence, we get that:
\[ \mathbb{P}\big( w^a_{G,\bs}(t)=v \land (\forall_{t' \le t} w^a_{G,\bs}(t') \neq u) \big) = \mathbb{P}\big( w^a_{G',\bs'}(t)=v \land (\forall_{t' \le t} w^a_{G',\bs'}(t') \neq u) \big),\]
which combined with equations~\eqref{eq:axioms:ed2} and \eqref{eq:axioms:ed3} implies equation~\eqref{eq:axioms:ed1}.

\subsubsection*{Edge Multiplication} 
Let $u,v$ be two arbitrary nodes and $k \in \mathbb{N}$ a natural number.
Consider a graph obtained from $(G,\bs)$ by adding $k$ copies of outgoing edges of $u$: $(G',\bs') = ((V, E \sqcup k \cdot \Gamma^+_u(G)), \bs)$.
Observe that in both graphs the initial distribution and the transition matrix of a random walk with decay are the same.
Thus, the probability that at time $t$ the walk visits node $v$ for both of these graphs is also the same:
$$\mathbb{P}\big( w^a_{G,\bs}(t)=v \big) = \mathbb{P}\big( w^a_{G',\bs'}(t)=v \big).$$
Hence, from Theorem~\ref{theorem:rpm} we get that $F_v(G,\bs) = F_v(G',\bs')$.

\subsubsection*{Edge Swap}
Fix two edges $(u,u'), (w,w') \in E$ such that $PR^a_u(G,\bs) = PR^a_w(G,\bs)$ and $\deg^+_u(G) = \deg^+_w(G)$.
Consider a graph obtained from $(G,\bs)$ by swapping the ends of both these edges, i.e., $(G',\bs') = ((V,E - \lBrace (u,u'), (w,w') \rBrace \sqcup \lBrace (u,w'), (w,u') \rBrace), \bs)$.
We have to prove that $PR^a(G,\bs) = PR^a(G',\bs')$.
To this end, let us define $x_v = PR^a_v(G,\bs)$ for every $v \in V$ and prove that $(x_v)_{v \in V}$ satisfies the system of PageRank recursive equations~\eqref{eq:pr:main} for graph $(G',\bs')$, i.e., that:
\begin{equation}\label{eq:axioms:es1}
x_v = a \cdot \left( \sum_{(s,v) \in \Gamma_v^-(G')} \frac{x_s}{\deg_s^+(G')}\right) + \bs'(v),
\end{equation}
holds for every $v \in V$.
Since this system of equations has exactly one solution which is $PR^a(G',\bs')$, this will prove that $PR^a_v(G',\bs') = x_v = PR^a_v(G,\bs)$ for every $v \in V$.

Consider an arbitrary node $v \in V$. 
From PageRank recursive equation for graph $(G,\bs)$ we have:
\begin{equation}\label{eq:axioms:es2}
x_v = a \cdot \left( \sum_{(s,v) \in \Gamma_v^-(G)} \frac{x_s}{\deg_s^+(G)}\right) + \bs(v).
\end{equation}
By the definition of $(G',\bs')$ we know that out-degree and weight of every node is the same in $(G,\bs)$ as in $(G',\bs')$.
Moreover, if $v \not \in \{u',w'\}$, then it has the same set of incoming edges in both graphs.
Hence, replacing in equation~\eqref{eq:axioms:es2} $\deg^+_s(G), \bs(v), \Gamma^-_v(G)$ with $\deg^+_s(G'), \bs'(v), \Gamma^-_v(G')$ respectively proves equation~\eqref{eq:axioms:es1}.

Assume $v=u'$ (for $v=w'$ the proof is analogous).
Consider the sum on the right-hand side of the equation~\eqref{eq:axioms:es2} for node $v=u'$.
In graph $(G',\bs')$ node $u'$ has an edge $(u,u')$ replaced with $(w,u')$, i.e., we have: $\Gamma_{u'}^-(G) - \lBrace (u,u') \rBrace = \Gamma_{u'}^-(G') - \lBrace (w,u') \rBrace$.
Note, however, that from our initial assumption we know that $x_u=x_w$ and $\deg^+_u(G) = \deg^+_w(G)$. 
Hence, we get that:
\begin{equation}\label{eq:axioms:es3}
\sum_{(s,u') \in \Gamma_{u'}^-(G) - \lBrace (u,u') \rBrace} \frac{x_s}{\deg_s^+(G)} + \frac{x_u}{\deg_u^+(G)}  = 
\sum_{(s,u') \in \Gamma_{u'}^-(G') - \lBrace (w,u') \rBrace} \frac{x_s}{\deg_s^+(G')} + \frac{x_w}{\deg_w^+(G')},
\end{equation}
where we also used the fact that the out-degree of any node is the same in $(G,\bs)$ as in $(G',\bs')$.
This combined with equation~\eqref{eq:axioms:es2} proves equation~\eqref{eq:axioms:es1}.

\subsubsection*{Node Redirect}
Let $u,w$ be out-twins.
Consider a random walk $w^{a*}_{G,\bs}$ that is a modification of a standard random walk with decay $w^a_{G,\bs}$ that does not distinguish between nodes $u$ and $w$, i.e., whenever $w^a_{G,\bs}$ is equal to either $u$ or $w$, sequence $w^{a*}_{G,\bs}$ holds a value $u \lor w$.
Formally, $w^{a*}_{G,\bs}$ is a random sequence of nodes such that for every $t \ge 0$
\begin{equation}
\label{eq:axioms:nr1}
	w^{a*}_{G,\bs}(t) =
	\begin{cases}
		w^a_{G,\bs}(t) & \mbox{if } w^a_{G,\bs}(t) \in V \setminus \{u,w\},\\
		u \lor w & \mbox{otherwise.}
	\end{cases}
\end{equation}
Observe that the initial distribution of $w^{a*}_{G,\bs}$ is the same as in the random walk with decay, except that $\mathbb{P}\big( w^{a*}_{G,\bs}(0)= u \lor w \big) = (\bs(u) + \bs(w))/\bs(G)$.
Transition probabilities of $w^{a*}_{G,\bs}$, i.e., $p^{*}_{s,t}$, for every $s, t \in V \setminus \{u,w\}$ are also the same as in $w^a_{G,\bs}$ and $p^{*}_{s,u \lor w} = p_{s,u} + p_{s,w}$.
Since $u$ and $w$ are out-twins, the probabilities of transition from them are the same, hence also $p^{*}_{u \lor w, t} = p_{u,t} = p_{w,t}$.
Finally, $p^{*}_{u \lor w, u \lor w} = p_{u,u} + p_{w,u} = p_{w,u} + p_{w,w}$.

Now, let us focus on a graph obtained from $(G,\bs)$ by redirecting node $u$ into $w$, i.e., $(G',\bs') = R_{u \rightarrow w}(G,\bs)$.
Consider a random walk with decay on this graph, $w^a_{G',\bs'}$.
Observe that if we identify state $w$ in $w^{a*}_{G,\bs}$ with $u \lor w$ in $w^{a*}_{G,\bs}$, both walks have the same initial distributions and the same transition probabilities.
Hence, each sequence is identically probable in both walks.
Therefore,
\[
	\mathbb{P}\big( w^a_{G',\bs'}(t)= v \big) =
	\mathbb{P}\big( w^{a*}_{G,\bs}(t)= v \big) \quad
	\mbox{for every } t \ge 0 \mbox{ and } v \in V \setminus \{u\}.
\]
Thus, from equation~\eqref{eq:axioms:nr1} and Theorem~\eqref{theorem:rpm} we get that
\[
	PR^a_v(G',\bs') =
	\begin{cases}
		PR^a_v(G,\bs) & \mbox{if } v \in V \setminus \{u,w\} \\
		PR^a_u(G,\bs) + PR_w(G,\bs) & \mbox{otherwise.}
	\end{cases}
\]

\subsubsection*{Baseline} 
Assume $v$ is an isolated node in $(G,\bs)$. 
Since node $v$ has no incoming edges, from PageRank recursive equation~\eqref{eq:pr:main} we get that $PR_v^a(G,\bs) = \bs(v)$.
This proves that PageRank satisfies Baseline.


\section[Independence of Axioms (Theorem 11)]{Independence of Axioms (Theorem~\ref{theorem:independence})}\label{section:appendix:independence}
In this appendix, we present the full proof of Theorem~\ref{theorem:independence}.
We divide it into six lemmas, Lemmas~\ref{lemma:independence:nd}--\ref{lemma:independence:dn}, each characterizing a centrality different than PageRank that satisfies five out of six our axioms.


\begin{lemma}\label{lemma:independence:nd}
A centrality measure $F$ defined for every graph $(G,\bs)$ and node $v$ as follows: 
\[ F_v(G,\bs) = PR^{a(G,\bs)}_v(G,\bs), \quad \quad \mbox{ where } a(G,\bs) = 1/(2 + \bs(G)) \]
satisfies Edge Deletion, Edge Multiplication, Edge Swap, Node Redirect and Baseline, but does not satisfy Node~Deletion.
\end{lemma}
\begin{proof}
Observe that for any two graphs $(G,\bs)$, $(G',\bs')$ and a node $v$, if we know that $\bs(G)=\bs'(G')$ and $PR_v^a(G,\bs) = PR_v^a(G',\bs')$ for every $a \in [0,1)$, then we have $a(G,\bs) = a(G',\bs')$ and
\begin{equation}\label{eq:independence:nd:2}
F_v(G,\bs) = PR^{a(G,\bs)}(G,\bs) = PR^{a(G',\bs')}_v(G',\bs') = F_v(G',\bs').
\end{equation}
We will consider each axiom separately. 
Fix an arbitrary graph $G=(V,E)$ and node weights $\bs$.

\begin{itemize}
\item For Edge Deletion, consider an arbitrary edge $(u,w) \in E$.
Let $(G',\bs') = ((V,E - \lBrace (u,w) \rBrace), \bs)$.
Clearly, $\bs(G)=\bs'(G')$.
Since PageRank satisfies Edge Deletion, for every $v \in V \setminus S_u(G)$ we have $PR_v^a(G,\bs) = PR_v^a(G',\bs')$ for $a \in [0,1)$ and equation~\eqref{eq:independence:nd:2} implies $F_v(G,\bs) = F_v(G',\bs')$.

\item For Edge Multiplication, consider node $u \in V$ and $k \in \mathbb{N}$.
Let $(G',\bs') = ((V, E \sqcup k \cdot \Gamma^+_u(G)), \bs)$.
Clearly, $\bs(G)=\bs'(G')$.
Since PageRank satisfies Edge Multiplication, equation~\eqref{eq:independence:nd:2} implies $F_v(G,\bs) = F_v(G',\bs')$ for every $v \in V$.

\item For Edge Swap, assume $(u,u'), (w,w') \in E$ are two edges such as $F_u(G,\bs)=F_w(G,\bs)$ and $\deg^+_u(G) = \deg^+_w(G)$.
From the definition of $F$ we get that also $PR^{a(G,\bs)}_u(G,\bs)=PR^{a(G,\bs)}_w(G,\bs)$.
Let $(G',\bs') = ((V, E - \lBrace (u,u'), (w,w') \rBrace \sqcup \lBrace (u,w'), (w,u') \rBrace), \bs)$.
Clearly, $\bs(G)=\bs'(G')$.
Since PageRank satisfies Edge Swap, equation~\eqref{eq:independence:nd:2} implies $F_v(G,\bs) = F_v(G',\bs')$ for every $v \in V$.

\item For Node Redirect, assume $u,w \in V$ are out-twins.
Let $(G',\bs') = R_{u \rightarrow w}(G,\bs)$.
Clearly, $\bs(G)=\bs'(G')$.
Since PageRank satisfies Node Redirect, for every $v \in V \setminus \{u,w\}$ we have $PR_v^a(G,\bs) = PR_v^a(G',\bs')$ for every $a \in [0,1)$ and equation~\eqref{eq:independence:nd:2} implies $F_v(G,\bs) = F_v(G',\bs')$.
Analogously, we get that
\[ F_u(G,\bs) +  F_w(G,\bs) = PR^{a(G,\bs)}_u(G,\bs) + PR^{a(G,\bs)}_w(G,\bs) = PR^{a(G',\bs')}_w(G',\bs') = F_w(G',\bs'). \]

\item For Baseline, assume $v$ is isolated. 
Since PageRank satisfies Baseline we get that $F_v(G,\bs) = \bs(v)$.
\end{itemize}

Finally, consider Node Deletion.
Consider graph $(G,\bs) = ((\{u,v,w\},\lBrace (u,v) \rBrace), [1,0,1])$. We have $a(G,\bs) = 1/4$, so $F_v(G,\bs) = PR_v^{1/4}(G,\bs) = 1/4$. 
Note that $w$ is isolated in $(G,\bs)$.
Now, if we delete node $w$ we will obtain a graph $(G',\bs') = ((\{u,v\},\lBrace (u,v) \rBrace), [1,0])$. Here, we have $a(G',\bs') = 1/3$, so $F_v(G,\bs) = PR_v^{1/3}(G,\bs) = 1/3$. Thus, Node Deletion is not satisfied.
\end{proof}


\begin{lemma}
\label{lemma:independence:ed}
A centrality measure $F$ defined for every graph $(G,\bs)$ and node $v$ and an arbitrary $a \in (0,1)$ as follows: 
\[ F^a_v(G,\bs) = \begin{cases} 2 \cdot PR^a_v(G,\bs) - \bs(v) & \mbox{if } v \mbox{ is a sink,}\\ PR^a_v(G,\bs) & \mbox{otherwise,}\end{cases} \]
satisfies Node Deletion, Edge Multiplication, Edge Swap, Node Redirect and Baseline, but does not satisfy Edge~Deletion.
\end{lemma}
\begin{proof}
Fix $a \in (0,1)$.
Observe that for any two graphs $(G,\bs), (G',\bs')$ and a node $v$, if we know that a node $v$ is a sink in $(G,\bs)$ if and only if $v$ is a sink in $(G',\bs')$ and $PR^a_v(G,\bs) = PR^a_v(G',\bs')$, then we have 
\begin{equation}\label{eq:independence:ed:2}
\begin{cases}
F^a_v(G,\bs) = 2 \cdot PR^a_v(G,\bs) - \bs(v) = 2 \cdot PR^a_v(G',\bs') - \bs'(v) = F^a_v(G',\bs') & \mbox{if } v \mbox{ is a sink in }(G,\bs),\\
F^a_v(G,\bs) = PR^a_v(G,\bs) = PR^a_v(G',\bs') = F^a_v(G',\bs') & \mbox{otherwise.}
\end{cases}
\end{equation}

We will consider each axiom separately. 
Fix an arbitrary graph $G=(V,E)$ and node weights $\bs$.

\begin{itemize}
\item For Node Deletion, assume $u$ is an isolated node. 
Let $(G',\bs') = ((V \setminus \{u\}, E), \bs)$. 
Fix $v \in V \setminus \{u\}$.
Note that $v$ is a sink in $(G,\bs)$ if and only if it is a sink in $(G',\bs')$.
Since PageRank satisfies Node Deletion, equation~\eqref{eq:independence:ed:2} implies $F^a_v(G,\bs) = F^a_v(G',\bs')$.

\item For Edge Multiplication, consider $u \in V$ and $k \in \mathbb{N}$.
Let $(G',\bs') = ((V, E \sqcup k \cdot \Gamma^+_u(G)), \bs)$.
Fix $v \in V$.
Note that $v$ is a sink in $(G,\bs)$ if and only if it is a sink in $(G',\bs')$.
Since PageRank satisfies Edge Multiplication, equation~\eqref{eq:independence:ed:2} implies $F^a_v(G,\bs) = F^a_v(G',\bs')$.

\item For Edge Swap, assume $(u,u'), (w,w') \in E$ are two edges such as $F_u(G,\bs)=F_w(G,\bs)$ and $\deg^+_u(G) = \deg^+_w(G)$.
Since $u$ and $w$ are not sinks, from the definition of $F$ this means that also $PR^a_u(G,\bs)=PR^a_w(G,\bs)$.
Let $(G',\bs') = ((V, E - \lBrace (u,u'), (w,w') \rBrace \sqcup \lBrace (u,w'), (w,u') \rBrace), \bs)$.
Fix $v \in V$.
Note that $v$ is a sink in $(G,\bs)$ if and only if it is a sink in $(G',\bs')$.
Since PageRank satisfies Edge Swap, equation~\eqref{eq:independence:ed:2} implies $F^a_v(G,\bs) = F^a_v(G',\bs')$.

\item For Node Redirect, assume $u,w \in V$ are out-twins.
Let $(G',\bs') = R_{u \rightarrow w}(G,\bs)$.
Fix $v \in V \setminus \{u,w\}$.
Note that $v$ is a sink in $(G,\bs)$ if and only if it is a sink in $(G',\bs')$.
Since PageRank satisfies Node Redirect, equation~\eqref{eq:independence:ed:2} implies $F_v(G,\bs) = F_v(G',\bs')$.
Analogously, we get that
\[ F^a_u(G,\bs) + F^a_w(G,\bs) = 2 \big( PR^a_u (G,\bs) - \bs(u) + PR^a_w (G,\bs) - \bs(w) \big) = 2 \cdot PR^a_w(G',\bs') - \bs(v) = F^a_w(G',\bs') \]
if $u$ and $w$ are sinks and if not, then:
\[ F^a_u(G,\bs) + F^a_w(G,\bs) = PR^a_u (G,\bs) + PR^a_w (G,\bs) = PR^a_w(G',\bs') = F^a_w(G',\bs'). \]

\item For Baseline, assume $v$ is isolated. 
Since $v$ is as sink and PageRank satisfies Baseline, we get that $F^a_v(G,\bs) = 2 \cdot PR^a_v(G,\bs) - \bs(v) = 2 \cdot \bs(v) - \bs(v) = \bs(v)$.
\end{itemize}

Finally, consider Edge Deletion.
Consider graph $(G,\bs) = ((\{u,v,w\},\lBrace (u,v), (v,w) \rBrace), [1,0,0])$. 
In $(G,\bs)$ node $v$ is not a sink, hence $F^a_v(G,\bs) = PR^a_v(G,\bs) = a$.
Note that $v$ is not its own successor in $G$.
Now, if we delete edge $(v,w)$ we will obtain a graph $(G',\bs') = ((\{u,v,w\}, \lBrace (u,v) \rBrace), [1,0,0])$.
Here, $v$ is a sink, hence $F_v(G',\bs') = 2 \cdot PR^a_v(G,\bs) = 2 \cdot a$.
Thus, Edge Deletion is not satisfied.
\end{proof}


\begin{lemma}\label{lemma:independence:oh}
A centrality measure $F$ defined for every graph $(G,\bs)$ and node $v$ and an arbitrary $a \in (0,1)$ as follows: 
\begin{equation}\label{eq:independence:oh}
F^a_v(G,\bs) = a \cdot \left( \sum_{(u,v) \in \Gamma^-_v(G)} \frac{F^a_u(G,\bs)}{\deg_u^+(G) + 1} \right) + \bs(v)
\end{equation}
satisfies Node Deletion, Edge Deletion, Edge Swap, Node Redirect and Baseline, but does not satisfy Edge~Multiplication.
\end{lemma}
\begin{proof}
Fix $a \in (0,1)$.
Let us start by showing that the centrality $F^a$ is well defined, i.e., that the recursive equation~\eqref{eq:independence:oh} has a unique solution.
To this end, consider a graph operation $f$ that for every graph $G=(V,E)$ with node weights $\bs$ adds a new node $t$ with a zero weight and adds one edge from every node in $V$ to $t$, i.e.:
$$f(G,\bs)=((V \cup \{t\}, E \sqcup \lBrace (v,t) : v \in V \rBrace), \bs_f)$$
where $\bs_f(t)=0$ and $\bs_f(v) = \bs(v)$ for every $v \in V$.
Observe that for every graph $(G, \bs)$ and node $v$ PageRank recursive equation~\eqref{eq:pr:main} for graph $f(G,\bs)$ and node $v$ is the same as equation~\eqref{eq:independence:oh}.
Hence, we get that centrality $F^a$ in graph $(G,\bs)$ is equal to PageRank in graph $f(G,\bs)$: $F^a(G,\bs) = PR^a(f(G,\bs))$.

Consider a second graph operation $h$.
Observe that if for some graph $(G,\bs)$ and node $v$ we have $f(h(G,\bs)) = h(f(G,\bs))$ and $PR^a_v(f(G,\bs)) = PR^a_v(h(f(G,\bs)))$, then 
\begin{equation}\label{eq:independence:oh:2}
F^a_v(G,\bs) = PR^a_v(f(G,\bs)) = PR^a_v(h(f(G,\bs))) = PR^a_v(f(h(G,\bs))) = F^a_v(h(G,\bs)).
\end{equation}

Now, we are ready for the axiomatic analysis of centrality $F^a$.
We will consider each axiom separately. 
Fix an arbitrary graph $G=(V,E)$ and node weights $\bs$.

\begin{itemize}
\item For Node Deletion, assume $u$ is an isolated node. 
Let us define a graph operation $h$ that for an arbitrary graph $(G',\bs') = ((V',E'),\bs')$ deletes node $u$ and all its edges: $h(G',\bs') = ((V' \setminus \{u\}, E' - \Gamma^+_u(G')-\Gamma^-_u(G')), \bs'')$ where $\bs''(v') = \bs'(v')$ for every $v' \in V'$. 
Fix $v \in V \setminus \{u\}$.
\begin{itemize}
\item[--] Clearly, operations $f$ and $h$ are commutative, so $f(h(G,\bs)) = h(f(G,\bs))$.
\item[--] Consider PageRank of $v$ in graph $f(G,\bs)$.
Note that in $f(G,\bs)$ node $u$ is no longer isolated, but has one outgoing edge to node $t$.
Nevertheless, since PageRank satisfies Edge Deletion, we know that removing this edge does not affect PageRank of $v$.
If edge $(u,t)$ is removed, then $u$ is isolated, so since PageRank satisfies Node Deletion, this node can also be removed without affecting PageRank of $v$.
Hence, we get that $PR^a_v(f(G,\bs)) = PR^a_v(h(f(G,\bs)))$.
\end{itemize}
As a result, equation~\eqref{eq:independence:oh:2} implies $F^a_v(G,\bs) = F^a_v(h(G,\bs))$.

\item For Edge Deletion, consider an arbitrary edge $(u,w) \in E$. 
Let us define a graph operation $h$ that for an arbitrary graph $(G',\bs') = ((V',E'), \bs')$ deletes edge $(u,w)$ from the graph: $h(G',\bs') = ((V', E' - \lBrace (u,w) \rBrace), \bs')$. 
Fix $v \not \in S_u(G)$.
\begin{itemize}
\item[--] Clearly, operations $f$ and $h$ are commutative, so $f(h(G,\bs)) = h(f(G,\bs))$.
\item[--] Note that $v \not \in S_u(G)$ implies $v \not \in S_u(f(G))$. 
Hence, since PageRank satisfies Edge Deletion, we get that $PR^a_v(f(G,\bs)) = PR^a_v(h(f(G,\bs)))$.
\end{itemize}
As a result, equation~\eqref{eq:independence:oh:2} implies $F^a_v(G,\bs) = F^a_v(h(G,\bs))$.

\item For Edge Swap, assume $(u,u'), (w,w') \in E$ are two edges such as $F^a_u(G,\bs)=F^a_w(G,\bs)$ and $\deg^+_u(G) = \deg^+_w(G)$.
Let us define a graph operation $h$ that for an arbitrary graph $(G',\bs') = ((V',E'), \bs')$ deletes edges $(u,u'), (w,w')$ and adds edges $(u,w'), (w,u')$: $h(G',\bs') = ((V', E' - \lBrace (u,u'), (w,w') \rBrace \sqcup \lBrace (u,w'), (w,u') \rBrace), \bs')$.
Fix $v \in V$.
\begin{itemize}
\item[--] Clearly, operations $f$ and $h$ are commutative, so $f(h(G,\bs)) = h(f(G,\bs))$.
\item[--] From the definition of $F^a$ we get that $F^a_u(G,\bs)=F^a_w(G,\bs)$ implies $PR^a_u(f(G,\bs)) = PR^a_w(f(G,\bs))$.
Also, $\deg^+_u(f(G)) = \deg^+_w(f(G)) = \deg^+_u(G) + 1$.
Hence, since PageRank satisfies Edge Swap, we have $PR^a_v(f(G,\bs)) = PR^a_v(h(f(G,\bs)))$.
\end{itemize}
As a result, equation~\eqref{eq:independence:oh:2} implies $F^a_v(G,\bs) = F^a_v(h(G,\bs))$.

\item For Node Redirect, assume $u,w \in V$ are out-twins. 
Let us define a graph operation $h$ as redirecting $u$ into $w$, i.e., $h = R_{u \rightarrow w}$.
Take an arbitrary $v \in V \setminus \{u,w\}$.
\begin{itemize}
\item[--] Clearly, operations $f$ and $h$ are commutative, so $f(h(G,\bs)) = h(f(G,\bs))$.
\item[--] Note that the fact that $u,w$ are out-twins in $(G,\bs)$ implies that they are also out-twins in $f(G,\bs)$.
Hence, since PageRank satisfies Node Redirect, we have $PR^a_v(f(G,\bs)) = PR^a_v(h(f(G,\bs)))$.
\end{itemize}
As a result, equation~\eqref{eq:independence:oh:2} implies $F^a_v(G,\bs) = F^a_v(h(G,\bs))$.
Analogously, we get that
\begin{multline*}
    F^a_u(G,\bs) + F^a_w(G,\bs) = PR^a_u(f(G,\bs)) + PR^a_w(f(G,\bs)) =  \\
    PR^a_w(h(f(G,\bs))) = PR^a_w(f(h(G,\bs))) = F^a_w(h(G,\bs)).
\end{multline*}

\item For Baseline, assume $v$ is isolated. From PageRank recursive equation~\eqref{eq:pr:main} we get that $F^a_v(G,\bs) = \bs(v)$.
\end{itemize}

Finally, consider Edge Multiplication. 
Consider graph $(G,\bs) = ((\{u,v\},\lBrace (u,v) \rBrace), [1,0])$.
We have that $F^a_u(G,\bs) = 1$ and $F^a_v(G,\bs) = 1/2 \cdot a$.
Now, if we add an additional copy of outgoing edges of node $u$ we will obtain a graph $(G',\bs') = ((\{u,v\},\lBrace (u,v), (u,v) \rBrace), [1,0])$.
Here, we have that $F^a_u(G,\bs) = 1$ and $F^a_v(G,\bs) = 2/3 \cdot a$.
Thus, Edge Multiplication is not satisfied.
\end{proof}


\begin{lemma}\label{lemma:independence:es}
A centrality measure $F$ defined for every graph $(G,\bs)$ and node $v$ as follows: 
\[
F_v(G,\bs) = \sum_{(u,v) \in \Gamma^-_v(G)} \frac{\bs(u)}{\deg_u^+(G)} + \bs(v)
\]
satisfies Node Deletion, Edge Deletion, Edge Multiplication, Node Redirect and Baseline, but does not satisfy Edge~Swap.
\end{lemma}
\begin{proof}
Note that the centrality $F$ of node $v$ depends solely on direct predecessors of $v$, proportion of their edges going to $v$, and weights of the predecessors and node $v$.
Specifically, we have that:
\begin{equation}\label{eq:independence:es:2}
F_v(G,\bs) = \sum_{u \in P^1_v(G)} \frac{\#_{(u,v)}(G)}{\deg_u^+(G)} \cdot \bs(u) + \bs(v)
\end{equation}

We will consider each axiom separately. 
Fix an arbitrary graph $G = (V,E)$ and node weights $\bs$.

\begin{itemize}
\item For Node Deletion, assume $u$ is an isolated node. 
Since $u$ is isolated, it is not a direct predecessor of any node and we get that $F_v(G,\bs) = F_v((V \setminus \{u\},E),\bs)$ for every $v \in V \setminus \{u\}$.

\item For Edge Deletion, consider arbitrary edge $(u,w) \in E$. 
For every node $v \not \in S_u(G)$ node $u$ is not a direct predecessor and we have $F_v(G,\bs) = F_v((V, E - \lBrace (u,w) \rBrace), \bs)$.

\item For Edge Multiplication, consider node $u \in V$ and $k \in \mathbb{N}$. 
Let $(G',\bs') = ((V,E \sqcup k \cdot \Gamma_u^+(G), \bs)$.
Observe that we have $\#_{(u,v)}(G)/\deg_u^+(G) = \#_{(u,v)}(G')/\deg_u^+(G')$ for every $u,v \in V$.
Also, for every $v \in V$, the weight and the set of direct predecessors in both graphs is the same: $\bs(v) = \bs'(v)$ and $P_v^1(G) = P_v^1(G')$.
Hence, from equation~\eqref{eq:independence:es:2} we get that $F_v(G,\bs) = F_v(G',\bs')$ for every $v \in V$.

\item For Node Redirect, assume $u,w \in V$ are out-twins. 
Let $(G',\bs') = R_{u \rightarrow w}(G,\bs)$.
Note that $\deg_u^+(G) = \deg_w^+(G)$ and $\#_{(u,v)}(G) = \#_{(w,v)}(G)$ for every $v \in V$.
Moreover, out-degrees of all nodes in $(G,\bs)$ are the same as in $(G',\bs')$.
Hence, for every $v \in V \setminus \{u,w\}$ we have:
\[
F_v(G,\bs) = \sum_{(s,v) \in \Gamma^-_u(G), s \not \in \{u,w\}} \frac{\bs(s)}{\deg_s^+(G)} + \frac{\#_{(w,v)}(G)}{\deg_w^+(G)} \cdot (\bs(u)+\bs(w)) + \bs(v) = F_v(G',\bs').
\]
For node $w$, we have that $\left(\#_{(u,u)}(G) + \#_{(u,w)}(G)\right) = \left(\#_{(w,u)}(G) + \#_{(w,w)}(G)\right) = \#_{(w,w)}(G')$. 
So we get that:
\begin{multline*}
    F_u(G,\bs) + F_w(G,\bs) = \\
    \sum_{(s,v) \in \Gamma^-_u(G) \sqcup \Gamma^-_w(G), s \not \in \{u,w\}} \frac{\bs(s)}{\deg_s^+(G)} + \frac{\#_{(u,u)}(G) + \#_{(u,w)}(G)}{\deg_w^+(G)} (\bs(u) + \bs(w)) = F_w(G',\bs').
\end{multline*}

\item For Baseline, assume $v$ is isolated. Since $v$ has no direct predecessors we have $F_v(G,\bs) = \bs(v)$.

\end{itemize}

Finally, consider Edge Swap.
Consider graph $(G,\bs) = ((\{u,v,w\},\lBrace (u,v),(v,w) \rBrace), [1,0,0])$.
We have that $F_u(G,\bs) = 1$, $F_v(G,\bs) = 1$ and $F_w(G,\bs) = 0$.
Observe that nodes $u$ and $v$ both have one outgoing edge and equal centralities.
Now, if we replace edges $(u,v), (v,w)$ with edges $(u,w), (v,v)$ we will obtain a graph $(G',\bs') = ((\{u,v,w\},\lBrace (u,w),(v,v) \rBrace), [1,0,0])$.
Here, we have that $F_w(G',\bs') = 1$.
Thus, Edge Swap is not satisfied.
\end{proof}


\begin{lemma}
\label{lemma:independence:ts}
A centrality measure $F$ defined for every graph $(G,\bs)$ and node $v$ as follows: 
\[
F_v(G,\bs) = \sum_{(u,v) \in \Gamma^-_v(G)} \frac{1}{\deg_u^+(G)} + \bs(v)
\]
satisfies Node Deletion, Edge Deletion, Edge Multiplication, Edge Swap and Baseline, but does not satisfy Node~Redirect.
\end{lemma}
\begin{proof}
Note that the centrality $F$ of node $v$ depends solely on its weight, its direct predecessors and the proportion of their edges going to $v$.
Specifically, we have that:
\begin{equation}\label{eq:independence:ts:2}
F_v(G,\bs) = \sum_{u \in P^1_v(G)} \frac{\#_{(u,v)}(G)}{\deg_u^+(G)} + \bs(v)
\end{equation}

We will consider each axiom separately. 
Fix an arbitrary graph $G = (V,E)$ and node weights $\bs$.

\begin{itemize}
\item For Node Deletion, assume $u$ is an isolated node. 
Since $u$ is isolated, it is not a direct predecessor of any node and we get that $F_v(G,\bs) = F_v((V \setminus \{u\},E),\bs)$ for every $v \in V \setminus \{u\}$.

\item For Edge Deletion, consider arbitrary edge $(u,w) \in E$. 
For every $v \not \in S_u(G)$ node $u$ is not a direct predecessor and we have $F_v(G,\bs) = F_v((V, E - \lBrace (u,w) \rBrace), \bs)$.

\item For Edge Multiplication, consider node $u \in V$ and $k \in \mathbb{N}$. 
Let $(G',\bs') = ((V,E \sqcup k \cdot \Gamma_u^+(G), \bs)$.
Observe that we have $\#_{(u,v)}(G')/\deg_u^+(G') = \#_{(u,v)}(G)/\deg_u^+(G)$ for every $u,v \in V$.
Also, for every $v \in V$, the weight and the set of direct predecessors in both graphs is the same: $\bs(v) = \bs'(v)$ and $P_v^1(G) = P_v^1(G')$.
Hence, from equation~\eqref{eq:independence:ts:2} we get that $F_v(G,\bs) = F_v(G',\bs')$ for every $v \in V$.

\item For Edge Swap, assume $(u,u'), (w,w') \in E$ are two edges such as $F_u(G,\bs)=F_w(G,\bs)$ and $\deg^+_u(G) = \deg^+_w(G)$.
Let $(G',\bs') = ((V, E - \lBrace (u,u'), (w,w') \rBrace \sqcup \lBrace (u,w'), (w,u') \rBrace), \bs)$.
Note that the out-degree and the weight of every node in $(G,\bs)$ is the same as in $(G',\bs')$: $\deg_v^+(G)=\deg_v^+(G')$ and $\bs(v) = \bs'(v)$ for every $v \in V$.
Hence, for every node $v$ with the same set of incoming edges in both graphs, i.e., $v \in V \setminus \{u',w'\}$, we get that $F_v(G,\bs) = F_v(G',\bs')$.
Now, observe that $\deg^+_w(G) = \deg^+_u(G')$. Hence, $1/\deg^+_u(G) = 1/\deg^+_w(G')$ which implies that
\[
F_{u'}(G,\bs) = \sum_{(s,u') \in \Gamma_{u'}^-(G') - \lBrace (w,u') \rBrace} \frac{1}{\deg_s^+(G')} + \frac{1}{\deg_w^+(G')} + \bs'(v) = F_{u'}(G', \bs').
\]
For $w'$ we get analogically $F_{w'}(G,\bs) = F_{w'}(G',\bs')$.

\item For Baseline, assume $v$ is isolated. Since $v$ has no direct predecessors we have $F_v(G,\bs) = \bs(v)$.

\end{itemize}

Finally, consider Node Redirect.
Consider graph $(G,\bs) = ((\{u,v,w\}, \lBrace (u,v),(w,v) \rBrace), [0,0,0])$.
We have $F_v(G,\bs) = 2$.
Note that $u$ and $w$ are out-twins.
Now, if we redirect node $u$ into $w$ we will get graph $(G',\bs') = ((\{v,w\},\lBrace (w,v) \rBrace), [0,0])$.
Here, we have that $F_v(G',\bs') = 1$.
Thus, Node Redirect is not satisfied.
\end{proof}


\begin{lemma}\label{lemma:independence:dn}
A centrality measure $F$ defined for every graph $(G,\bs)$ and node $v$ and an arbitrary $a \in (0,1)$ as follows: 
\[ F^a_v(G,\bs) = 2 \cdot PR^a_v(G,\bs) \]
satisfies Node Deletion, Edge Deletion, Edge Multiplication, Edge Swap and Node Redirect, but does not satisfy Baseline.
\end{lemma}
\begin{proof}
Fix $a \in (0,1)$.
Observe that for any two graphs $(G,\bs)$, $(G',\bs')$ and a node $v$, if we know that $PR_v^a(G,\bs) = PR_v^a(G',\bs')$, then we have 
\begin{equation}\label{eq:independence:dn:2}
F^a_v(G,\bs) = 2 \cdot PR^a(G,\bs) = 2 \cdot PR^a_v(G',\bs') = F^a_v(G',\bs').
\end{equation}

We will consider each axiom separately. 
Fix an arbitrary graph $G=(V,E)$ and node weights $\bs$.

\begin{itemize}
\item For Node Deletion, assume $u \in V$ is an isolated node. 
Let $(G',\bs') = ((V \setminus \{u\}, E), \bs)$.
Since PageRank satisfies Node Deletion, equation~\eqref{eq:independence:dn:2} implies $F^a_v(G,\bs) = F^a_v(G',\bs')$ for every $v \in V \setminus \{u\}$.

\item For Edge Deletion, consider arbitrary edge $(u,w) \in E$.
Let $(G',\bs') = ((V, E - \lBrace (u,w) \rBrace), \bs)$.
Since PageRank satisfies Edge Deletion, equation~\eqref{eq:independence:dn:2} implies $F^a_v(G,\bs) = F^a_v(G',\bs')$ for every $v \in V \setminus S_u(G)$.

\item For Edge Multiplication, consider $u \in V$ and $k \in \mathbb{N}$.
Let $(G',\bs') = ((V, E \sqcup k \cdot \Gamma_u^+(G)), \bs)$.
Since PageRank satisfies Edge Multiplication, equation~\eqref{eq:independence:dn:2} implies $F^a_v(G,\bs) = F^a_v(G',\bs')$ for every $v \in V$.

\item For Edge Swap, assume $(u,u'), (w,w') \in E$ are two edges such as $F^a_u(G,\bs)=F^a_w(G,\bs)$ and $\deg^+_u(G) = \deg^+_w(G)$.
From the definition of $F^a$ this means that also $PR^a_u(G,\bs)=PR^a_w(G,\bs)$.
Let $(G',\bs') = ((V, E - \lBrace (u,u'), (w,w') \rBrace \sqcup \lBrace (u,w'), (w,u') \rBrace), \bs)$.
Since PageRank satisfies Edge Swap, equation~\eqref{eq:independence:dn:2} implies $F^a_v(G,\bs) = F^a_v(G',\bs')$ for every $v \in V$.

\item For Node Redirect, assume $u,w \in V$ are out-twins.
Let $(G',\bs') = R_{u \rightarrow w}(G,\bs)$.
Since PageRank satisfies Node Redirect, equation~\eqref{eq:independence:dn:2} implies $F^a_v(G,\bs) = F^a_v(G',\bs')$ for every $v \in V \setminus \{u,w\}$.
Analogously, we get that
\[ F^a_u(G,\bs) +  F^a_w(G,\bs) = 2 \cdot (PR^a_u(G,\bs) + PR^a_w(G,\bs)) = 2 \cdot PR^a_w(R_{u \rightarrow w}(G,\bs)) = F^a_w(R_{u \rightarrow w}(G,\bs)). \]
\end{itemize}

Finally, consider Baseline. Observe that $F^a_v((\{v\},\emptyset), [1]) = 2$. Hence, the axiom is not satisfied.
\end{proof}


\section[Other centrality measures (Table 1)]{Other centrality measures (Table~\ref{table:axioms})}\label{section:appendix:other_centralities}
In this appendix, we present proofs for satisfaction of our axioms by several centralities from the literature.
The summary of these results is discussed in Section~\ref{section:comparison} and presented in Table~\ref{table:axioms}.

We begin by extending the definition of centrality measures to include measures defined only for a subset of all graphs.
For a class of graphs $\mathcal{G}$, a \emph{centrality measure $F$ defined on} $\mathcal{G}$ is a function that for every node $v$ in a graph with node weights $(G,\bs)$ with $G \in \mathcal{G}$ assigns a non-negative real value, denoted by $F_v(G,\bs)$.
Now, PageRank, degree centrality, beta measure, decay centrality and betweenness centrality are defined on all graphs.
Eigenvector centrality, Seeley index and closeness centrality are defined on strongly connected graphs; we will denote this class of graphs by $\mathcal{SCG}$.
Katz centrality with the decay factor $a \in (0,1)$ is defined on graphs with $\lambda(G) < 1/a$; we will denote this class of graphs by $\mathcal{KG}^a$.
Note that every acyclic graph $G$ belongs to $\mathcal{KG}^a$ for every $a \in (0,1)$ since $\lambda(G) = 0$.

To consider axioms from Section~\ref{section:axioms} for centrality measures defined on a (not complete) class of graphs we need to restrict definitions of axioms to this class.
This can be done by adding an additional condition that every graph considered in the axiom belongs to the considered class of graphs.
In this way, we obtain a weaker version of a general axiom.
We present the formal definitions in the subsequent sections.

\subsection{Node Deletion}

We begin by formally defining the axiom restricted to a class of graphs $\mathcal{G}$.

\begin{quote}\textit{\textbf{Node Deletion on $\mathcal{G}$:}
For every graph $G = (V,E) \in \mathcal{G}$, node weights $\bs$ and isolated node $u \in V$, if $(V \setminus \{u\}, E) \in \mathcal{G}$, then it holds that
$$ F_v(G,\bs) = F_v \big( \big( V \setminus \{u\}, E \big), \bs \big) \quad \mbox{for every } v \in V \setminus \{u\}.$$
}\end{quote}

\begin{proposition}
Node Deletion is satisfied by degree centrality, beta measure, decay centrality, betweenness centrality, Katz centrality (on $\mathcal{KG}^a$ for every $a \in (0,1)$), Bonacich centrality (on $\mathcal{KG}^a$ for every $a \in (0,1)$), eigenvector centrality (on $\mathcal{SCG}$), Seeley index (on $\mathcal{SCG}$), and closeness centrality (on $\mathcal{SCG}$).
\end{proposition}
\begin{proof}
Consider an arbitrary graph $G=(V,E)$, node weights $\bs$, and an isolated node $u \in V$. Let $(G',\bs')=((V \setminus \{u\},E),\bs)$.

For degree centrality, beta measure, decay centrality and betweenness centrality, observe that the centrality of any node $v$ depends solely on a connected component it is in.
Since $u$ is isolated, it is not in the same component as any $v \in V \setminus \{u\}$.
Hence, we get $D_v(G,\bs) = D_v(G',\bs')$, $\beta_v(G,\bs) = \beta_v(G',\bs')$, $Y^a_v(G,\bs) = Y^a_v(G',\bs')$ as well as $B_v(G,\bs) = B_v(G',\bs')$ for every $v \in V \setminus \{u\}$.

For Katz centrality we have to prove that $K^a_v(G,\bs)=K^a_v(G',\bs')$ for every $v \in V \setminus \{u\}$.
Let us define $x_v = K^a_v(G,b)$ for every $v \in V \setminus \{u\}$.
From Katz recursive equation~\eqref{eq:c:katz} for graph $(G,\bs)$ we get that $x_v = a \cdot \sum_{(w,v) \in \Gamma^-_v(G)} x_w + \bs(v)$ for every $v \in V \setminus \{u\}$.
Now, observe that removing an isolated node $u$ from the graph, does not affect the set of incoming edges for any node.
Hence, $\Gamma_v^-(G') = \Gamma_v^-(G)$ which implies $x_v = a \cdot \sum_{(w,v) \in \Gamma^-_v(G')} x_w + \bs'(v)$ for every $v \in V \setminus \{u\}$. 
As a result, we get that values $(x_v)_{v \in V \setminus \{u\}}$ satisfy the system of Katz recursive equations~\eqref{eq:c:katz} for graph $(G',\bs')$.
Thus, $K^a_v(G,\bs) = K^a_v(G',\bs')$ for every $v \in V \setminus \{u\}$, which means that the axiom is satisfied.

For Bonacich centrality, from the fact that Katz centrality satisfies Node Deletion, for every node $v \in V \setminus \{u\}$ we get that:
$$BK^a_v(G,\bs) = \big( K^a_v(G,\bs) - \bs(v) \big) / a = \big( K^a_v(G',\bs') -\bs'(v) \big) / a = BK^a_v(G',\bs').$$

For eigenvector centrality, Seeley index and closeness centrality,
note that in the class of strongly connected graphs $\mathcal{SCG}$ there is no graph with an isolated node. 
Hence, Node Deletion on $\mathcal{SCG}$ is trivially satisfied by every centrality measure.
\end{proof}

\subsection{Edge Deletion}

Again, let us begin by formally defining the axiom restricted to a class of graphs $\mathcal{G}$.

\begin{quote}\textit{\textbf{Edge Deletion on $\mathcal{G}$:}
For every graph $G = (V,E) \in \mathcal{G}$, node weights $\bs$ and edge $(u,w) \in E$, if $(V, E - \lBrace (u,w) \rBrace) \in \mathcal{G}$, then it holds that
$$ F_v(G,\bs) = F_v \big( \big(V, E - \lBrace (u,w) \rBrace \big), \bs \big) \quad \mbox{for every } v \in V \setminus S_u(G).$$
}\end{quote}

\begin{proposition}
Edge Deletion is satisfied by degree centrality, beta measure, decay centrality, Katz centrality (on $\mathcal{KG}^a$ for every $a \in (0,1)$), eigenvector centrality (on $\mathcal{SCG}$), Seeley index (on $\mathcal{SCG}$), and closeness centrality (on $\mathcal{SCG}$).
Edge Deletion is not satisfies by betweenness centrality.
\end{proposition}
\begin{proof}
Consider an arbitrary graph $G=(V,E)$, node weights $\bs$, and an arbitrary edge $(u,w) \in E$.
Let $(G',\bs') = ((V, E - \lBrace (u,w) \rBrace), \bs)$.

For degree centrality, beta measure and decay centrality, observe that the centrality of any node $v$ depends solely on its predecessors and their edges.
Since $u$ is not a predecessor of $v$, we get that $D_v(G,\bs) = D_v(G',\bs')$, $\beta_v(G,\bs) = \beta_v(G',\bs')$ and $Y^a_v(G,\bs) = Y^a_v(G',\bs')$ for any $v \in V \setminus S_u(G)$.

For Katz centrality we have to prove that $K^a_v(G,\bs)=K^a_v(G',\bs')$ for every $v \in V \setminus S_u(G)$.
To this end, let us consider an auxilliary graph $(G'',\bs'')$ obtained from $(G,\bs)$ by removing all successors of $u$ (i.e., a subgraph of $G$ induced by the nodes in $V \setminus S_u(G)$):
$$G'' = \left( V \setminus S_u(G), E - \bigsqcup_{v \in S_u(G)} \Gamma^\pm_v(G) \right)$$
and $\bs''(v) = \bs(v)$ for every $v \in V \setminus S_u(G)$.
We will show that for every $v \in V \setminus S_u(G)$ we have $K^a_v(G,\bs)=K^a_v(G'',\bs'')$.
To this end, let us define $x_v = K^a_v(G,b)$ for every $v \in V \setminus S_u(G)$.
From Katz recursive equation~\eqref{eq:c:katz} for graph $(G,\bs)$ we get that $x_v = a \cdot \sum_{(w,v) \in \Gamma^-_v(G)} x_w + \bs(v)$ for every $v \in V \setminus S_u(G)$.
Now, observe that since $u$ is not a predecessor of any node in $V \setminus S_u(G)$, removing edge $(u,w)$ from the graph, does not affect the set of incoming edges of any node in this set.
Hence, $\Gamma_v^-(G'') = \Gamma_v^-(G)$ which implies $x_v = a \cdot \sum_{(w,v) \in \Gamma^-_v(G'')} x_w + \bs''(v)$ for every $v \in V \setminus S_u(G)$. 
As a result, we get that values $(x_v)_{v \in V \setminus \{u\}}$ satisfy the system of Katz recursive equations~\eqref{eq:c:katz} for graph $(G'',\bs'')$.
Thus, $K^a_v(G,\bs) = K^a_v(G'',\bs'')$ for every $v \in V \setminus S_u(G)$.
Analogously, we prove that $K^a_v(G',\bs') = K^a_v(G'',\bs'')$ for every $v \in V \setminus S_u(G)$.
Hence, we get that $K^a_v(G,\bs) = K^a_v(G',\bs')$ for every $v \in V \setminus S_u(G)$ which proves that Katz centrality satisfies Edge Deletion.

For Bonacich centrality, from the fact that Katz centrality satisfies Edge Deletion, for every node $v \in V$ we get that:
$$BK^a_v(G,\bs) = \big( K^a_v(G,\bs) - \bs(v) \big) / a = \big( K^a_v(G',\bs') -\bs'(v) \big) / a = BK^a_v(G',\bs').$$

For eigenvector centrality, Seeley index and closeness centrality, note that for every graph in the class of strongly connected graphs, i.e., $G =(V,E) \in \mathcal{SCG}$ and node $u \in V$, it holds that every node in a graph is a successor of $u$, i.e., $V \setminus S_u(G) = \emptyset$.
Hence, Edge Deletion on $\mathcal{SCG}$ is trivially satisfied by every centrality measure.

Finally, for betweenness centrality, consider the following graphs:
\begin{align*}
(G,\bs) &= \big( \big( \{u,v,w\},\lBrace (u,v),(v,w)\rBrace \big),[1,1,1] \big), \\
(G',\bs') &= \big( \big( \{u,v,w\},\lBrace (u,v) \rBrace \big),[1,1,1] \big).
\end{align*}
Graph $(G',\bs')$ is obtained from $(G,\bs)$ by deleting edge $(v,w)$.
Clearly, node $v$ is not its own successor in $G$: $v \not \in S_v(G)$.
However, $B_v(G,\bs) = 1$ and $B_v(G',\bs') = 0$. Hence, betweenness centrality does not satisfy Edge Deletion.
\end{proof}

\subsection{Edge Multiplication}

Let us formally define the axiom restricted to a class of graphs $\mathcal{G}$.

\begin{quote}\textit{\textbf{Edge Multiplication on $\mathcal{G}$:}
For every graph $G = (V,E) \in \mathcal{G}$, node weights $\bs$, node $u \in V$ and $k \in \mathbb{N}$, if $(V, E \sqcup k \cdot \Gamma^+_u(G)) \in \mathcal{G}$, then it holds that
$$F_v(G,\bs) = F_v \big( \big(V, E \sqcup k \cdot \Gamma^+_u(G) \big), \bs \big) \quad \mbox{for every } v \in V.$$
}\end{quote}

\begin{proposition}
Edge Multiplication is satisfied by beta measure, decay centrality, Seeley index (on $\mathcal{SCG}$) and closeness centrality (on $\mathcal{SCG}$).
Edge Multiplication is not satisfied by degree centrality, betweenness centrality, eigenvector centrality and Katz centrality.
\end{proposition}
\begin{proof}
Consider an arbitrary graph $G=(V,E)$, node weights $\bs$, and an edge $(u,w) \in E$.
Let $(G',\bs') = ((V, E \sqcup k \cdot \Gamma^+_u(G)), \bs)$.
Note that for any two nodes $u,v \in V$ the same fraction of outgoing edges of $u$ goes to $v$ in $G'$ as in $G$:
\begin{equation}
\label{eq:other:centralities:em}
	\frac{\#_{(u,v)}(G)}{\deg^+_u(G)} = \frac{\#_{(u,v)}(G')}{\deg^+_u(G')}
	\quad \mbox{for every } u,v \in V.
\end{equation}

For beta measure, observe that the centrality of any node $v \in V$ depends solely on the direct predecessors of $v$ and the fraction of their outgoing edges that goes to $v$.
Hence, from equation~\eqref{eq:other:centralities:em} we get that $\beta_v(G,\bs) = \beta_v(G',\bs')$ for every $v \in V$.

For closeness centrality and decay centrality, observe that the centrality of any node $v \in V$ depends solely on the distances to $v$ from other nodes.
Since the distance between any two nodes is the same in $G$ as in $G'$, we get that $C_v(G,\bs) = C_v(G',\bs')$ and $Y^a_v(G,\bs) = Y^a_v(G',\bs')$ for every $v \in V$.

For Seeley index, we have to prove that $SI_v(G,\bs) = SI_v(G',\bs')$ for every $v \in V$.
Let us define $x_v = SI_v(G,b)$ for every $v \in V$.
From Seeley index recursive equation~\eqref{eq:c:katz_prestige} for graph $(G,\bs)$ we get that:
\[ x_v = a \cdot \sum_{(w,v) \in \Gamma^-_v(G)} \frac{x_w}{\deg^+_w(G)} + \bs(v) = \sum_{u \in P^1_v(G)} \frac{\#_{(u,v)}(G)}{\deg_u^+(G)} x_w + \bs(v)
\quad \mbox{for every $v \in V$}. \]
By equation~\eqref{eq:other:centralities:em} this means that:
\[ x_v = \sum_{u \in P^1_v(G')} \frac{\#_{(u,v)}(G')}{\deg_u^+(G')} x_w + \bs'(v) = a \cdot \sum_{(w,v) \in \Gamma^-_v(G')} \frac{x_w}{\deg^+_w(G')} + \bs'(v). \]
As a result, we get that values $(x_v)_{v \in V \setminus \{u\}}$ satisfy the system of Seeley index recursive equations~\eqref{eq:c:katz} for graph $(G',\bs')$.
Thus, $SI_v(G,\bs) = SI_v(G',\bs')$ for every $v \in V$, which means that the axiom is satisfied.

For degree centrality, betweenness centrality, Katz centrality, and Bonacich centrality, consider the following graphs:
\begin{align*}
	(G,\bs) &= \big( \big( \{u,v,v',w\},\lBrace (u,v),(u,v'),(v,w),(v',w) \rBrace \big),[0,1,0,0] \big),\\
	(G',\bs') &= \big( \big( \{u,v,v',w\},\lBrace (u,v),(u,v'),(v,w),(v,w),(v',w) \rBrace \big),[0,1,0,0] \big).
\end{align*}
Graph $(G',\bs')$ is obtained from $(G,\bs)$ by multiplying outgoing edges of node $v$.
Since graphs are acyclic, $(G,\bs),(G',\bs') \in \mathcal{KG}^a$ for every $a \in (0,1)$.
Moreover, $D_w(G,\bs) = 2 \neq 3 = D_w(G',\bs')$, $B_v(G,\bs) = 1/2 \neq 2/3 = B_v(G',\bs')$,  $K^a_w(G,\bs) = a \neq 2a = K^a_w(G',\bs')$, and $BK^a_w(G,\bs) = 1 \neq 2 = BK^a_w(G',\bs')$.
Hence, all four centrality measures do not satisfy Edge Multiplication.

Finally, for eigenvector centrality consider the following graphs:
\begin{align*}
(G,\bs) &= \big( \big( \{u,v\},\lBrace (u,v),(v,u)\rBrace \big),[1,1] \big),\\
(G',\bs') &= \big( \big( \{u,v\},\lBrace (u,v), (u,v), (u,v), (u,v),(v,u)\rBrace \big),[1,1] \big).
\end{align*}
Graph $(G',\bs')$ is obtained from $(G,\bs)$ by multiplying outgoing edges of node $u$.
Observe that $(G,\bs),(G',\bs') \in \mathcal{SCG}$.
Moreover, $\lambda(G)= 1$, $\lambda(G')=1/2$ which implies $EV_u(G,\bs)=EV_v(G,\bs)=1/2$, $EV_u(G',\bs')=1/3$ and $EV_v(G',\bs')= 2/3$.
Hence, eigenvector centrality does not satisfy Edge Multiplication.
\end{proof}

\subsection{Edge Swap}

Once again, let us  start with the formal definition of the axiom restricted to a class of graphs $\mathcal{G}$.

\begin{quote}\textit{\textbf{Edge Swap on $\mathcal{G}$:}
For every graph $G =  (V,E) \in \mathcal{G}$, node weights $\bs$, and edges $(u, u'),  (w, w')  \in  E$ such that $F_u(G,\bs)=F_w(G,\bs)$ and $\deg^+_u(G)=\deg^+_w(G)$, let $E' = E - \lBrace (u,u'), (w,w') \rBrace \sqcup \lBrace (u,w'), (w,u') \rBrace$. If $(V, E') \in \mathcal{G}$, then it holds that
$$ F_v(G,\bs) = F_v \big( \big(V, E' \big), \bs \big) \quad \mbox{for every } v \in V.$$
}\end{quote}

\begin{proposition}
Edge Swap is satisfied by degree centrality, beta measure, eigenvector centrality (on $\mathcal{SCG}$), Seeley index (on $\mathcal{SCG}$) and Katz centrality (on $\mathcal{KG}^a$ for every $a \in (0,1)$).
Edge Swap is not satisfied by decay centrality, betweenness centrality and closeness centrality.
\end{proposition}
\begin{proof}
Consider arbitrary graph $G=(V,E)$, node weights $\bs$, and edges $(u,u'),(w,w') \in E$.
Let $(G',\bs') = ((V,E - \lBrace (u,u'),(w,w') \rBrace  \sqcup \lBrace (u,w'),(w,u') \rBrace), \bs)$.

For degree centrality, observe that the number of incoming edges of every node is the same in both $G$ and $G'$.
Hence, we get
$D_v(G,\bs) = |\Gamma^-_v(G)| = |\Gamma^-_v(G')| = D_v(G',\bs')$
for every $v \in V$.

For beta measure, the centrality depends only on the number of outgoing edges of the predecessors.
Note that the number of outgoing edges of every node is the same in $G$ and $G'$. Hence, every node with the same set of predecessors, i.e., other than $u'$ and $w'$, has the same beta measure in both graphs: $\beta_v(G,\bs) = \beta_v(G',\bs')$ for every $v \in V \setminus \{u',w'\}$.
For $u'$ observe that $\deg_u^+(G)=\deg_w^+(G')$, hence
\begin{multline*}
	\beta_{u'}(G,\bs) =
	\sum_{(s,u') \in \Gamma^-_{u'}(G) -\lBrace (u,u') \rBrace}
		\frac{1}{\deg^+_s(G)}
	+ \frac{1}{\deg^+_u(G)} = \\
	\sum_{(s,u') \in \Gamma^-_{u'}(G') -\lBrace (w,u') \rBrace}
		\frac{1}{\deg^+_s(G')}
	+ \frac{1}{\deg^+_w(G')} =
	\beta_{u'}(G',\bs').
\end{multline*}
Analogously, $\beta_{w'}(G,\bs) = \beta_{w'}(G',\bs')$.

For eigenvector centrality, Katz centrality and Seeley index we will show that the solutions to the system of recursive equations for graph $(G,\bs)$, i.e., $(x_v)_{v \in V}$, satisfies also the system of recursive equations for graph $(G',\bs')$.
This will prove that the centralities in both graphs are equal.
Observe that for every $v \in V \setminus \{u',w'\}$, the recursive equations for $(G,\bs)$ and $(G',\bs')$ are identical, hence $(x_v)_{v \in V}$ satisfy them.
Therefore, it suffices to consider equations for $u'$ and $w'$.
Moreover, the case for $w'$ is symmetrical to $u'$, hence in what follows we will focus only on node $u'$.
Observe that from the construction of graph $G'$ we have
\begin{equation}
\label{eq:other:centralities:es:1}
\Gamma^-_{u'}(G) - \lBrace (u,u') \rBrace = \Gamma^-_{u'}(G') - \lBrace (w,u') \rBrace, \quad 
\deg^+_{u}(G) = \deg^+_{w}(G') \quad \mbox{and} \quad
x_u = x_w.
\end{equation}
Now, let us consider each of these three centrality measures separately:
\begin{itemize}
\item For eigenvector centrality we have that $x_v = EV_v(G,\bs)$ for every $v \in V$.
From eigenvector centrality recursive equation~\eqref{eq:c:eigenvector} for node $u'$ in graph $(G,\bs)$, we get
\[
x_{u'} = 1/\lambda(G) (\textstyle\sum_{(s,u') \in \Gamma^-_{u'}(G) - \lBrace (u,u') \rBrace} x_s + x_u).
\]
Observe that from~\eqref{eq:other:centralities:es:1}, this equation is equivalent to the recursive equation for $u'$ in graph $(G',\bs')$:
$x_{u'} = 1/\lambda(G) (\sum_{(s,u') \in \Gamma^-_{u'}(G') - \lBrace (w,u') \rBrace} x_s + x_w)$.
Therefore, we get that $(x_v)_{v \in V}$ satisfy eigenvector centrality recursive equation for graph $(G',\bs')$ and $\lambda(G)$ which implies that $\lambda(G) = \lambda(G')$~\cite{Kitti:2016} and $EV_v(G',\bs') = x_v = EV_v(G,\bs)$ for every $v \in V$.
\item For Katz centrality we take $x_v = K^a_v(G,\bs)$ for every $v \in V$.
From Katz centrality recursive equation~\eqref{eq:c:katz} for node $u'$ and graph $(G,\bs)$, we get
$x_{u'} = a \cdot \sum_{(s,u') \in \Gamma^-_{u'}(G) - \lBrace (u,u') \rBrace} x_s + a \cdot x_u + \bs(u')$.
From~\eqref{eq:other:centralities:es:1}, this equation is equivalent to
$x_{u'} = a \cdot \sum_{(s,u') \in \Gamma^-_{u'}(G') - \lBrace (w,u') \rBrace} x_s + a \cdot x_w + \bs(u')$.
Therefore, we get that $(x_v)_{v \in V}$ satisfy Katz centrality recursive equation for graph $(G',\bs')$. 
This implies $K^a_v(G',\bs') = x_v = K^a_v(G,\bs)$ for every $v \in V$.
\item Finally, for Seeley index consider $x_v = SI_v(G,\bs)$ for every $v \in V$.
From Seeley index recursive equation~\eqref{eq:c:katz_prestige} for node $u'$ and graph $(G,\bs)$, we get
$x_{u'} = \sum_{(s,u') \in \Gamma^-_{u'}(G) - \lBrace (u,u') \rBrace} x_s/\deg^+_s(G) + x_u/\deg^+_u(G) + \bs(u')$.
From~\eqref{eq:other:centralities:es:1}, this equation is equivalent to
\[
    x_{u'} = \textstyle\sum_{(s,u') \in \Gamma^-_{u'}(G') - \lBrace (w,u') \rBrace} x_s/\deg^+_s(G') + x_w/\deg^+_w(G') + \bs'(u').
\]
As a result, we obtain that $(x_v)_{v \in V}$ satisfy Seeley index recursive equation for graph $(G',\bs')$.
Hence, $SI_v(G,\bs) = SI_v(G',\bs')$ for every $v \in V$.
\end{itemize}

For Bonacich centrality, consider the following graphs:
\begin{align*}
(G,\bs) &= \big(\big( \{u, u', w, w'\}, \lBrace (u, u'), (w,w') \rBrace \big), [1,0,0,0] \big),\\
(G',\bs') &= \big(\big( \{u, u', w, w'\}, \lBrace (u, w'), (w,u') \rBrace \big), [1,0,0,0] \big).
\end{align*}
Graph $(G',\bs')$ is obtained from $(G,\bs)$ by swapping the ends of both edges.
Both graphs are acyclic, hence $(G,\bs),(G',\bs') \in \mathcal{KG}^a$ for every $a \in (0,1)$.
Also, $\deg^+_u(G)=1=\deg^+_w(G)$ and $BK^a_u(G,\bs)=0=BK^a_w(G,\bs)$.
However, $BK^a_{u'}(G,\bs)=1 \neq 0 = BK^a_{u'}(G',\bs')$.
Thus, Bonacich centrality does not satisfy Edge Swap.

Finally, for closeness centrality, decay centrality and betweenness centrality consider the following graphs:
\begin{align*}
(G,\bs) &= \big( \big( \{u, v, w\}, \lBrace (u,v), (u,v), (v,w), (v,w), (w,u), (w,u) \rBrace \big), [1,1,1] \big), \\
(G',\bs') &= \big( \big( \{u, v, w\}, \lBrace (u,u), (u,v), (v,w), (v,w), (w,u), (w,v) \rBrace \big), [1,1,1] \big).
\end{align*}
Graph $(G',\bs')$ is obtained from $(G,\bs)$ by swapping the ends of edges $(u,v)$ and $(w,u)$.
Observe that both graphs belong to $\mathcal{SCG}$ class.
Moreover, $\deg^+_u(G) = 2 = \deg^+_w(G)$ as well as $C_u(G,\bs) = 1/3 = C_w(G,\bs)$, $Y^a_u(G,\bs) = a+a^2 = Y^a_w(G,\bs)$, and $B_u(G,\bs) = 1 = B_w(G,\bs)$.
However, $C_v(G,\bs) = 1/3 \neq 1/2 = C_v(G',\bs')$, $Y^a_v(G,\bs) = a+a^2 \neq 2a = Y^a_v(G',\bs')$, and $B_u(G,\bs) = 1 \neq 0 = B_u(G',\bs')$.
Therefore, closeness centrality, decay centrality and betweenness centrality do not satisfy Edge Swap.
\end{proof}

\subsection{Node Redirect}
Let us begin with the formulation of the axiom restricted to the class of graphs $\mathcal{G}$.

\begin{quote}\textit{\textbf{Node Redirect on $\mathcal{G}$:}
For every graph $G = (V,E) \in \mathcal{G}$, node weights $\bs$ and out-twins $u,w \in V$, if $R_{u \rightarrow w}(G) \in \mathcal{G}$, then it holds that
$$F_v(G,\bs) = F_v \big(R_{u \rightarrow w}(G, \bs) \big) \quad \mbox{for every } v \in V \setminus \{u,w\}$$
and
$F_u(G,\bs) + F_w(G,\bs) = F_w\big(R_{u \rightarrow w}(G, \bs) \big)$.
}\end{quote}

\begin{proposition}
Node Redirect is satisfied by eigenvector centrality (on $\mathcal{SCG}$), Seeley index (on $\mathcal{SCG}$), Katz centrality (on $\mathcal{KG}^a$ for every $a \in (0,1)$) and Bonacich centrality (on $\mathcal{KG}^a$ for every $a \in (0,1)$).
Node Redirect is not satisfied by degree centrality, beta measure, closeness centrality, decay centrality and betweenness centrality.
\end{proposition}
\begin{proof}

Consider arbitrary graph $G=(V,E)$, node weights $\bs$, and out-twins $u,w \in V$.
Let $(G',\bs')=R_{u \rightarrow w}(G, \bs)$.

First, consider eigenvector centrality, Katz centrality, and Seeley index.
Let $F$ be one of these centrality measures and let us define $(x_v)_{v \in V \setminus \{u\}}$ as follows: $x_v = F_v(G,\bs)$ for every $v \in V \setminus \{u,w\}$ and $x_w = F_u(G,\bs) + F_w(G,\bs)$.
We will show that $(x_v)_{v \in V \setminus \{u\}}$ is a solution to the system of recursive equations that defines $F$ for graph $(G',\bs')$.
This will prove that $F_v(G,\bs) = F_v(G',\bs')$ for every $v \in V \setminus \{u,w\}$ and $F_u(G,\bs) + F_w(G,\bs) = F_w(G',\bs')$ which means that $F$ satisfies Node Redirect.
Observe that since $u$ and $w$ are out-twins, we have
\begin{equation}\label{eq:other:centralities:nr:1}
\#_{(u,v)}(G) = \#_{(w,v)}(G) = \#_{(w,v)}(G')
\end{equation}
for every $v \in V \setminus \{u,w\}$.
Moreover, incoming edges of $u$ in $G$ are redirected to $w$ in $G'$, hence 
\begin{equation}\label{eq:other:centralities:nr:2}
\#_{(u,u)}(G) + \#_{(u,w)}(G) = \#_{(w,u)}(G) + \#_{(w,w)}(G) = \#_{(w,w)}(G').
\end{equation}

Now, let us consider each of these three centrality measures separately:

\begin{itemize}
\item For eigenvector centrality let $x_v = EV_v(G,\bs)$ for every $v \in V \setminus \{u,w\}$ and $x_w = EV_u(G,\bs) + EV_w(G,\bs)$.
Fix node $v \in V \setminus \{u,w\}$.
From eigenvector centrality recursive equation~\eqref{eq:c:eigenvector} for node $v$ in graph $(G,\bs)$ we get
\[
	x_{v} = 1/\lambda(G) \left(
		\sum_{(s,v) \in \Gamma^-_{v}(G) : s \not \in \{u,w\}} x_s + 
		\#_{(u,v)}(G) \cdot EV_u(G,\bs) + 
		\#_{(w,v)}(G) \cdot EV_w(G,\bs)
	\right).
\]
Since incoming edges of $v$ from nodes other than $u$ and $w$ are the same in $G'$ as in $G$, from~\eqref{eq:other:centralities:nr:1} we get the recursive equation for $v$ in graph $(G',\bs')$:
\[
	x_{v} = 1/\lambda(G) \left(
		\sum_{(s,v) \in \Gamma^-_{v}(G') : s \neq w} x_s + \#_{(w,v)}(G') \cdot x_w
	\right).
\]
Now, let us focus on node $w$.
From eigenvector centrality recursive equation for node $w$ in graph $(G,\bs)$ we get
{\small
\[
	x_{w} = 1/\lambda(G) \left(
		\sum_{(s,u) \in \Gamma^-_{u}(G) \sqcup \Gamma^-_{w}(G) : s \not \in \{u,w\}} x_s + 
		\big(\#_{(u,u)}(G) + \#_{(u,w)}(G) \big) \cdot  \big( EV_u(G,\bs) + EV_w(G,\bs) \big)
	\right).
\]
}
Observe that all incoming edges of $u$ and $w$ from other nodes in $G$ are now incoming edges of $w$ in $G'$.
Therefore, from equation~\eqref{eq:other:centralities:nr:2} we obtain the equation for $w$ in graph $(G',\bs')$, i.e.,
\(
	x_{w} = 1/\lambda(G) \left(
		\sum_{(s,w) \in \Gamma^-_{w}(G') : s \neq w} x_s + \#_{(w,w)}(G') \cdot x_w
	\right).
\)
As a result, we get that $(x_v)_{v \in V \setminus \{u\}}$ satisfy all eigenvector centrality recursive equations for graph $(G',\bs')$ and $\lambda(G)$, which means that $\lambda(G) = \lambda(G')$ \cite{Kitti:2016} and $EV_v(G',\bs') = x_v = EV_v(G,\bs)$ for every $v \in V \setminus \{u\}$.

\item For Katz centrality let us take $x_v = K^a_v(G,\bs)$ for every $v \in V \setminus \{u,w\}$ and $x_w = K^a_u(G,\bs) + K^a_w(G,\bs)$.
Fix $v \in V \setminus \{u,w\}$.
From Katz centrality recursive equation~\eqref{eq:c:katz} for node $v$ in graph $(G,\bs)$ we get
\[
	x_{v} = a \left(
		\sum_{(s,v) \in \Gamma^-_{v}(G) : s \not \in \{u,w\}} x_s + 
		\#_{(u,v)}(G) \cdot K^a_u(G,\bs) + 
		\#_{(w,v)}(G) \cdot K^a_w(G,\bs)
	\right) + \bs(v).
\]
Observe that incoming edges of $v$ from nodes other than $u$ and $w$ are the same in $G$ as in $G'$ and also $\bs(v)=\bs'(v)$. Therefore, from equation~\eqref{eq:other:centralities:nr:1} we get the recursive equation for $v$ in $(G',\bs')$, i.e.,
\(
	x_{v} = a \left(
		\sum_{(s,v) \in \Gamma^-_{v}(G') : s \not \in \{w\}} x_s + \#_{(w,v)}(G') \cdot x_w
	\right) + \bs'(v).
\)
Now, consider node $w$.
From Katz centrality recursive equation for node $w$ in graph $(G,\bs)$ we get
{\small
\begin{multline*}
	x_{w} = 
	a \left(
		\sum_{(s,u) \in \Gamma^-_{u}(G) \sqcup \Gamma^-_{w}(G) : s \not \in \{u,w\}} x_s + 
		\big(\#_{(u,u)}(G) + \#_{(u,w)}(G) \big) \!\cdot\!  \big( K^a_u(G,\bs) \! + \! K^a_w(G,\bs) \big) \!
	\right) \! + \bs(u) + \bs(w).
\end{multline*}
}
Observe that incoming edges of $u$ and $w$ from other nodes in $G$ are incoming edges of $w$ in $G'$.
Moreover, $K^a_u(G,\bs) + K^a_w(G,\bs)=x_w$ and $\bs(u)+\bs(w) = \bs'(w)$. 
Hence, from equation~\eqref{eq:other:centralities:nr:2} we get the recursive equation for $w$ in graph $(G',\bs')$:
\[
	x_{w} = a \left(
		\sum_{(s,w) \in \Gamma^-_{w}(G') : s \not \in \{w\}} x_s + \#_{(w,w)}(G') \cdot x_w
	\right) + \bs'(w).
\]
As a result, we get that $(x_v)_{v \in V \setminus \{u\}}$ satisfy all Katz centrality recursive equations for graph $(G',\bs')$, so $K^a_v(G',\bs') = x_v = K^a_v(G,\bs)$ for every $v \in V \setminus \{u\}$.

\item For Seeley index let $x_v = SI_v(G,\bs)$ for every $v \in V \setminus \{u,w\}$ and $x_w = SI_u(G,\bs) + SI_w(G,\bs)$.
First, observe that the number of outgoing edges of every node is the same in $G$ and $G'$.
Hence,
\begin{equation}
\label{eq:other:centralities:nr:3}
	\deg^+_u(G) = \deg^+_w(G) = \deg^+_w(G') \quad \mbox{and} \quad 
	\deg^+_v(G) = \deg^+_v(G') \quad \mbox{for every } v \in V \setminus \{u,w\}.
\end{equation}
Fix $v \in V \setminus \{u,w\}$.
From Seeley index recursive equation~\eqref{eq:c:katz_prestige} for node $v$ in graph $(G,\bs)$ we obtain
\[
	x_{v} = \left(
		\sum_{(s,v) \in \Gamma^-_{v}(G) : s \not \in \{u,w\}} \frac{x_s}{\deg^+_s(G)} + 
		\#_{(u,v)}(G) \cdot \frac{SI_u(G,\bs)}{\deg^+_u(G)} + 
		\#_{(w,v)}(G) \cdot \frac{SI_w(G,\bs)}{\deg^+_w(G)}
	\right).
\]
Observe that incoming edges of $v$ from nodes other than $u$ and $w$ are the same in graph $G$ and in graph $G'$.
Therefore, from equations~\eqref{eq:other:centralities:nr:1} and \eqref{eq:other:centralities:nr:3} we obtain the recursive equation for $v$ in $(G',\bs')$, i.e.,
\(
	x_{v} = \left(
		\sum_{(s,v) \in \Gamma^-_{v}(G') : s \not \in \{w\}} x_s/\deg^+_s(G') + \#_{(w,v)}(G') \cdot x_w/\deg^+_w(G')
	\right)
\)
It remains to consider node $w$.
From Seeley index recursive equation for graph $(G,\bs)$ we have that
{\small
\[
	x_{w} = \left(
		\sum_{(s,u) \in \Gamma^-_{u}(G) \sqcup \Gamma^-_{w}(G) : s \not \in \{u,w\}} \frac{x_s}{\deg^+_s(G)} + 
		\big(\#_{(u,u)}(G) + \#_{(u,w)}(G) \big) \cdot  \left( \frac{SI_u(G,\bs)}{\deg^+_u(G)} + \frac{SI_w(G,\bs)}{\deg^+_w(G)} \right)
	\right).
\]
}
Observe that incoming edges of $u$ and $w$ from other nodes in $G$ are incoming edges of $w$ in $G'$.
Hence, from the fact that $SI_u(G,\bs) + SI_w(G,\bs)=x_w$ and from equations~\eqref{eq:other:centralities:nr:2} and~\eqref{eq:other:centralities:nr:3} we get the recursive equation for $w$ in $(G',\bs')$:
\[
	x_{w} = 1/\lambda(G) \left(
		\sum_{(s,w) \in \Gamma^-_{w}(G') : s \not \in \{w\}} x_s/\deg^+_s(G') + \#_{(w,w)}(G') \cdot x_w/\deg^+_w(G')
	\right).
\]
As a result, we obtain that $(x_v)_{v \in V \setminus \{u\}}$ satisfy all Seeley index recursive equations for graph $(G',\bs')$.
Therefore, $SI_v(G',\bs') = x_v = SI_v(G,\bs)$ for every $v \in V \setminus \{u\}$.
\end{itemize}

For Bonacich centrality, from the fact that Katz centrality satisfies Node Redirect, for every node $v \in V \setminus \{u,w\}$ we get that:
$$BK^a_v(G,\bs) = \big( K^a_v(G,\bs) - \bs(v) \big) / a = \big( K^a_v(G',\bs') -\bs'(v) \big) / a = BK^a_v(G',\bs')$$
and also:
$$BK^a_u(G,\bs) + BK^a_w(G,\bs) = \big( K^a_u(G,\bs) + K^a_w(G,\bs) -\bs(u) - \bs(w) \big) / a = \big( K^a_w(G',\bs') - \bs'(w) \big) / a = BK^a_w(G',\bs').$$

Finally, for degree centrality, beta measure, closeness centrality, decay centrality and betweenness centrality, let us consider the following graphs:
\begin{align*}
	(G,\bs) &= \big( \big(\{ u,v,w \}, \lBrace (u,v), (w,v), (v,u), (v,w) \rBrace \big), [1,1,1] \big),\\
	(G',\bs') &= \big( \big(\{ v, w \}, \lBrace (w,v), (v,w), (v,w) \rBrace \big), [1,2] \big).
\end{align*}
Graph $(G',\bs')$ is obtained from $(G,\bs)$ by redirecting node $u$ into node $v$.
Observe that indeed $u$ and $w$ are out-twins and both graphs belongs to the $\mathcal{SCG}$ class.
However, $D_v(G,\bs) = 2 \neq 1 = D_v(G',\bs')$, $\beta_v(G,\bs) = 2 \neq 1 = \beta_v(G',\bs')$, $C_v(G,\bs) = 1/2 \neq 1 = C_v(G',\bs')$, $Y^a_v(G,\bs) = 2a \neq a = Y^a_v(G',\bs')$, and $B_v(G,\bs) = 2 \neq 0 = B_v(G',\bs')$.
Therefore, Node Redirect is not satisfied by these centrality measures.
\end{proof}

\subsection{Baseline}

As usual, we begin with the formulation of the axiom restricted to the class of graphs $\mathcal{G}$.

\begin{quote}\textit{
\textbf{Baseline on $\mathcal{G}$:}
For every graph $G = (V,E) \in \mathcal{G}$, node weights $\bs$ and an isolated node $v \in V$ it holds that $F_v(G, \bs) = \bs(v)$.
}
\end{quote}

\begin{proposition}
Baseline is satisfied by Katz centrality (on $\mathcal{KG}^a$ for every $a \in (0,1)$), eigenvector centrality (on $\mathcal{SCG}$), Seeley index (on $\mathcal{SCG}$), and closeness centrality (on $\mathcal{SCG}$).
Baseline is not satisfied by degree centrality, beta measure, decay centrality and betweenness centrality.
\end{proposition}
\begin{proof}
For Katz centrality, the thesis is straightforward from Katz centrality recursive equation~\eqref{eq:c:katz}: since $v$ does not have any incoming edges we get that $K^a_v(G,\bs) = \bs(v)$.

For eigenvector centrality, Seeley index and closeness centrality,
note that in the class of strongly connected graphs $\mathcal{SCG}$ there is no graph with an isolated node. 
Hence, Baseline on $\mathcal{SCG}$ is trivially satisfied by every centrality measure.

Finally, for degree centrality, beta measure, decay centrality, betweenness centrality, and Bonacich centrality, consider graph $(G,\bs) = ((\{v\},\emptyset)[1])$.
We see that $D_v(G,\bs) = \beta_v(G,\bs) = Y^a_v(G,\bs) = B_v(G,\bs) = BK_v(G,\bs) = 0 \neq 1 = \bs(v)$.
Hence, Baseline is not satisfied by these centrality measures.
\end{proof}

\bibliography{bibliography}

\end{document}